\documentclass[aps,a4paper,showkeys,nofootinbib,longbibliography,notitlepage,superscriptaddress,twocolumn]{revtex4-2}

\usepackage[utf8]{inputenc}
\usepackage{filecontents}
\usepackage{natbib}
\usepackage{amsmath,amssymb,bm,amsthm}
\usepackage{subfigure}
\usepackage{graphicx}
\usepackage{dcolumn}
\usepackage{bm}
\usepackage[mathlines]{lineno}
\usepackage{booktabs}
\usepackage{color}
\newcounter{one}
\usepackage{url}

\usepackage{textcase}
\usepackage{mathrsfs}

\usepackage{colortbl}
\usepackage{tabularx}
\usepackage{verbatim}
\usepackage{multirow}

\usepackage[T1]{fontenc} 
\usepackage{lmodern}
\usepackage{bbm}
\usepackage[utf8]{inputenc}
\usepackage{amsfonts}
\usepackage{array}

\usepackage{bbm}
\usepackage{enumitem}
\usepackage{umoline}
\usepackage[usenames,svgnames]{xcolor}
\usepackage{natbib}
\usepackage[hyperindex,breaklinks]{hyperref}
\hypersetup{
     colorlinks=true,       		
     linkcolor=Navy,          	
     citecolor=Navy,            
     filecolor=Navy,      		
     urlcolor=Navy,           	
    runcolor=cyan,
 }
\usepackage{graphicx}
\usepackage{amsfonts}
\usepackage{amssymb}
\usepackage{amsmath}
\usepackage{bbm}
\usepackage{enumitem}

\usepackage{xcolor}

\usepackage{dsfont}

\newcommand{\tr}[0]{ {\rm tr}}

\newcommand{\half}[1]{{ \rm h}}

\newcommand{\co}{{\rm c}}

\def\beq{\begin{equation}}
\def\eeq{\end{equation}}
\def\nbeq{\begin{equation*}}
\def\neeq{\end{equation*}}
\def\<{\langle}
\def\>{\rangle}

\def\tr{{\rm tr}}

\newtheorem{theorem}{Theorem}

\newtheorem{lemma}{Lemma}

\newtheorem{definition}{Definition}  
\newtheorem{prop}[lemma]{Proposition} 

\newcommand{\sectionprl}[1]{{\par\it #1.---}}

\newcommand{\br}[1]{\left( #1 \right)}

 \newcommand{\norm}[1]{\left \|  #1 \right \|}

\newcommand{\ave}[1]{\left \langle #1 \right \rangle}
\newcommand{\abs}[1]{\left | #1 \right |}

\usepackage{mathtools}
\def\multiset#1#2{\ensuremath{\left(\kern-.3em\left(\genfrac{}{}{0pt}{}{#1}{#2}\right)\kern-.3em\right)}}

\renewcommand\thefootnote{*\arabic{footnote}}

\begin{document}

\title{Energy diffusion in the long-range interacting spin systems}

\author{Hideaki Nishikawa}\email{srwv.nishikawa@gmail.com}
\affiliation{Department of Physics, Kyoto University, Kyoto 606-8502, Japan}
\affiliation{Analytical Quantum Complexity RIKEN Hakubi Research Team, RIKEN Center for Quantum Computing (RQC), 2-1 Hirosawa, Wako, Saitama 351-0198, Japan}


\author{Keiji Saito}\email{keiji.saitoh@scphys.kyoto-u.ac.jp}\affiliation{Department of Physics, Kyoto University, Kyoto 606-8502, Japan}

\begin{abstract}
We investigate energy diffusion in long-range interacting spin systems, where the interaction decays algebraically as $V(r) \propto r^{-\alpha}$ with the distance $r$ between the sites. We consider prototypical spin systems, the transverse Ising model, and the XYZ model in the $D$-dimensional lattice with a finite exponent $\alpha >D$ which guarantees the thermodynamic extensivity. In one dimension, both normal and anomalous diffusion are observed, where the anomalous diffusion is attributed to anomalous enhancement of the amplitude of the equilibrium current correlation. We prove the power-law clustering property of arbitrary orders of joint cumulants in general dimensions. Applying this theorem to equal-time current correlations, we further prove several theorems leading to the statement that the sufficient condition for normal diffusion in one dimension is $\alpha > 3/2$ regardless of the models. The fluctuating hydrodynamics approach consistently explains L\'{e}vy diffusion for $\alpha < 3/2$, which implies the condition is optimal. In higher dimensions of $D \geq 2$, normal diffusion is indicated as long as $\alpha > D$. 
\end{abstract}

\maketitle

\sectionprl{Introduction} 
Long-range interactions are ubiquitous, exemplified by gravity, the Coulomb interaction, and the dipole interaction \cite{campa2009statistical,campa2014physics}. Here, `long-range' means that the potential follows a power-law form  $V(r)\propto r^{-\alpha}$, where $r$ is the distance between particles. Recent experimental advancements have remarkably enabled precise control over the interaction regime and the power-law exponent $\alpha$. Notable examples include Rydberg atoms, cold atomic, and trapped ionic systems \cite{saffman2010quantum,adams2019rydberg,gallagher2008dipole,bendkowsky2009observation,bernien2017probing,yan2013observation,doi:10.1080/00018730701223200,RevModPhys.80.885,bendkowsky2009observation,bruzewicz2019trapped,britton2012engineered,richerme2014non,neyenhuis2017observation,bendkowsky2009observation,Islam583,zhang2017observation}. These experimental platforms have not only enabled manipulation in information processing but have also served as pivotal testbeds for probing fundamental physics, encompassing thermalization, information propagation, and beyond.

Long-range interactions can induce novel physics that are absent in systems governed solely by short-range interactions such as for thermodynamic properties \cite{campa2009statistical,campa2014physics,lynden1968gravo,HertelThirring1971,mukamel2008statistical,barre2001inequivalence}, equilibrium phase transitions \cite{Dyson1969,PhysRevLett.37.1577,PhysRevLett.29.917,PhysRevB.64.184106,fey2019quantum}, and ground state \cite{kuwahara2019area,koffel2012entanglement,vodola2014kitaev}, which are examples in the static properties. Dynamical properties are also significantly influenced by long-range interactions, as correlations can persist over long distances \cite{torcini1999equilibrium,latora1999superdiffusion,latora1998lyapunov,anteneodo1998breakdown,yamaguchi2004stability}. Recently, impacts of long-range interaction have been intensively studied for the Lieb-Robinson bound \cite{Kuwahara_2016_njp,PhysRevX.10.031010,tran2021optimal,tran2021liebrobinson}, out-of-time-order correlators \cite{PhysRevLett.126.030604}, and fast scrambling \cite{bentsen2019fast}. Most studies on these dynamical properties have focused on spin systems, which are relevant to the aforementioned experimental platforms that are well described by spin-based frameworks \cite{defenu2023long}. As another crucial dynamical property, energy propagation or diffusion should be a crucial subject in the spin systems. However, thus far, systematic studies on the impact of long-range interaction have not been reported. Existing studies for energy diffusion primarily address {\it spring} systems \cite{olivares2016role,bagchi2017energy,iubini2018heat,xiong2024observation,wang2020thermal,tamaki2020energy,andreucci2025thermal}, often as extensions of low-dimensional anomalous energy transport phenomena attributed to long-time tails in the current correlations \cite{benenti2023non,livi2023anomalous,lepri2003thermal,dhar2008heat,spohn2014nonlinear}.

The primary aim of this paper is to understand the general properties of energy diffusion in spin systems with long-range interactions, with a particular focus on classical spin systems as a first step in this direction of study. 
We consider prototypical systems, including the transverse Ising and XYZ models (encompassing the XY model), imposing the condition $\alpha >D$ ($D$: the spatial dimension), where the extensivity in equilibrium thermodynamics is satisfied and useful analytical technique is available.
 Specifically, we address the following questions: Is there a transition of normal and anomalous energy diffusion by changing the exponent $\alpha$ even if thermodynamic extensivity holds at the equilibrium state? If so, what is the universal mechanism and systematic method to analyze the class of energy diffusion? We should note that energy diffusion in the insulators cannot be fully understood in a particle-based picture. 

Energy diffusion is quantitatively connected to thermal conductivity, which is formulated via the Green-Kubo formula for the energy current correlation. Through numerical calculation, we found that anomalous behavior emerges not due to the long-time tail in the current correlation, but with another mechanism, anomalous enhancement of the amplitude of current correlations. With this observation, we focus on the equal-time correlation to consider the amplitude. We derive three theorems (Theorems 1, 2, and 3, below), which remain valid even for quantum systems. These theorems reveal the universal structure of energy diffusion. A critical technique in our analysis is the cumulant power-law clustering theorem (Theorem 1), which states that any joint cumulants of physical quantities supported in distinct regions exhibit clustering properties in the form of a power-law decay. We apply this theorem to analyze the equal-time energy current correlations (Theorem 2), leading to the following statements: the condition $\alpha >\alpha_c =3/2$ is sufficient for normal energy diffusion in one dimension, and this condition is universal in the sense that it holds across all prototypical models. Numerical computations for one-dimensional systems strongly suggest that the condition is optimal, with $\alpha < 3/2$ leading to L\'{e}vy diffusion. L\'{e}vy diffusion is explained through fluctuating hydrodynamics. For the $D$-dimensional lattices, we prove the theorem indicating that the normal diffusion is maintained as long as $\alpha > D$, which satisfies thermodynamic extensivity at equilibrium (Theorem 3). These findings establish a foundation for understanding the general aspects of energy diffusion in long-range interacting spin systems. Furthermore, our analysis provides a powerful and versatile tool for addressing this class of problems.

\begin{figure}[t]
\centering
\hspace*{-0.5cm}\includegraphics[width=0.62\textwidth]{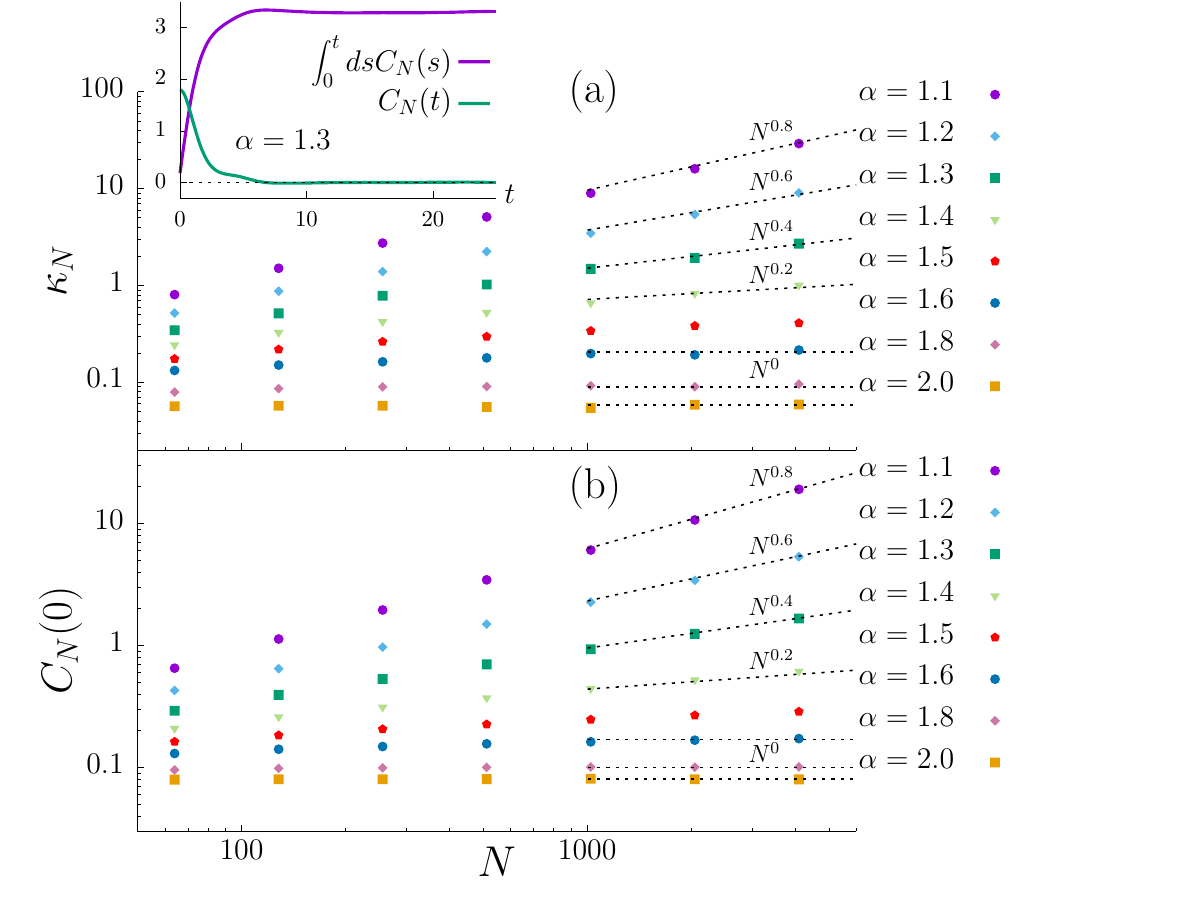}
\caption{The thermal conductivity for the long-range transverse Ising model with the parameters $J=h=1$, and $T=10$ ($k_{\rm B}\equiv 1$). (a): $\kappa_N$ as a function of size $N$. The inset is the time dependence as well as the integration of the total current correlation for $\alpha=1.3$ and $N=1024$. (b): $C_N(0)$ as a function of $N$.}
\label{fig1}
\end{figure}

\sectionprl{Setup and thermal conductivity}
Let us first consider the one-dimensional uniform classical spin Hamiltonian $H=\sum_{n=1}^N \hat{h}_n $, where $\hat{h}_n$ is the $n$th local Hamiltonian and $N$ is the number of sites. We consider the two types of prototypical spin systems for $\hat{h}_n$:
\begin{align}
\begin{split}
\hat{h}_n^{({\rm I})} \!&= \!-J \sum_{r=1}^{N/2} {S_n^z S_{n+r}^z + S_{n-r}^z S_n^z 
 \over 2 r^{\alpha}} + h S_n^x \, , \\
\hat{h}_n^{({\rm XYZ})}\! &= \!
-\sum_{r=1}^{N/2}  {1\over 2 r^{\alpha}} \!\!\sum_{\sigma =x,y,z} J_{\sigma} (S_n^{\sigma} S_{n+r}^{\sigma} + S_{n-r}^{\sigma}S_n^{\sigma}  )
\,  , 
\end{split}\label{ham2}
\end{align}
where $\hat{h}_n^{({\rm I})}$ and $\hat{h}_n^{({\rm XYZ})}$ represent the local Hamiltonians for the long-range transverse Ising and XYZ models, respectively. The continuous variable $S_n^{\sigma}$ denotes the $\sigma$ component of the spin at site $n$ and satisfies $\sum_{\sigma}(S_n^{\sigma} )^2 =1$. Periodic boundary condition is imposed such that $n \pm N \equiv n$. The parameter $\alpha$ represents the power-law exponent of the long-range interaction, with $\alpha > 1$ to guarantee the thermodynamic extensivity. 

When the system shows normal energy diffusion, the energy diffusion constant ${\mathfrak D}$ is given using the thermal conductivity $\kappa$ and the specific heat per unit length $c_{\rm V}$ as ${\mathfrak D} = \kappa / c_{\rm V}$, where $c_{\rm V}$ is finite in general. Thus, determining whether energy diffusion is normal reduces to examining whether the thermal conductivity is finite. Let us consider the thermal conductivity, which is formulated using the Green-Kubo formula with the equilibrium energy current correlation \cite{kubo2012statistical}. The energy current is defined through the continuity equation for the local energy, i.e.,  $\partial_t \hat{h}_n = - {\cal J}_{n+1} + {\cal J}_{n}$. For systems with nearest-neighbor interactions, deriving the current expression is straightforward. However, a more refined expression is necessary for systems with long-range interactions, as follows:
\begin{align}
{\cal J}_n &= \sum_{i \ge n, j\le n-1 \atop |i-j| \le N/2} t_{i \leftarrow j} \, ,~~~  t_{i \leftarrow j} := \{ \hat{h}_i , \hat{h}_j\} \, ,\label{energycurrent}
\end{align}
where $\{..., ...\}$ is the Poisson bracket in the spin dynamics, $\{A,B \}=\sum_{i} \sum_{a,b,c}\varepsilon_{abc}(\partial A/\partial S_i^a)(\partial B/\partial S_i^b) S_i^c$ with the Levi-Civita symbol $\varepsilon_{abc}$. See the Appendix A in the End matter for the details of the spin dynamics and derivation of energy current. Here, $t_{i \leftarrow j}$ is interpreted as the energy transmission from the site $j$ to $i$ across the site $n$. The current ${\cal J}_{n}$ is defined as the collection of the transmissions through the virtual surface between the sites $n-1$ and $n$ (See the schematic inside the Fig.\ref{fig1}(a)) \cite{tamaki2020energy}. The Green-Kubo formula is given by \cite{kubo2012statistical} 
\begin{align} \kappa_{N}\! &= \!{1\over k_{\rm B} T^2}\!\! \int_{0}^{\infty}\!\! \! \! dt\, C_N(t) \, , ~C_N(t)\! = \!\sum_{n=1}^{N} \langle {\cal J}_n(t) {\cal J}_{0} (0)\rangle, 
\end{align} 
where $k_{\rm B}$ is the Boltzmann constant and $\langle ... \rangle$ denotes the equilibrium average with the temperature $T$. The subscript $N$ emphasizes size-dependence. 

Numerical calculations are performed for the transverse Ising model. 
The details of the numerical method are presented in the Appendix B. 
We show the results in Fig.\ref{fig1}. Fig.\ref{fig1}(a) shows the size dependence of $\kappa_N$, with the inset demonstrating the rapid decay of the current correlation function for $\alpha=1.3$, leading to the convergence of the Green-Kubo integral. This convergence is physically expected, as the system is thought to be non-integrable and is a lattice system without continuous translational symmetry (e.g.,~\cite{saito2003strong}). We should note a stark difference from the low-dimensional spring systems which generally show the long-time tail in the current correlation due to the total momentum conservation arising from the continuous translational symmetry \cite{benenti2023non,livi2023anomalous,lepri2003thermal,dhar2008heat,spohn2014nonlinear}. We present additional evidence of the rapid decay of the current correlation for different parameters and models in the Appendix B. The main panel of Fig.\ref{fig1}(a) shows anomalous behavior for $\alpha < 1.5$, while for larger $\alpha$, $\kappa_N$ saturates, indicating normal energy diffusion. These results clearly indicate a transition in energy diffusion.

The observed anomalous behavior is not due to the long-time tail of the current correlation as clearly shown in the inset of Fig.\ref{fig1}(a), but arises from another mechanism, the anomalous enhancement of the amplitude of the correlation function. This is shown from the size-dependence of the equal-time equilibrium correlation $C_N(0)$, shown in Fig.\ref{fig1}(b). Note that $C_N(0)$ and $\kappa_N$ exhibit identical qualitative behavior, suggesting that the equal-time correlation amplitude governs the anomalous energy diffusion in the present systems. 

These observations strongly suggest that {\it the equal-time current correlation categorizes the nature of energy diffusion}. The rest of the paper demonstrates that this perspective is consistently supported, enabling a comprehensive classification of energy diffusion in long-range interacting lattice systems.

\sectionprl{Basic theorems}
To analyze the equal-time current correlations, we prove the cumulant power-law clustering theorem, a generalization of clustering properties to higher cumulants, valid in $D$-dimensional lattices.
{\theorem\label{cpower-main}
Consider $k$-local long-range spin systems on a $D$-dimensional lattice. Let $O_{X_j}$ be a quantity supported on $X_j$ with $X_i \cap X_{j} =\emptyset$ for $i\neq j$. For a finite exponent $\alpha >D$ and $T> T_c$, the the $n$th order joint cumulant $\langle O_{X_1} ...O_{X_n} \rangle_c$ in equilibrium is bounded as
\begin{align}
&\langle O_{X_1} ...O_{X_n} \rangle_c \leq c \, 2^n n! \times \nonumber \\
&\prod_{j=1}^n \| O_{X_j} \| \sum_P \prod_{i=1}^{n-1} {|X_{P_i}| \, |X_{P_{i+1}}| e^{(|X_{P_i}| + |X_{P_{i+1}}|)/k } \over (d_{X_{P_i},X_{P_{i+1}}})^{\alpha} }\, , 
 \end{align}
 where $T_c = 4 e n g_u k$ and $g_u$ is the energy bound per one site. Here, $c$ is a constant and $\|O_X\|$ is the maximum absolute value of $O_X$. $P$ stands for the Hamilton path for a graph with the vertices $X_1, \cdots, X_{n}$, and $X_{P_i}$ is the support at the $i$th step of the Hamilton path $P$. $d_{X,Y}$ is the shortest distance between the supports $X$ and $Y$, and $|X|$ is the number of sites of the support $X$.}

In the Hamilton path, a set $X_j$ is regarded as a vertex, connected with edges. 
$k$-local means that the system contains at most $k$-body interactions. Hence, the present models (\ref{ham2}) are $2$-local Hamiltonian.
Example of the joint cumulant is the second-order: $\langle O_{X_1} O_{X_2} \rangle_c = \langle O_{X_1} O_{X_2} \rangle - \langle O_{X_1} \rangle \langle O_{X_2} \rangle $.
The second-order cumulant in this type of theorem is commonly referred to as the clustering property. Especially for short-range interacting systems, the exponential clustering property is widely recognized as one of the most robust features in nature, characterized by the exponential decay of equilibrium correlations. The present theorem is a generalization of the clustering property to {\it higher-order cumulants in long-range interacting systems}. It states that the $n$th order cumulant among the $n$ regions decays faster than or equal to the power-law form governed by the long-range interaction exponent $\alpha$. The second-order case for long-range interacting systems has been discussed in Ref. \cite{kim2025thermal}. The proof employs the cluster expansion technique, detailed in the Supplementary Material (SM)\cite{Supplement_thermal}. The cluster expansion requires the high temperature regime ($T > T_c$) for convergence, while the clustering properties are physically expected to hold for a wide range of finite temperature regime (See the Appendix C). This theorem is a basis to derive Theorem \ref{upper-main} and \ref{upper-main-highD} below for the equal-time current correlation.

{\theorem\label{upper-main} For the one-dimensional long-range interacting spin systems (\ref{ham2}) in the infinite line, the equal-time equilibrium current correlation of $T > T_c$ in the regime $1<\alpha \leq 2$ is upper-bounded as
\begin{align}
|\langle {\cal J}_{n} {\cal J}_{0} \rangle |< c' \, n^{-(2\alpha -2)} \, , \label{upp-main1}
\end{align} 
where $c'$ is a constant that does not depend on $n$. }

For the regime $\alpha>2$, the transverse Ising and XY cases are bounded by $c' \, n^{2 -2\alpha}$, whereas the XYZ case is bounded by $c' \, n^{-\alpha}$. The energy current is written in terms of $2$-body ($3$-body) spin variables in the transverse Ising (XYZ) case, indicating that the current correlation involves $4$-body ($6$-body) spin variables accordingly. The $n$th-order moment can be systematically decomposed into all possible contributions up to the $n$th-order joint cumulants, as exemplified by the relation $\langle O_{X_1} O_{X_2} \rangle =\langle O_{X_1} O_{X_2}\rangle_{\rm c}+\langle O_{X_1}\rangle_{\rm c} \langle O_{X_2}\rangle_{\rm c}$ for $n=2$.
We then invoke Theorem \ref{cpower-main} to consider the equal-time current correlation. By systematically establishing general diagrammatic rules, we eventually derive the upper bound in Theorem \ref{upper-main}. An outline of the derivation is provided in the Appendix D, with detailed proofs available in the SM \cite{Supplement_thermal}. Although theorem \ref{upper-main} is mathematically derived for high temperature regime, it is physically expected to hold for a broad temperature range, since the theorem \ref{cpower-main} should robustly hold in the physical viewpoint.  Notably, Theorems \ref{cpower-main}, \ref{upper-main}, and \ref{upper-main-highD} below remain valid even for quantum systems. 

Theorem \ref{upper-main} plays a pivotal role in characterizing universal energy diffusion. It leads to the upper bound on the correlation amplitude, $C_N(0) < c' N^{3 - 2 \alpha} + {\rm const}$. This result aligns optimally with the numerical observations in Fig.\ref{fig1}(b). {\it Consequently, it can be asserted that $\alpha > 3/2$ constitutes a sufficient condition for normal diffusion, irrespective of the choice of the models.} Furthermore, as elaborated below, Theorem \ref{upper-main} provides a theoretical framework for anomalous diffusion with fluctuating hydrodynamics.

\sectionprl{Fluctuating hydrodynamics and L\'{e}vy diffusion}
We analyze the anomalous regime using Theorem \ref{upper-main}. We introduce the following fluctuating hydrodynamics for long-range interacting systems with the continuous field representation in the infinite line \cite{Supplement_thermal}:
\begin{align}
{\partial \epsilon (x,t) \over \partial t} &=\! {\partial \over \partial x} 
\Bigl[\! \int\!\! d x' D(x-x') {\partial \over \partial x'} \epsilon (x' ,t) \! + \! \xi (x,t)\Bigr] \, , \label{fht1}
\end{align}
where $\epsilon (x,t)$ is the energy field at the continuous coordinate $x$ per unit length and $\xi (x,t)$ is a noise. The noise average satisfies $\langle\! \langle \xi (x,t) \xi (x', t') \rangle\! \rangle = 2 c_{\rm V} k_{\rm B} T^2 D(x-x') \delta (t-t')$. Fluctuating hydrodynamics can be derived either phenomenologically \cite{landau2,spohn2014nonlinear} or through physical arguments based on Zwanzig's projection operator method \cite{zwanzig,zubarevmorozof,saito2021microscopic} as outlined in the SM \cite{Supplement_thermal}.

The nonlocal diffusion coefficient $D(x - x')$ is crucial to characterize the system. For nearest-neighbor interactions, $D(x - x')$ decays rapidly in general, leading to the approximation $D(x - x') \sim {\mathfrak{D}} \delta(x - x')$. However, in long-range interacting systems, $D(x - x')$ remains long-range. Using the fluctuating hydrodynamics, we discuss the space-time equilibrium correlation function:
\begin{align}
C(x,t) &:= \langle\!\langle \delta \epsilon (x,t) \, \delta \epsilon (x=0, t=0)\rangle\!\rangle \, , \label{stc-fht}
\end{align} 
where $\delta \epsilon (x,t):= \epsilon (x,t) - \bar{\epsilon}_x$, where $\bar{\epsilon}_x$ is the long-time average of the local energy. It is straightforward to obtain the equation of motion for $C(x,t)$ as $\partial_t C (x,t) = \partial_x  \int d x' D(x-x') \partial_{x'} C(x' ,t)$. With the Fourier transform $a(x,t)=(2\pi)^{-1}\int dk a (k,t)e^{i k x}$, we have $\partial_t C (k,t) = -k^2 D(k) C (k,t)$. When the diffusion coefficient scales as $D(k) \sim k^{-\nu}$ for the small wave numbers, consistent asymptotic behavior of $C(x,t)$ should have the following scaling form with a function $f$
\begin{align}
C(x,t) &\sim \left\{ 
\begin{array}{ll}
t^{-1/(2-\nu)} f_{\rm L\'{e}vy}\bigl( {x \over t^{1/(2-\nu)}}\bigr) &~~~ \nu > 0 \, , \\
t^{-1/2} f_{\rm Gauss}\bigl( {x \over t^{1/2}}\bigr) &~~~ \nu < 0 \, . 
\end{array}
\right. \label{Levyscaling}
\end{align}
When nonlocal diffusion coefficient exhibits a power-law decay $D(x-x') \sim |x-x'|^{-\eta}$ for large $|x-x'|$, $\nu$ is given by $\nu=1-\eta$. Hence $\eta > 1 $ shows normal diffusion, while $\eta < 1$ means anomalous (L\'{e}vy) diffusion.

\begin{figure}[h]
\centering
\includegraphics[width=0.55\textwidth]{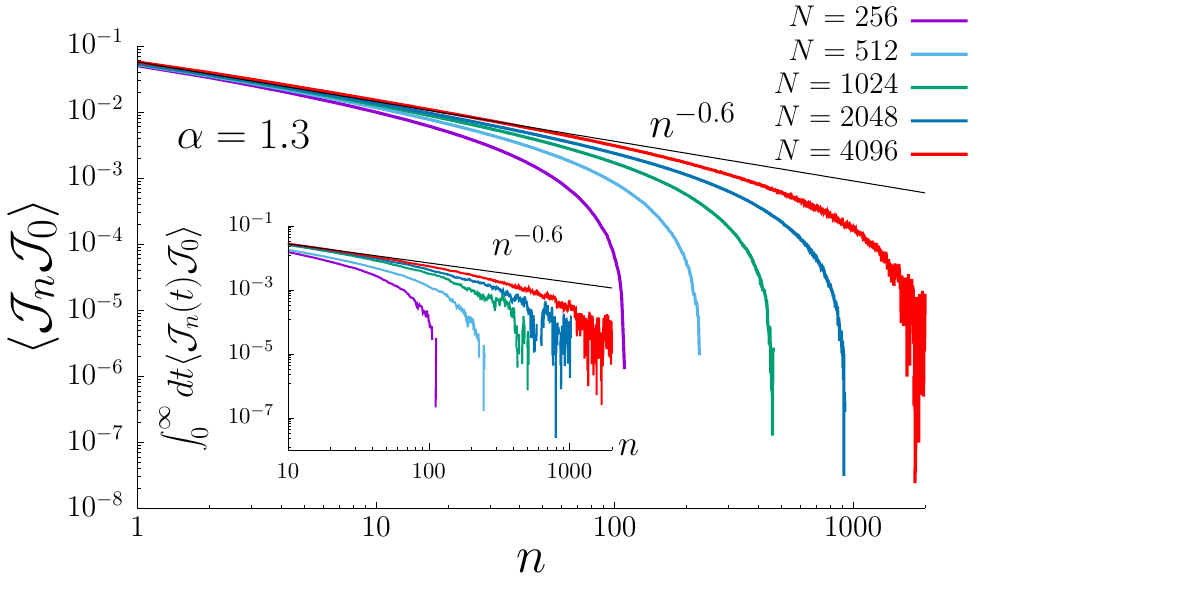}
\caption{The correlation function $\langle {\cal J}_n {\cal J}_0\rangle$ for the transverse Ising model with the parameter $\alpha=1.3$ and $T=10~(k_{\rm B}\equiv 1)$ in the main panel. The inset shows the Green-Kubo like formula. The solid line is the theoretical line (\ref{upp-main1}), i.e., $n^{-(2\alpha - 2)}$.
The equal-time correlation $\langle {\cal J}_n {\cal J}_0\rangle$ explains the power law exponent, $\eta=2\alpha-2$.}
\label{fig2}
\end{figure}

\begin{figure*}[!t]
 	\centering
\includegraphics[width=16.4cm,height=6.75cm]{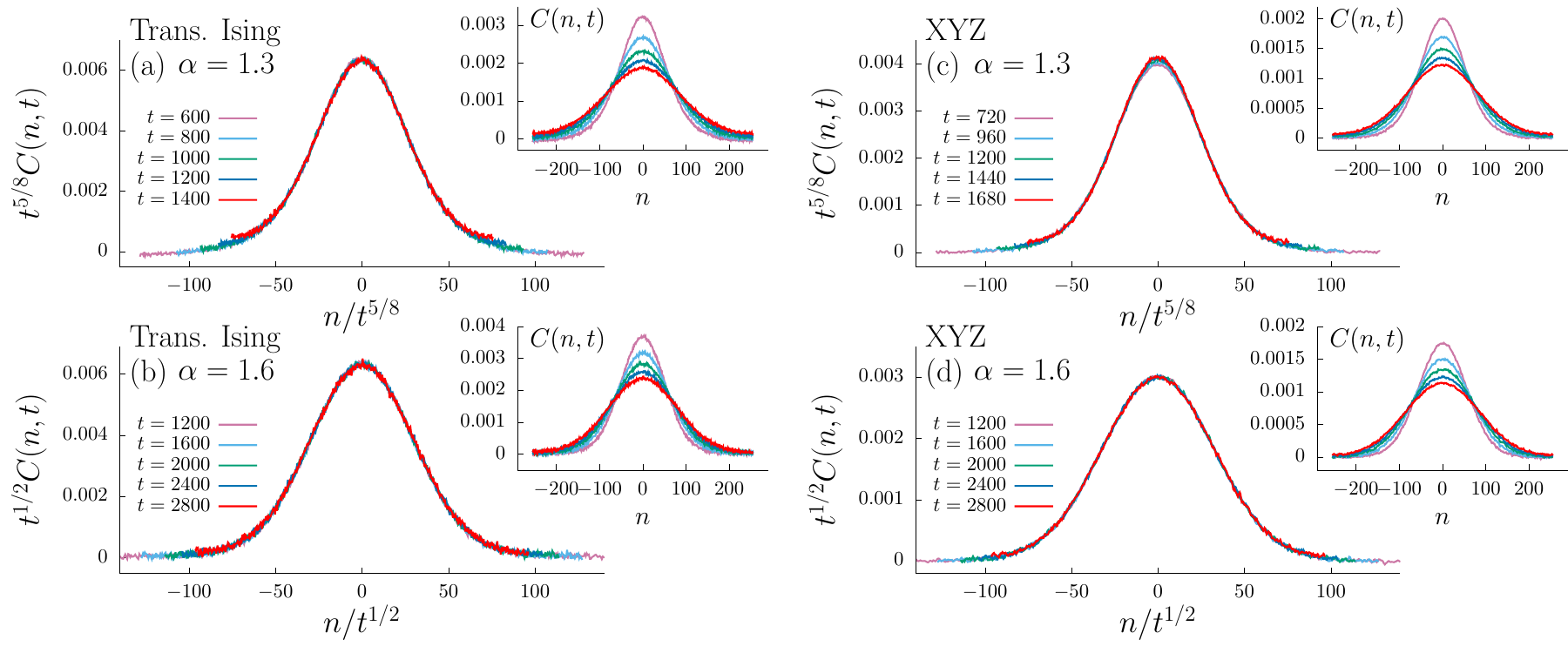}
\caption{Space-time energy correlation function $C(n,t):=\langle \delta \hat{h}_{n} (t) \delta \hat{h}_0\rangle$ for the transverse Ising model with the parameters $J=1$ and $h=1$ ((a) and (b)) and the XYZ model with the parameters $J_x=1,J_y=0.8,J_z=0.5$ ((c) and (d)). $T=10~(k_{\rm B}\equiv 1)$. System size is $N=512$. In (a) and (c), we use the parameter $\alpha=1.3$, while (b) and (d) are for $\alpha=1.6$. It is clear that in both models, $\alpha=1.3$ case obeys the L\'{e}vy scaling, while $\alpha=1.6$ shows the diffusive scaling. }
\label{fig3}
\end{figure*}

We now perform a more microscopic analysis of the nonlocal diffusion coefficient. In the discrete lattice picture, the nonlocal diffusion coefficient is effectively expressed by  
the Green-Kubo-like expression $D(n-n') = (k_{\rm B} T^2 c_{\rm V})^{-1} \int_0^{\infty} dt \langle {\cal J}_n (t) {\cal J}_{n'}\rangle$ \cite{saito2021microscopic,Supplement_thermal}, for which we can argue the amplitude of the equal-time part with the Theorem \ref{upper-main}. In the main panel of Fig.\ref{fig2}, we present the equal-time correlation function $\langle {\cal J}_n {\cal J}_0\rangle$ for the transverse Ising model with the exponent $\alpha=1.3$. In the inset, we also show the results of the Green-Kubo-like expression for the same parameter. We checked that $\langle {\cal J}_n (t) {\cal J}_{0}\rangle$ rapidly decays in time and the integration converges (not shown). The data are displayed on a log-log scale, revealing a power-law behavior as the system size increases. We observe the finite-size effect at the largest distance. Remarkably, the power-law exponent $\eta$ matches the theoretical upper bound in (\ref{upp-main1}), i.e., $\eta=2\alpha-2 =0.6$. The inset has the same power-law exponent, indicating that the equal-time correlation determines the qualitative behavior. We also checked this agreement across various parameter settings (see the SM \cite{Supplement_thermal}), strongly reaffirming the tightness of the upper bound in (\ref{upp-main1}). This observation strongly suggests the scaling in (\ref{Levyscaling}), irrespective of the choice of models, with 
\begin{align}
\nu&=3-2\alpha \, . \label{levyexponent}
\end{align} 

To confirm this scaling, we numerically examined the space-time correlation for the transverse Ising and the XYZ models with various exponents $\alpha$. In the Hamiltonian picture, the space-time correlation is written as $C(n,t):=\langle \delta \hat{h}_{n} (t) \delta \hat{h}_0\rangle$, where $\delta \hat{h}_n (t) := \hat{h}_n (t) - \langle h_{n} \rangle$. We show the results in Fig.\ref{fig3}. Figs. \ref{fig3}(a) and (b) correspond to the transverse Ising model, while (c) and (d) are for the XYZ model. Figs. \ref{fig3}(a) and (c) show the results for $\alpha=1.3$, while (b) and (d) are for $\alpha=1.6$. The results clearly show that for both models, $\alpha=1.3$ case obeys the L\'{e}vy scaling, while $\alpha=1.6$ case shows the diffusive scaling, as suggested by the scaling (\ref{Levyscaling}) with (\ref{levyexponent}). See the SM \cite{Supplement_thermal} for the other parameters.

\sectionprl{Higher dimensions} Theorem \ref{cpower-main} holds even in higher dimensions, and we second analyze the generic $D$-dimensional lattice with $N^D$ sites. Let us consider the energy diffusion that spreads uniformly in concentric patterns, defining the mean square displacement $ S(t):= {\cal N}^{-1} \sum_{i=1}^{D} \sum_{\bm n} n_{i}^2 \langle \delta \epsilon ({\bm n},t) \,  \delta \epsilon ({\bm 0},0)\rangle $, where $\epsilon ({\bm n},t)$ is the local energy of the $D$-dimensional site vector ${\bm n}$ at time $t$, and ${\cal N}= \sum_{\bm n} \langle \delta \epsilon ({\bm n}) \,  \delta \epsilon ( {\bm 0})\rangle$ \cite{basile2009thermal}. If the energy diffusion is normal, this quantity yields the diffusion constant which is formulated as \cite{Supplement_thermal}.
\begin{align}
\begin{split}
\mathfrak{D}&=\lim_{t\to\infty} { S(t) \over 2t} = {1 \over k_{\rm B} T^2 c_{\rm V}}\int_{0}^{\infty} dt \,  C_N^{(D)} (t)\, , \\
 C_N^{(D)} (t) &= {1\over N^{D-1}} \sum_{i=1}^D  \sum_{ n_i =1}^N \langle {\cal J}^{(i)}_{ n_i \mid n_i -1 } (t) {\cal J}^{(i)}_{ 0 \mid -1 } (0) \rangle 
 \, , \label{dform}
\end{split}
\end{align}
where $c_{\rm V}$ is the specific heat per unit volume, and ${\cal J}^{(i)}_{ n_i \mid n_i -1 }$ is the energy current in the $i$th direction, defined as the collection of all energy transmission across the virtual surface between $\bm{n}-\bm{e}_i$ and $\bm{n}$. Here, ${\bm e}_i$ is the unit vector in the $i$th direction. We derive the following theorem:
{\theorem\label{upper-main-highD}
In the $D$-dimensional long-range interacting spin systems with $\alpha > D $ and $D \ge 2$, one-dimensional version of which is depicted in (\ref{ham2}), $C_N^{(D)}(0)$ for $T> T_c$ is finite.}

Standing on the premise that the amplitude of the equal-time current correlation determines the class of energy diffusion, this theorem implies that energy diffusion in the higher dimensions $D\ge 2$ is always normal, as long as we consider $\alpha >D$, which guarantees thermodynamic extensivity, and hence $c_{\rm V}$ is finite. We provide a detailed proof in the SM \cite{Supplement_thermal}.

\sectionprl{Summary and outlook}
We investigate the general classification of energy diffusion in prototypical long-range interacting spin systems described by the local Hamiltonian (\ref{ham2}). Focusing on the regime $\alpha > D$, where the extensive property in equilibrium thermodynamics holds, we establish that equal-time current correlations determine the class of energy diffusion, as the correlation function decays rapidly in time and the integration converges for finite system sizes. 
To analyze the equal-time current correlations, we prove the cumulant power-law clustering theorem (Theorem \ref{cpower-main}), enabling a rigorous evaluation of whether these correlations diverge in the thermodynamic limit (Theorems \ref{upper-main} and \ref{upper-main-highD}). 
Our theorems lead to the physical statement that $\alpha > 3/2$ is the sufficient condition for normal diffusion in one dimension, and in higher dimensions ($D \geq 2$), normal diffusion is indicated as long as $\alpha > D$, irrespective of the choice of models. The numerical calculation indicates that our bound in Theorem \ref{upper-main} is tight and $\alpha < 3/2$ corresponds to L\'{e}vy diffusion, as explained by fluctuating hydrodynamics. 

Our analysis provides a systematic and exact framework for addressing transport phenomena in long-range interacting systems. It may be useful for other diffusive phenomena, including particle and spin transport \cite{bernardin2025macroscopic, kirone2024,schuckert2020nonlocal}. The energy diffusion for the $\alpha <D$ regime is an open problem, since the present analysis is not applicable.
 Theorem \ref{cpower-main} is valid even in quantum systems. It is intriguing to check the same energy diffusive phenomena in the quantum case. To derive the Theorem 1, we have used the property of finite value of one-site energy ($g_u < \infty$ and $\| O_{X_j} \| < \infty$) and hence, the theorem is not directly applicable to the system with unbounded variables, i.e., ($ g_u \to \infty$ and $\| O_{X_j} \| \to \infty$). Hence, energy diffusion in unbounded variable systems, such as Bosonic lattice systems, is an open question.


\section*{acknowledgments}

{~}\\
The authors are grateful Hiroyoshi Nakano for the fruitful discussion and numerical support.
The computation in this study has been done using the facilities of the Supercomputer Center, the Institute for Solid State Physics, the University of Tokyo.
They also thank Tomotaka Kuwahara for insightful discussion. This work was supported by JSPS KAKENHI Grant No. JP23K25796.


{\it Data availability}. - All data and codes that support the findings of this study are available in \cite{nishikawa_data2025} for public access. 

\bibliography{long-range}

\begin{thebibliography}{65}
\bibitem [{ft1()}]{ft1}%
  \BibitemOpen
  \href@noop {} {}We used the fourth-order symplectic time-integrator method
  with the time increment $dt=0.01$ for the time evolution. The thermal average
  is computed using the Monte Calro method. As shown in the inset of Fig.1(a),
  $C_N(t)$ decays rapidly in time, and $\int^t_0 ds C_N(s)$ converges by
  $t=20$. Hence, we take $t=20$ as the upper limit for the time integration of
  $\kappa_N$. The sample size is $10^6$. The same applies to other numerical
  calculations presented in the main text.\BibitemShut {Stop}%
\end{thebibliography}

\section*{End Matter}
\begin{figure*}[t]
    \begin{tabular}{c}
      \begin{minipage}{0.33\hsize}
        \begin{center}
                \includegraphics[width= \textwidth]{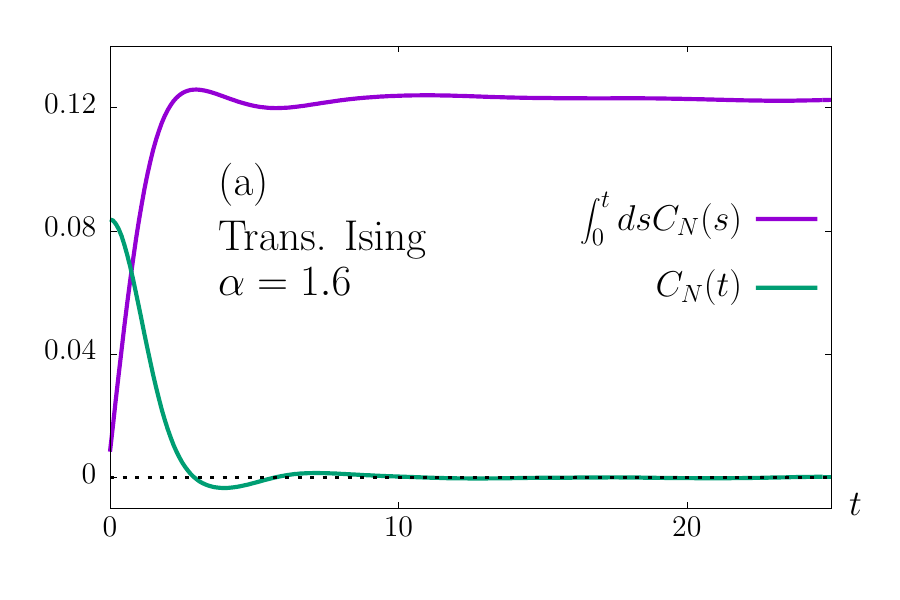}
        \end{center}
      \end{minipage}

      \begin{minipage}{0.33\hsize}
        \begin{center}
                \includegraphics[width= \textwidth]{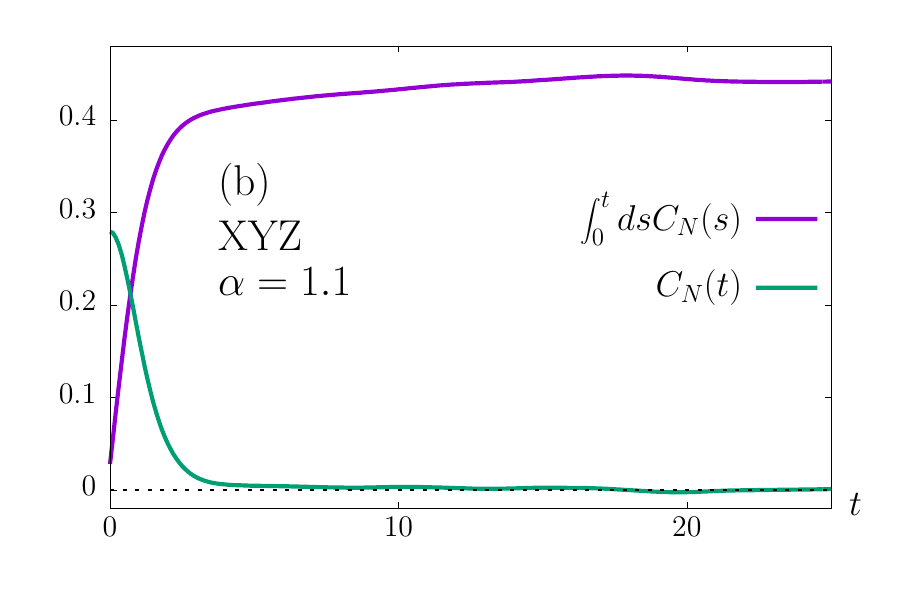}
        \end{center}
      \end{minipage}
      \begin{minipage}{0.33\hsize}
        \begin{center}
                \includegraphics[width= \textwidth]{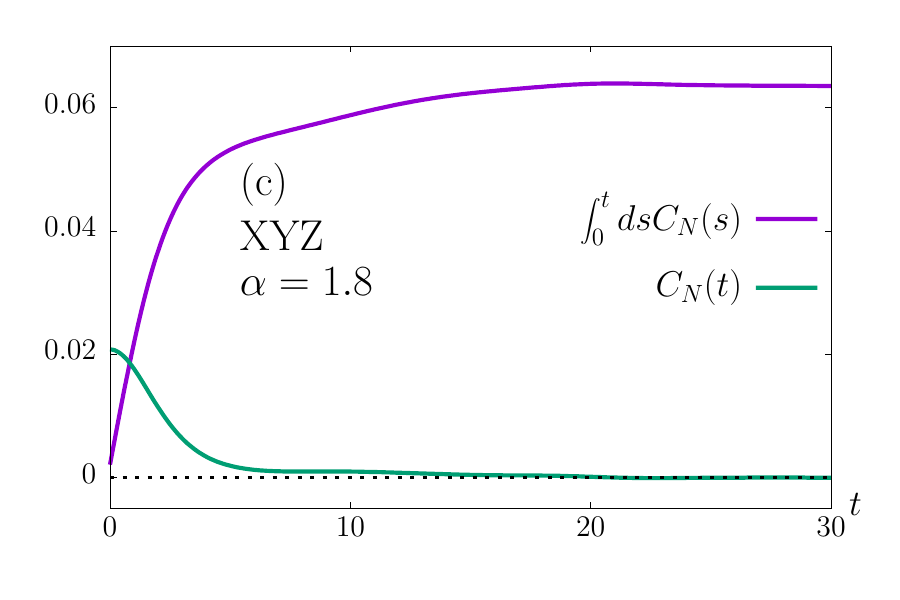}
        \end{center}
      \end{minipage}
    \end{tabular}
\caption{Rapid decay of total current correlation $C_N(t):= \sum_n \ave{ {\cal J}_n (t) {\cal J}_0 (0) }$ and the convergence of Green-Kubo integral $\int^\infty_0 dt C_N(t)$. (a): Transverse Ising case for $\alpha=1.6$ and $J=1, h=-0.5, N=1024$, (b): XYZ case for $\alpha=1.1$ and $J_x=1, J_y=0.8, J_z=0.5, N=256$ and (c): XYZ case for $\alpha=1.8$ and $J_x=1, J_y=0.8, J_z=0.5, N=256$. } \label{gk_decay} 
\end{figure*}

\sectionprl{Appendix A: Spin dynamics and the energy current expression}
The time evolution of the classical spin system is governed by the equation
\begin{align}
{\partial \over \partial t} S_{i}^{\sigma}
 &= \left[ {\bm S}_i \times \left(- \frac{\partial  H}{\partial {\bm S}_i} \right)\right]_{\sigma} =: \{S_{i}^{\sigma} , H \} \, , \label{spindynamics}
\end{align}
where $\{..., ...\}$ is the Poisson bracket in the spin dynamics, $\{A,B \}=\sum_{i} \sum_{a,b,c}\varepsilon_{abc}(\partial A/\partial S_i^a)(\partial B/\partial S_i^b) S_i^c$ with the Levi-Civita symbol $\varepsilon_{abc}$. Any function dependent on the spin variables $f(\{{\bm S}_i \})$ is then expressed with the Poisson bracket:
\begin{align}
{\partial \over \partial t} f(\{{\bm S}_i \}) &= \sum_{i,\sigma} {\partial f \over \partial S_{i}^{\sigma}}  \{S_{i}^{\sigma} , H \} = \{f  , H \} \, .
\end{align} 
Now let us derive the energy current expression (\ref{energycurrent}). The time-derivative of the local energy in the one-dimensional system is given as follows.
\begin{align}
{\partial \over \partial t} \hat{h}_n &=\{\hat{h}_n  , H \} 
= \sum_{j} \{ \hat{h}_n , \hat{h}_j\} \nonumber \\ &
= \sum_{j \le n-1} \{ \hat{h}_n , \hat{h}_j\}
 - \sum_{i \ge n+1} \{ \hat{h}_i , \hat{h}_n\}
\, \nonumber \\ & 
= \sum_{j \le n-1} \{ \hat{h}_n , \hat{h}_j\} + \sum_{i\ge n+1 \atop j \le n-1} \{ \hat{h}_i , \hat{h}_j\} \nonumber \\ 
&- \sum_{i \ge n+1} \{ \hat{h}_i , \hat{h}_n\} - \sum_{i\ge n+1 \atop j \le n-1} \{ \hat{h}_i , \hat{h}_j\}  \nonumber \\
&=\sum_{i \ge n \atop j \le n-1} \{ \hat{h}_i , \hat{h}_j\} -\sum_{i \ge n+1 \atop j \le n} \{ \hat{h}_i , \hat{h}_j\} \, . 
\end{align}
By comparing this equation with the continuity equation $\partial_t \hat{h}_n = {\cal J}_n - {\cal J}_{n+1}$, we identify the energy current expression as in (\ref{energycurrent}).

\sectionprl{Appendix B: Numerical method and rapid decay of the correlation function}
To integrate the equation (\ref{spindynamics}), we employ the fourth-order symplectic integrator method \cite{suzuki1992general} with a time step of $dt = 0.01$. To compute the correlation function, initial spin configurations are sampled using the Monte Carlo method at temperature $T$. Thermal averages are taken over $10^6$ independent initial configurations. As illustrated in the inset of Fig.~1(a), the autocorrelation function $C_N(t)$ exhibits rapid decay in time, and the time integral $\int_0^t ds, C_N(s)$ converges by $t = 20$. Accordingly, we set $t = 20$ as the upper limit of integration in the calculation of $\kappa_N$. The same sample size of $10^6$ is used in all other numerical results presented in the main text. 

In addition to Fig.~1(a), we present supplemental figures to emphasize the rapid temporal decay of the correlation function and the convergence of the Green-Kubo integral for the other parameter sets and models. See Fig.~\ref{gk_decay}.

\sectionprl{Appendix C: Remark on the theorem \ref{cpower-main}}
The theorem \ref{cpower-main} is a generalization of the clustering property which is usually represented for the second order cumulant $n=2$. The functional form of the clustering property for short-range (SR) interacting systems as well as the long-range (LR) interacting systems is given as
\begin{align}
| \langle O_{x_1} O_{x_2} \rangle_c | &< c \| O_{x_1} \| \| O_{x_2} \| e^{- c' d_{x_1 , x_2}} ~~~({\rm SR})   \, ,  \label{sr} \\
| \langle O_{x_1} O_{x_2}  \rangle_c | &< c \| O_{x_1} \| \| O_{x_2} \| {1\over (d_{x_1 , x_2})^{\alpha}}  ~~~({\rm LR})  \, ,  \label{lr}
\end{align}
where $x_1$ and $x_2$ denote single lattice sites, i.e., $|x_1|=|x_2|=1$. The constants $c$ and $c'$ are independent of the observables. Here, $d_{x_1 , x_2}$ is a distance between $x_1$ and $x_2$.  The inequality (\ref{sr}) is the exponential clustering property, which is widely accepted as one of the most robust properties in nature, valid in a broad temperature regime, even down to low temperatures, except at the critical temperature if a phase transition occurs. Likewise, the inequality (\ref{lr}) which applies to long-range interacting systems should robustly hold, as recently discussed in Ref. \cite{kim2025thermal}. 
However, when we derive these inequalities for {\it general} systems using the cluster expansion technique, the convergence of the series expansion requires the high-temperature condition $(T>T_c)$ due to the potential of finite temperature phase transition \cite{kuwahara2020gaussian,PhysRevX.4.031019,kim2025thermal}. 

\begin{figure}[!hb]
\centering
\includegraphics[width=0.4\textwidth,height=3.0cm]{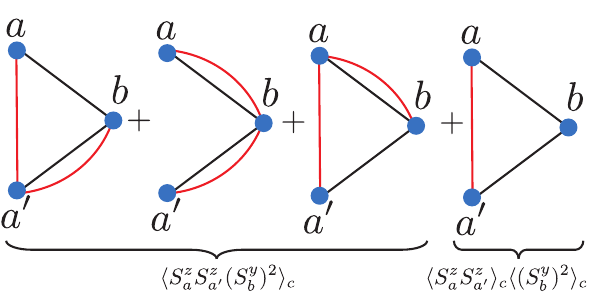}
\caption{The schematic diagrams for (\ref{ising_proof}).}
\label{fig4}
\end{figure}

The cluster expansion technique provides the functional form for the clustering property of the higher-order cumulant in the SM \cite{Supplement_thermal}. As in the same argument as above, this is physically expected to hold across a broad temperature regime, beyond the restriction of high temperatures. The derivation was performed for the following general $k$-local long-range interacting system in the $D$-dimension, where the Hamiltonian is given in the following form
\begin{align}
H = \sum_{Z:|Z|\le k} h_Z \, , 
\end{align}
where $Z$ is a subset indicating the interacting sites, and $h_Z$ is the local Hamiltonian. We assume that the interaction form satisfies the following property
\begin{align}
J_{i,i'}:= \sum_{Z:Z\ni \{i,i'\}}\norm{ h_Z} &\le  { g \over d_{i,i'}^{\alpha} } \, , ~~~(i\neq i'). 
\end{align} 
Here, $J_{i,i'}$ is the maximum absolute value of the local Hamiltonian containing the sites $i$ and $i'$. In the quantum case, it is the maximum spectral norm of them. The quantity $g$ is a finite constant. Let $\Lambda$ be the set of sites in the lattice. For $\alpha >D$, one-site energy is bounded as
\begin{align}
\sum_{Z \in i} \| h_Z \|\le \sum_{j\in \Lambda}\sum_{Z \ni (i,j)} \| h_Z \| &\le g u =: g_u,
\end{align}
where we define $d_{i,i}\equiv 1$ and $u:=2^{\alpha}\sum_{j\in \Lambda }1/d_{i,j}^{\alpha}$.

\sectionprl{Appendix D: Outline of the proof of the main theorem \ref{upper-main}}
We present the outline of the proof of Theorem \ref{upper-main}. We here outline the proof only for the transverse Ising case. More details are explained in the SM \cite{Supplement_thermal}. The equal-time current correlation for the transverse Ising case is given by joint cumulants up to the fourth order:
\begin{align}
    &\ave{ \mathcal{J}_n \mathcal{J}_0 }
    = \Big(\frac{Jh}{2} \Big)^2 \sum_{\substack{k \geq n \\ \ell < n }} \sum_{\substack{k' \geq 0 \\ \ell' < 0}} \sum_{ \substack{ \{a,b\}=\{k,\ell\} \\ \{a',b'\}=\{k',\ell'\} }} 
    \frac{ \varepsilon_{a,b} \varepsilon_{a',b'} \delta_{b,b'} }{ d_{a,b}^\alpha d_{a',b}^\alpha } 
    \notag \\ & \hspace{10pt}\times
    \bigl[ \ave{ S^z_a S^z_{a'} (S^y_b)^2 }_c + \ave{ S^z_a S^z_{a'} }_c \ave{ (S^y_b)^2 }_c  \bigr]
    , \label{ising_proof}
\end{align}
where $\varepsilon_{a,b}$ is second order Levi-Civita symbol and $\delta_{b,b'}$ is the Kronecker's delta function.

Using Theorem \ref{cpower-main}, we upper-bound Eq. (\ref{ising_proof}). All the Hamilton paths for the term $\ave{ S^z_a S^z_{a'} (S^y_b)^2 }_c$ are drawn with red lines in the left three diagrams in Fig.\ref{fig4}. The most right diagram is for the term $\ave{ S^z_a S^z_{a'} }_c \ave{ (S^y_b)^2 }_c$. The vertices $a$, $a'$ and $b$ correspond to supports of spin variables $S^z_a$, $S^z_{a'}$ and $(S^y_b)^2$, respectively. The black edges, $a\! -\! b$ and $a'\!\!-\! b$, respectively mean the power-law components $d_{a,b}^{-\alpha}$ and $d_{a',b}^{-\alpha }$ in Eq. (\ref{ising_proof}). The red edges stand for the newly added power-law components given by Theorem \ref{cpower-main} along the Hamilton path. For these structures, we upper-bound the each term. For instance, the most left diagram in Fig.\ref{fig4} is upper-bounded as
\begin{align}
&c \sum_{k,\ell \atop k' , \ell'}\sum_{ a,b \atop a' } {1\over d_{a,b}^{\alpha}} {1\over d_{a',b}^{2 \alpha}} {1 \over d_{a,a'}^{\alpha} } 
< c' \sum_{k,\ell \atop k' , \ell'} \sum_{ a,b \atop a' } {1\over d_{a,b}^{\alpha}}{1\over  d_{a',b}^{\alpha}}{1\over  d_{a,a'}^{\alpha}}  
 \nonumber \\
&< c'' \sum_{k \geq n \atop \ell' < 0} {1\over d_{k,\ell'}^{2\alpha}}
< c''' {1\over n^{2\alpha -2}}, 
\end{align}
where $c,c',c'',c'''$ are constants. Here we have used $d_{a',b}^{-2 \alpha} < d_{a',b}^{-\alpha}$ where we set the lattice constant to one. We also use the inequality $\sum_{y} d_{x,y}^{-\alpha} d_{y,z}^{-\alpha}< u d_{x,z}^{-\alpha}$, which is proven as follows (See also Lemma 2 in the SM \cite{Supplement_thermal}):
    \begin{align}
    &\sum_{y} \frac{1}{d_{x,y}^\alpha} \frac{1}{d_{y,z}^\alpha}
    = \frac{1}{d_{x,z}^\alpha} \sum_y \frac{ d_{x,z}^\alpha }{ d_{x,y}^\alpha d_{y,z}^\alpha }
    \notag \\ &
    \leq \frac{1}{d_{x,z}^\alpha} \sum_y \frac{ (d_{x,y} + d_{y,z})^\alpha }{ d_{x,y}^\alpha d_{y,z}^\alpha }
    \leq  \frac{2^{\alpha-1}}{d_{x,z}^\alpha} \sum_y \frac{ d_{x,y}^\alpha + d_{y,z}^\alpha }{ d_{x,y}^\alpha d_{y,z}^\alpha }
    \notag \\ &
    =  \frac{2^{\alpha-1}}{d_{x,z}^\alpha} \sum_y \Big( \frac{1}{d_{x,y}^\alpha} + \frac{1}{d_{y,z}^\alpha} \Big)
    = \frac{u}{d_{x,z}^{\alpha}} \, , 
    \end{align}
where we define $d_{i,i}\equiv 1$ and use $\{(x+y)/2\}^\alpha \leq (x^\alpha + y^\alpha )/2$. $u = 2^{\alpha} \sum_k d_{k,m}^{-\alpha}$ is finite for $\alpha>1$. The other diagrams are also upper-bounded to get $|\ave{\mathcal{J}_n \mathcal{J}_0}|< \mathcal{O}(1) n^{-2\alpha+2}$.

\clearpage
\newpage

\renewcommand\thefootnote{*\arabic{footnote}}

\addtocounter{section}{0}

\addtocounter{equation}{19}

\renewcommand{\theequation}{S.\arabic{equation}}

\renewcommand{\thesection}{S.\Roman{section}}
\begin{widetext}

\begin{center}
{\large \bf Supplementary Material for  ``Energy diffusion in the long-range interacting spin systems''}\\
\vspace*{0.3cm}
Hideaki Nishikawa$^{1,2}$ and Keiji Saito$^{1}$ \\
\vspace*{0.1cm}
 $^{1}${\small {\it Department of Physics, Kyoto University, Kyoto 606-8502, Japan}}
 \\
 $^{2}${\small {\it Analytical Quantum Complexity RIKEN Hakubi Research Team, RIKEN Center for Quantum Computing (RQC), 2-1 Hirosawa, Wako, Saitama 351-0198, Japan}}
\end{center}

\tableofcontents

\section{Cumulant power-law clustering}
\label{cpc}

\subsection{Notations}
\label{notations}
We consider a $D$-dimensional long-range interacting spin systems, where each spin sits on the squared lattice. We denote the set of sites by $\Lambda$. We denote the cardinality of the set $X (\subseteq \Lambda )$ by $|X|$. We define the distance $d_{X,Y}$ between the two subsets $X$ and $Y (\subseteq \Lambda)$ by the Manhattan distance of the shortest path between the sets. We assume that the lattice constant is one. For the convenience in computation, we define $d_{i,i}:=1$. Note that the Euclidean distance is shorter than the Manhattan distance, and hence we also use the Euclidean distance to make various upper bounds. We also denote the complementary set of $X$ by $X^\co$, i.e., $X^\co = \Lambda\setminus X$. We frequently use $O_{X}$ for a quantity $O$ (an operator in the quantum case) supported by the region $X$. The symbol $\|O\|$ implies the maximum value of the quantity $|O|$, which means the spectral norm of the operator $O$ in the quantum case. The symbol $\|O\|_1$ is L${}^1$ norm for the quantity $O$, corresponding to the trace norm in the quantum case. In this section, we use mostly the terminology of quantum case, since the results are valid even in the quantum case.

\subsection{Joint cumulant}
\noindent
We consider the joint cumulant defined through the cumulant generating function. Let $O_{X}$ be an operator supported by the region $X$. Suppose that $X_1, X_2, \cdots , X_n$ have no overlap between them, i.e., $X_{i}\cup X_j =\emptyset$ for any $i$ and $j$ ($1\le i,j \le n$). The joint cumulant is defined as
\begin{align}
\langle O_{X_1}O_{X_2} \cdots O_{X_n} \rangle_c &= 
{\partial^n \over \partial \lambda_1 \partial \lambda_2 \cdots \partial \lambda_n} \Bigr|_{{\bm \lambda}={\bm 0}} \mu ({\bm \lambda }) \, , \label{defcumulant}\\
\mu ({\bm \lambda}) &=\ln \Bigl< \exp\left( \sum_{i} \lambda_i O_{X_i} \right) \Bigr> \, ,\label{defcumulantgf}
\end{align}
where $\langle ... \rangle$ is the average over an equilibrium state $\rho$ with the inverse temperature $\beta(=(k_{\rm B} T)^{-1})$ and $\langle ... \rangle_c$ is the joint cumulant. Here, $k_{\rm B}$ is the Boltzmann constant and $T$ is the temperature. 
The first order cumulant is identical to the first order moment, i.e., $\langle O_{X} \rangle_c = \langle O_{X} \rangle$. Other examples include the second order: $\langle O_{X_1} O_{X_2} \rangle_c = \langle O_{X_1} O_{X_2} \rangle - \langle O_{X_1} \rangle \langle O_{X_2} \rangle $ and the third order: $\langle O_{X_1} O_{X_2} O_{X_3} \rangle_c = \langle (O_{X_1} - \langle O_{X_1}\rangle )(O_{X_2} - \langle O_{X_2}\rangle )(O_{X_3} - \langle O_{X_3}\rangle )\rangle$. The crucial property is that the $n$th order mixed moment is expressed through all decomposition of cumulants, i.e., 
\begin{align}
\langle O_{X_1} O_{X_2} \cdots O_{X_n} \rangle &=\langle O_{X_1} O_{X_2} \cdots O_{X_n} \rangle_c
+ D_n \, ,\label{momentofromcumulant1}
\end{align}
where $D_n$ is a sum of all decomposition with joint cumulants up to $(n-1)$th order. Examples include the second and third-order expressions:
\begin{align}
\langle O_{X_1} O_{X_2} \rangle &= \langle O_{X_1} O_{X_2} \rangle_c + D_2 \, , \\
D_2&:=\langle O_{X_1} \rangle_c \langle O_{X_2} \rangle_c  \, , \\
\langle O_{X_1} O_{X_2} O_{X_3}\rangle &= \langle O_{X_1} O_{X_2} O_{X_3}\rangle_c + D_3 \, , \\
D_3&:=\langle O_{X_1} O_{X_2}\rangle_c \langle O_{X_3}\rangle_c + \langle O_{X_1} O_{X_3}\rangle_c \langle O_{X_2}\rangle_c + \langle O_{X_2} O_{X_3}\rangle_c \langle O_{X_1}\rangle_c + \langle O_{X_1} \rangle_c \langle O_{X_2} \rangle_c\langle O_{X_3} \rangle_c  \, . 
\end{align}
The following expression is also useful in the following analysis.
\begin{align}
\langle O_{X_1} O_{X_2} \cdots O_{X_n} \rangle_c&=\langle O_{X_1} O_{X_2} \cdots O_{X_n} \rangle - D_n \, . \label{momentofromcumulant2}
\end{align}

\subsection{Cumulant power-law clustering theorem}
We consider a $k$-local long-range interacting spin system, where the Hamiltonian is categorized into the following form
\begin{align}
H = \sum_{Z:|Z|\le k} h_Z \, ,  \label{def:Hamiltonian}
\end{align}
where $Z$ is a subset indicating the interacting sites, and $h_Z$ is the local Hamiltonian. Note that $|Z|$ is the number of the interacting sites, and hence $h_Z$ means $|Z|$-body interaction. For instance, for the long-range transverse Ising Hamiltonian $H=-J \sum_{i,j}S_i^z S_j^z/d_{i,j}^{\alpha} -h \sum_{i} S_i^x$, we have the subset $Z=\{1,2,3,\cdots, (1,2),(1,3),(1,4),\cdots, (2,3),(2,4)\cdots \}$ and local Hamiltonian $h_{Z=(i)}=-hS_i^x$ and $h_{Z=(i,j)}=-JS_i^z S_j^z / d_{i,j}^{\alpha}$. To derive the cumulant power-law clustering theorem for a general setting, we assume that the interaction form satisfies the following property
\begin{align}
\label{def_short_range_long_range}
J_{i,i'}:= \sum_{Z:Z\ni \{i,i'\}}\norm{ h_Z} &\le  { g \over d_{i,i'}^{\alpha} } \, , ~~~(i\neq i'). 
\end{align} 
Here, $J_{i,i'}$ is the maximum absolute value of the local Hamiltonian containing the sites $i$ and $i'$. In the quantum case, it is the maximum spectral norm of them. The quantity $g$ is a constant. 
We consider the $D$-dimensional lattice, and the regime $\alpha >D$ is focused on so that the thermodynamic extensivity holds. 
We note the relation 
\begin{align}
\sum_{Z \in i} \| h_Z \|\le \sum_{j\in \Lambda}\sum_{Z \in (i,j)} \| h_Z \| &\le g u ,
\end{align}
where $u:=2^{\alpha}\sum_{j\in \Lambda }1/d_{i,j}^{\alpha}$.

\begin{theorem}[Cumulant power-law clustering theorem]\label{cumulantpowerlawtheorem}
We consider the $k$-local long-range spin systems in the $D$-dimension, which satisfies the condition (\ref{def_short_range_long_range}). Let $X_1,...,X_n$ be supports of physical quantities that have no overlaps between them. Under the conditions $\alpha>D$ and $\beta < \beta_c$, the $n$th order joint cumulant in the equilibrium state is bounded as follows.
\begin{align}
\ave{ O_{X_1} ... O_{X_n} }_c \leq c \, 2^n n! \prod_{j=1}^n \| O_{X_j} \| \sum_P \prod_{i=1}^{n-1} {|X_{P_i}| \, |X_{P_{i+1}}| e^{(|X_{P_i}| + |X_{P_{i+1}}|)/k } \over (d_{X_{P_i},X_{P_{i+1}}})^{\alpha} }\, , 
 \end{align}
where $c$ is a constant, and $P$ stands for the Hamilton path for a graph with the vertices $X_1, \cdots , X_{n}$. $X_{P_i}$ is the support at the $i$th step of the Hamilton path $P$, and $P_i$-$P_{i+1}$ corresponds to the $i$th edge. $\beta_c=(ngu k )^{-1} W(1/2e^2)$, where $W$ is the Lambert function. We also note $\beta_c< (4 engu k )^{-1} $ through the inequality $W(1/2e^2) < 1/(4e)$.
\end{theorem}

\noindent 

From now on, we prove Theorem \ref{cumulantpowerlawtheorem}. To this end, we use the cluster expansion~\cite{CMIclustering,kuwahara2020gaussian,PhysRevX.4.031019} for the $n$th order joint cumulant. The crucial method that we use here is the multi-phase space technique. Our approach is available even for the quantum case, and hence the multi-phase space reads the multi-Hilbert space in the quantum case. Let $O_X^{(i)}$ be an operator $O_X$ supposed by the region $X$ in the $i$th space $(i=1,\cdots , n)$, and let $\tilde{\rho}$ be a multi-copied density operator of the equilibrium state $\rho$
\begin{align}
\tilde{\rho} &= \rho^{(1)} \otimes \rho^{(2)} \otimes \cdots \otimes \rho^{(n)} \, ,
\end{align}
where $\rho^{(i)}$ is the copied density operator to the $i$th space. We also define the operator $O_{X}^{(i,j)}\,(i<j)$ as
\begin{align}
O_{X}^{(i,j)} &:= O_{X}^{(i)} - O_{X}^{(j)}  \,, ~~~ (i<j) , 
\end{align}
where we simply denote $O_{X}^{(i)}$ as $O_{X}^{(i)} = \bm{1 } \otimes \bm{1} \otimes \cdots \otimes O_{X}^{(i)} \otimes \cdots \otimes \bm{1}$. 
In addition, we define the following cumulant operator.
\begin{definition}(Cumulant operator)
We define the $n$th order cumulant operator $\tilde{O}_{X_1, X_2 , \cdots , X_n}^{(i_1, i_2, \cdots , i_n)}$ 
\begin{align}
\tilde{O}_{X_1, X_2 , \cdots , X_n}^{(i_1, i_2, \cdots , i_n)} &:= O_{X_1}^{(i_1 )}\prod_{j=2}^{n}\left( \sum_{k=1}^{j-1}O_{X_j}^{(i_k , i_j)} \right) \, , ~~~ i_1 < i_2 < \cdots < i_n \, . 
\end{align}
The product is aligned from the left to right in the quantum case. When $n=1$, $\tilde{O}_{X_1}^{(i_1)} =O_{X_1}^{(i_1)}$. Examples include the second order ($n=2$): $\tilde{O}_{X_1, X_2}^{(1,2)}=O_{X_1}^{(1)}O_{X_2}^{(1,2)}$, and the third order ($n=3$): $\tilde{O}_{X_1, X_2}^{(1,2)}=O_{X_1}^{(1)}O_{X_2}^{(1,2)}(O_{X_3}^{(1,3)}+ O_{X_3}^{(2,3)})$.
\end{definition}

\begin{prop}\label{cumulant-lemma}The $n$th order cumulant operator yields the $n$th order joint cumulant
\begin{align}
{\rm tr} \left( \tilde{\rho}\, \tilde{O}_{X_1, X_2 , \cdots , X_n}^{(i_1, i_2, \cdots , i_n)} \right) &= \langle O_{X_1} O_{X_2} \cdots O_{X_n} \rangle_{c} \, ,~~~~~(i_1 < i_2 <\cdots < i_n) \, ,
\end{align}
where ${\rm tr}$ means the phase-space integral over the multi-phase space, and in the quantum case, it reads the trace over the multi-Hilbert space. 
\end{prop}
\noindent 
Here, the choice of space $(i_1,i_2, \cdots, i_n)$ is arbitrary as long as $i_1 < i_2< \cdots < i_n$. Practically, we use $(i_1,i_2, \cdots, i_n)=(1,2,\cdots, n)$. The proof of the proposition \ref{cumulant-lemma} is provided in the next subsection.

From the proposition \ref{cumulant-lemma}, the $n$th order joint cumulant in the equilibrium state can be obtained as 
\begin{align}
\langle O_{X_1}O_{X_2}\cdots O_{X_n} \rangle_{c} &=\tr \left( \tilde{\rho} \tilde{O}_{X_1, X_2 , \cdots , X_n}^{(1, 2, \cdots , n)} \right) \nonumber \\
&=\tr\left(\tilde{\rho} O_{X_1}^{(1)} O_{X_2}^{(1,2)} (O_{X_3}^{(1,3)}+ O_{X_3}^{(2,3)} )\cdots (O_{X_n}^{(1,n)}+ \cdots +O_{X_n}^{(n-1,n)} ) \right) \, ,
\end{align}
where $\tilde{\rho}$ is the multi-copied equilibrium density operator defined as
\begin{align}
  \tilde{\rho} &=\rho^{(1)}\otimes \rho^{(2)}\otimes \cdots \otimes \rho^{(n)} = {1\over {\cal Z}^n} \exp (-\beta H^{\rm sym}) \, , \\
H^{\rm sym} &=H\otimes {\bm 1}\otimes \cdots \otimes{\bm 1}+{\bm 1}\otimes H \otimes \cdots \otimes{\bm 1}+\cdots {\bm 1}\otimes {\bm 1} \otimes \cdots \otimes H \, . 
\end{align}
Here, ${\cal Z}$ is the partition function for only one space, i.e., ${\cal Z}=\tr_{(1)} e^{-\beta H}$ where $\tr_{(1)}$ means the trace only for one space. In general, it is convenient to use the symmetrized operators over all spaces for an arbitrary operator $O$:
\begin{align}
O^{\rm sym}&=O\otimes {\bm 1}\otimes \cdots \otimes{\bm 1}+{\bm 1}\otimes O \otimes \cdots \otimes{\bm 1}+\cdots +{\bm 1}\otimes {\bm 1} \otimes \cdots \otimes O \, . 
\end{align}

We consider the case where the sets $X_1, X_2, \cdots , X_n$ have no overlap between them. In the cluster expansion, we expand the Hamiltonian using the Taylor expansion, i.e., high-temperature expansion, and then it is crucial to use local Hamiltonians with the connected sets $Z_1, Z_2, \cdots $ which connect the sets $X_1, X_2, \cdots , X_n$. ``Connected'' here implies that any set $X_k$ can be reached from any starting set $X_{k'}$ by successively traversing overlapping subsets in $\{Z_1, Z_2, \cdots \}$ and $X_1, X_2, \cdots , X_n$. If $Z_1, Z_2, \cdots , Z_m $ do not connect the sets $X_1, X_2, \cdots , X_n$, then, it is easy to show 
\begin{align}
&\tr\left( O_{Z_1}^{\rm sym}O_{Z_2}^{\rm sym}\cdots O_{Z_m}^{\rm sym} \tilde{O}_{X_1, X_2 , \cdots , X_n}^{(1, 2, \cdots , n)} \right) \nonumber \\
&=
\tr\left( O_{Z_1}^{\rm sym}O_{Z_2}^{\rm sym}\cdots O_{Z_m}^{\rm sym} O_{X_1}^{(1)} O_{X_2}^{(1,2)} (O_{X_3}^{(1,3)}+ O_{X_3}^{(2,3)} )\cdots (O_{X_n}^{(1,n)}+ \cdots +O_{X_n}^{(n-1,n)} ) \right)=0 \, . \nonumber 
\end{align}
This immediately results in that connected sets give finite contributions in the cluster expansion. Thus, the $n$th order joint cumulant in the equilibrium state is formally expressed as 
\begin{align}
\langle O_{X_1}O_{X_2}\cdots O_{X_n} \rangle_{c} &= \tr \bigl( \tilde{\rho} 
\tilde{O}_{X_1, X_2 , \cdots , X_n}^{(1, 2, \cdots , n)}\bigr) \nonumber \\
&=\frac{1}{{\cal Z}^n}\sum_{m=0}^\infty \sum_{Z_1,Z_2,\ldots,Z_m: \atop {\rm connected}} \frac{(-\beta )^m}{m!} \tr \br{h_{Z_1}^{\rm sym} h_{Z_2}^{\rm sym} \cdots h_{Z_m}^{\rm sym} \tilde{O}_{X_1, X_2 , \cdots , X_n}^{(1, 2, \cdots , n)}},
\end{align}
where the summation $\sum_{Z_1, Z_2,\ldots, Z_m:\atop {\rm connected}}$ implies taking all connected sets. Note that $\{Z_s\}_{s=1}^m$ can be the same, e.g., $Z_1=Z_2$. We introduce the distribution $\rho^{\rm con}$ which consists of connected sets:
\begin{align}
\label{correlation_expansion}
\rho^{\rm con}=\frac{1}{{\cal Z}^n} \sum_{m=0}^\infty \sum_{Z_1,Z_2,\ldots,Z_m:\atop {\rm connected}} \frac{(-\beta )^m}{m!} h_{Z_1}^{\rm sym} h_{Z_2}^{\rm sym} \cdots h_{Z_m}^{\rm sym}. 
\end{align}
Using the operator $\rho^{\rm con}$, the $n$th order joint cumulant is upper-bounded as
\begin{align}
\label{correlation_up_cluster}
\langle O_{X_1}O_{X_2}\cdots O_{X_n} \rangle_{c}  &\le 2^n n! \| O_{X_1} \| \| O_{X_1} \| \cdots  \| O_{X_n} \|
\| \rho^{\rm con} \|_1 \, .
\end{align}
Here $\|...\|$ and $\|...\|_1$ are respectively the spectral norm and trace norm in the quantum case, while in the classical case, they are respectively the maximum absolute value and $L{}^1$ norm of the quantities. We note that in the connected set ${Z_1, Z_2, \cdots , Z_m}$, not all the sets contribute to connecting to the sets $X_1, X_2, \cdots, X_n$. Namely, some subsets may be separated from the connected subsets. Hence we classify $\{ Z_1, Z_2, \cdots , Z_m \}$ into the subsets $w_{\rm cl}$ participating to the connection to the set $\{X_1, X_2, \cdots, X_n \}$ and the ones not participating (or separated) $w_{\rm cl}^\co$, i.e., 
\begin{align}
\{ Z_1,Z_2,\ldots,Z_m\} = w_{\rm cl} \oplus w_{\rm cl}^\co,\quad w_{\rm cl}^\co:= \{ Z_1,Z_2,\ldots,Z_m\}\setminus w_{\rm cl} \, . 
\end{align}
Any subsets in $w_{\rm cl}$ have no overlap with any subsets in $w^\co_{\rm cl}$. Only subsets in $w_{\rm cl}$ connects the subsets $X_1, X_2, \cdots, X_n$. We define all the sets of $w_{\rm cl}$ by $\mathcal{G}_{\rm cl}$. We also define the sets $\mathcal{G}^{\rm c}_{\rm cl}$ for all interaction subsets except $w_{\rm cl}$. Note that $\mathcal{G}^{\rm c}_{\rm cl}$ depends on the set $w_{\rm cl}$, and $w_{\rm cl}^\co \in \mathcal{G}^{\rm c}_{\rm cl}$. Let $H^{\rm sym}_{\co}$ be the Hamiltonian with the local Hamiltonian with the supports $\mathcal{G}^{\rm c}_{\rm cl}$. Obviously $H_{\co}=H-H_{w_{\rm cl}}$, where $H_{w_{\rm cl}}=\sum_{Z \in w_{\rm cl}} h_Z$.
We then obtain 
\begin{align}
\rho^{\rm con} 
&= {1\over {\cal Z}^n} \sum_{m=0}^{\infty }{(-\beta )^m \over m! }\sum_{s=0}^{m} {m! \over (m-s)! s! } \sum_{\{ Z_1 , \cdots , Z_s\} = w_{\rm cl} \in \mathcal{G}_{\rm cl}}\!\!\!\! \tilde{h}(Z_1, \cdots , Z_s) \sum_{w_{\rm cl}^{\rm c} \in \mathcal{G}^{\rm c}_{\rm cl}}\sum_{Z_{s+1}, \cdots , Z_m  \atop \in w_{\rm cl}^{\rm c}}\!\!\!\!\!\! h^{\rm sym}_{Z_{s+1}} h^{\rm sym}_{Z_{s+2}}\cdots  h^{\rm sym}_{Z_{m}} \nonumber \\
&= {1\over {\cal Z}^n} \sum_{s=0}^{\infty }{(-\beta )^s \over s! } \sum_{w_{\rm cl} \in \mathcal{G}_{\rm cl}} \tilde{h}(Z_1, \cdots , Z_s) e^{-\beta H^{\rm sym}_{\co}} .\label{decomp_cl_and_cl-co}
\end{align}
Here, we have defined $\tilde{h}(Z_1 , \cdots , Z_s)  =\sum_{\tilde{P}} h^{\rm sym} (Z_{\tilde{P}_1})\cdots  h^{\rm sym}(Z_{\tilde{P}_s}) $, where $\tilde{P}$ implies taking all combinations of different elements for $Z$. In particular, when $Z$ are all different, it means taking a permutation. We next note the following relation using the Golden-Thompson inequality.
\begin{align}
\|e^{-\beta H^{\rm sym}_{\co}} \|_1 &=\tr \bigl(e^{-\beta H^{\rm sym}_{\co}} \bigr)= \tr \bigl( e^{-\beta (H^{\rm sym} -H^{\rm sym}_{{w_{\rm cl}}})} \bigr) 
\le  \tr \bigl( e^{-\beta H^{\rm sym}  }  e^{\beta  H^{\rm sym}_{{w_{\rm cl}}}} \bigr) \le {\cal Z}^n  e^{n\beta g u ks} \, . \label{ebhsymco}
\end{align}
where we use $\tr \bigl( e^{-\beta H^{\rm sym}  }  e^{\beta  H^{\rm sym}_{{w_{\rm cl}}}} \bigr) =\| e^{-\beta H^{\rm sym}  }  e^{\beta  H^{\rm sym}_{{w_{\rm cl}}}}\|_1 \le \| e^{-\beta H^{\rm sym}  }  \|_1 \| e^{\beta  H^{\rm sym}_{{w_{\rm cl}}}}\|=\tr (e^{-\beta H^{\rm sym}  } ) \| e^{\beta  H^{\rm sym}_{{w_{\rm cl}}}}
\|$. 
We also have used $\|H^{\rm sym}_{{w_{\rm cl}}} \|\le g u k s $. Using the inequality (\ref{ebhsymco}), we obtain the following relation
\begin{align}
\| \rho^{\rm con} \|_1& \le \sum_{s=0}^{\infty} { (n \beta e^{n\beta g u k   })^s } \sum_{\{Z_1, \cdots , Z_{s}\} =w_{\rm cl} \in \mathcal{G}_{\rm cl}}  \|h_{Z_1} \|\|h_{Z_2} \|\cdots \|h_{Z_s} \| \, ,
\end{align}
where we have used the relation $\tilde{h}(Z_1 , \cdots , Z_s)\le s! n^s \| h_{Z_1}\| \| h_{Z_2}\|\cdots \|h_{Z_s}\|$.

\begin{figure}[!ht]
    \begin{tabular}{c}
      \begin{minipage}{0.45\hsize}
        \begin{center}
                \includegraphics[width= \textwidth]{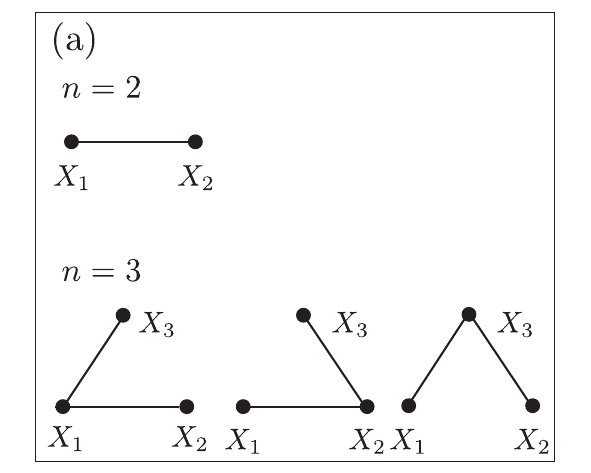}
        \end{center}
      \end{minipage}

      \begin{minipage}{0.45\hsize}
        \begin{center}
                \includegraphics[width= \textwidth]{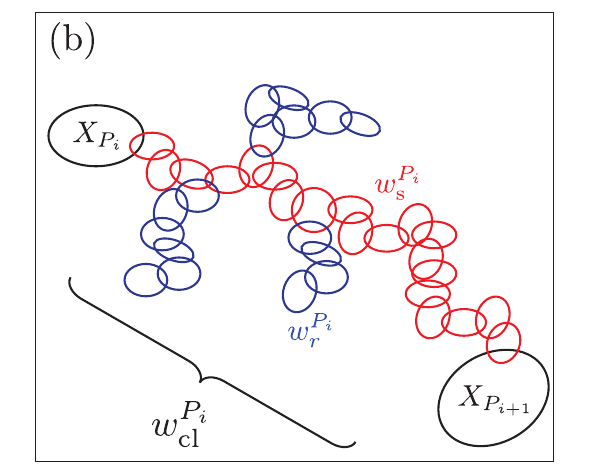}
        \end{center}
      \end{minipage}
    \end{tabular}
\caption{(a): Patterns of the Hamilton Path. (b): Subsets of the shortest path $w_{\rm s}^{P_i}$ and the subsets of the remaining part of the subsets $w_{\rm r}^{P_i}$}\label{pattern}
\end{figure}

We now classify the set $w_{\rm cl}$ in more detail. Note again that the subsets in $w_{\rm cl}$ connects the sets $X_1$, $X_2$,..., $X_s$. There are many patterns in connecting the sets. For instance, $n=2$, there is only $1$ pattern. The case $n=3$ has $3$ patterns. In general, the case of $n$ has $n!/2$ patterns. See the Fig.\ref{pattern} (a). Finding these patterns is equivalent to finding the pattern of the Hamilton path problem in a graph with the vertex $X_1,\cdots,X_n$. Let $P$ be the index of the pattern (Hamilton path), and let $P_i\,(i=1,\cdots ,n)$ be the vertex at the $i$th step, which corresponds to the set $X_{P_i}$. The edge between $P_i$ and $P_{i+1}$ is connected by a cluster consisting of the subsets in $w_{\rm cl}$. 
Let $w_{\rm cl}^{P_i} (\in w_{\rm cl})$ be a set that consists of the subsets contributing to the edge $P_i$-$P_{i+1}$. See the Fig.\ref{pattern}(b) on these situations.

Let us focus on the edge $P_i$-$P_{i+1}$ and the set $w_{\rm cl}^{P_i}$. We divide the set $w_{\rm cl}^{P_i}$ into the two subsets $w_{\rm s}^{P_i}$ and $w_{\rm r}^{P_i}~(w_{\rm cl}^{P_i} = w_{\rm s}^{P_i} \oplus w_{\rm r}^{P_i})$, where $w_{\rm s}^{P_i}$ is a set representing the shortest path and $w_{\rm r}^{P_i}$ is the rest of the subsets. Let $\ell_s$ and $\ell_r$ be respectively the lengths for the sets $w_{\rm s}^{P_i}$ and $w_{\rm r}^{P_i}$. We consider these contributions. To this end, we use the following lemma.

\begin{lemma} \label{chain}
    Assuming that $\alpha>D$, the following inequalities hold.
    \begin{align}
        \sum_{m } \frac{1}{d_{i,m}^\alpha} \frac{1}{d_{m,j}^\alpha} &\leq  \frac{u}{d_{i,j}^\alpha},  \label{chain-1}\\
\sum_{m \in \Lambda} J_{i,m} J_{m,j} &\le {g^2 u \over d_{i,j}^{\alpha}} \, , \label{twice}
    \end{align}
    where $u = 2^\alpha \sum_{m } d_{i,m}^{-\alpha} (<\infty )$. The iterative use of (\ref{twice}) leads to $\left[ {\bm J}^{\ell}\right]_{i,k} \le {g^{\ell}u^{\ell -1} / (d_{i,k})^{\alpha}}$.
\end{lemma}

\begin{proof}[Proof of Lemma \ref{chain}] Proving (\ref{chain-1}) is enough.
    \begin{align}
    &\sum_{k} \frac{1}{d_{i,k}^\alpha} \frac{1}{d_{k,j}^\alpha}
    = \frac{1}{d_{i,j}^\alpha} \sum_k \frac{ d_{i,j}^\alpha }{ d_{i,k}^\alpha d_{k,j}^\alpha }
    \leq \frac{1}{d_{i,j}^\alpha} \sum_k \frac{ (d_{i,k} + d_{k,j})^\alpha }{ d_{i,k}^\alpha d_{k,j}^\alpha }
    \notag \\ &
    \leq 2^{\alpha-1} \frac{1}{d_{i,j}^\alpha} \sum_k \frac{ d_{i,k}^\alpha + d_{k,j}^\alpha }{ d_{i,k}^\alpha d_{k,j}^\alpha }
    = 2^{\alpha-1} \frac{1}{d_{i,j}^\alpha} \sum_k \Big( \frac{1}{d_{i,k}^\alpha} + \frac{1}{d_{k,j}^\alpha} \Big)
    = \frac{u}{d_{i,j}^{\alpha}}
    \end{align}
Here, we use $\{(x+y)/2\}^\alpha \leq (x^\alpha + y^\alpha )/2$
, following from the convexity of $f(x)=x^\alpha$ and the convexity inequality. Note that we define $d_{i,i}=1$ throughout this paper.
\end{proof}
We first consider the contribution from the shortest path, i.e., $w_{{\rm s}}^{P_i}$ is computed as follows. 
\begin{align}
\sum_{\{Z_1, \cdots Z_{\ell_s}\} = w_{\rm s}^{P_i} \atop \in w_{\rm cl}^{P_i} \in w_{\rm cl} \in \mathcal{G}_{\rm cl}}\!\!\!\! \| h_{Z_1}\|\cdots \|h_{Z_{\ell_s}}\|
&\le \sum_{i_1 \in X_{P_{i}}}\sum_{i_2 \in \Lambda } \sum_{Z_{1} \ni \{ i_1, i_2\}}\sum_{i_3 \in \Lambda }\sum_{Z_{2} \ni \{ i_2 , i_3\}} \cdots \sum_{i_{\ell_s} \in \Lambda } \sum_{i_{\ell_s+1} \in X_{P_{i+1}} }\sum_{Z_{\ell_s}\ni \{i_{\ell_s}, i_{\ell_{s} +1} \} } \!\!\!\! \| h_{Z_1} \| \cdots \| h_{Z_{\ell_s}}\| \nonumber \\
&= \sum_{i_1 \in X_{P_{i}}} \sum_{i_2 \in \Lambda } \cdots  \sum_{i_{\ell_s} \in \Lambda } \sum_{i_{\ell_s +1} \in X_{P_{i}+1}} J_{i_1,i_2}\cdots J_{i_{\ell_s},i_{\ell_s+1}} \nonumber \\
&\le {u^{-1} ( g u )^{\ell_s} }{|X_{P_i}| \,|X_{P_{i+1}}| \over (d_{X_{P_i},X_{P_{i+1}}} )^{\alpha }} \, .
\end{align}
We next consider the contribution from the remaining part of the subsets, i.e., $w_{\rm r}^{P_i}$. Let us define the set $L:=X_{P_i} \cup X_{P_{i+1}} \cup Z_{1} \cup \cdots Z_{\ell_s}$, where $Z_1, \cdots Z_{\ell_s}$ are the subsets in $w_{\rm s}^{P_i}$. To consider the summation over the remaining part of the subsets, we compute the following summation of the tuples
\begin{align}
\sum_{\{ Z'_1, \cdots, Z'_{\ell_r}\}=w_{\rm r}^{P_i} \atop \in w_{\rm cl}^{P_i} \in w_{\rm cl} \in \mathcal{G}_{\rm cl}}\!\!\!\!\!\!\!\! \| h_{Z'_1} \| \cdots \| h_{Z'_{\ell_r}}\| &={1\over \ell_r !} \sum_{( Z'_1, \cdots, Z'_{\ell_r}) : w_{\rm r}^{P_i}} \| h_{Z_1'}\|\cdots  \| h_{Z_{\ell_r}'}\| \, . \label{remaining1}
\end{align}
Here, $(Z'_1, \cdots, Z'_{\ell_r})$ is the summation over possible tuples. Below, for the notational simplicity, we use the notation $\xi_i$ instead of $Z'_{i} $. We define the set $w_L :=\{L,\xi_1, \cdots, \xi_{\ell_r} \}$. In addition, following Ref.\cite{kuwahara2020gaussian}, we introduce the layered structure:
\begin{align}
w_L &= L \oplus \xi^{(1)}\oplus \xi^{(2)} \oplus \cdots \oplus \xi^{(v)}\, ~~~1 \le v \le \ell_r \, ,
\end{align}
where $\xi^{(j)}$ has an overlap with the set $\xi^{(j-1)}$ ($\xi^{(0)}:=L$). $\xi^{(j)}$ is a set achievable from the set $L$ with the $j$th step. This layered strature is depicted in Fig.\ref{fig:layer}. Note that $\xi^{(j)}$ is a set consisting of multiple sets. We denote $|\xi^{(j)}|$ by $n_j~ (\ge 1)$. Then (\ref{remaining1}) is written as follows
\begin{align}
{1\over \ell_r !} \sum_{( Z'_1, \cdots, Z'_{\ell_r}) : w_{\rm r}^{P_i}} \| h_{Z_1'}\|\cdots  \| h_{Z_{\ell_r}'}\| 
&= {1\over \ell_r !} \sum_{v=1}^{\ell_r} \sum_{n_1 + \cdots + n_v =\ell_r \atop n_j \ge 1} {\ell_r ! \over n_1 ! \cdots n_v !} 
\nonumber \\
&\times \Bigl( \sum_{\xi^{(1)}_i | L} \| h_{\xi^{(1)}_i}\| \Bigr)^{n_1} \Bigl( \sum_{ \xi^{(2)}_i | \{ L,\xi^{(1)} \} } \| h_{\xi^{(2)}_i}\|\Bigr)^{n_2}\cdots \Bigl( \sum_{ \xi^{(v)}_i | \{ L,\xi^{(1)},\cdots , \xi^{(v-1)} \} } \| h_{\xi^{(v)}_i}\|\Bigr)^{n_v} \, , 
\end{align}
where $\xi^{(j)}_i |\{L, \xi^{(1)},\cdots , \xi^{(j-1)} \}$ implies taking a summation over $\xi^{(j)}_i $ with a fixed set $\{L, \xi^{(1)},\cdots , \xi^{(j-1)} \}$. We now compute the most right contribution:
\begin{align}
\Bigl( \sum_{ \xi^{(v)}_i | \{ L,\xi^{(1)},\cdots , \xi^{(v-1)} \} } \| h_{\xi^{(v)}_i}\|\Bigr)^{n_v} &\le \Bigl( \sum_{i\in \xi^{(v-1)}} \sum_{Z \ni i}\| h_{Z}\|\Bigr)^{n_v} \le (n_{v-1} k gu )^{n_v} \, . 
\end{align}
Iterating the similar computation down to the integration over $\xi^{(1)}_i$, we obtain the following bound
\begin{align}
{1\over \ell_r !} \sum_{( Z'_1, \cdots, Z'_{\ell_r}) : w_{\rm r}^{P_i}} \| h_{Z_1'}\|\cdots  \| h_{Z_{\ell_r}'}\| 
&\le \sum_{v=1}^{\ell_r} \sum_{n_1 + \cdots + n_v =\ell_r \atop n_j \ge 1} {1 \over n_1 ! \cdots n_v !} (L gu)^{n_1} (k n_1 gu)^{n_2} \cdots(k n_{v-1} gu)^{n_v} \, .  
\end{align}
We here use the relations $n_j \ge (n_j / e)^{n_j}$ and $(1+ x/y)^y$ for $x>0$ and $y >0$ to get 
\begin{align}
{1\over \ell_r !} \sum_{( Z'_1, \cdots, Z'_{\ell_r}) : w_{\rm r}^{P_i}} \| h_{Z_1'}\|\cdots  \| h_{Z_{\ell_r}'}\| 
&\le 
\sum_{v=1}^{\ell_r} \sum_{n_1 + \cdots + n_v =\ell_r \atop n_j \ge 1} e^{n_1} (|L| gu/n_1)^{n_1} e^{n_2}(k n_1 gu/n_2)^{n_2} \cdots e^{n_v}(k n_{v-1} gu/n_v)^{n_v} \, \nonumber \\
&\le (g k u e )^{\ell_r} \sum_{v=1}^{\ell_r} \sum_{n_1 + \cdots + n_v =\ell_r \atop n_j \ge 1}  (1+ (|L|/k) /n_1)^{n_1} (1 + n_1 /n_2)^{n_2} \cdots (1+ n_{v-1} /n_v)^{n_v} \, \nonumber \\
&\le (g k u e^2 )^{\ell_r} e^{|L|/k}\sum_{v=1}^{\ell_r} \sum_{n_1 + \cdots + n_v =\ell_r \atop n_j \ge 1} \nonumber \\
&= 2^{-1}(2 g k u e^2 )^{\ell_r} e^{|L|/k} \nonumber \\
&\le 2^{-1}(2 g k u e^2 )^{\ell_r} e^{(|X_{P_i}|+|X_{P_{i+1}}|)/k + \ell_s } \, ,
\end{align}
where we use the relation $|L| < |X_{P_i}|+ |X_{P_{i+1}}|+ \ell_s k$.
\begin{figure}[t]
    \begin{tabular}{c}
                \includegraphics[width= 0.55\textwidth]{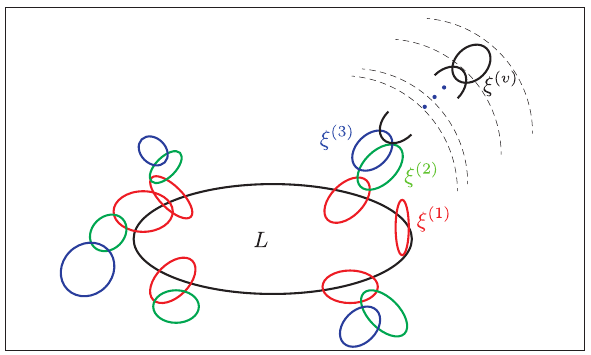}
    \end{tabular}
\caption{Schematic picture of the layered structure}
\label{fig:layer}
\end{figure}

We put all contributions together to get the final expression of the bound for $\| \rho^{\rm con}\|_1$: 
\begin{align}
\|\rho^{\rm con} \|_1  &\le \sum_P \sum_{s=0}^{\infty} (n\beta e^{n\beta guk })^{s}\bigl[\prod_{i=1}^{n-1}\sum_{s_i=0}^{\infty}\delta(\sum_{i=1}^{n-1} s_i ,s) \bigr]\prod_{i=1}^{n-1} \sum_{\ell_{s,i}=0}^{s_i} \sum_{\ell_{r,i}=0}^{s_i} \delta (\ell_{s,i}+\ell_{r,i}, s_i) \, \nonumber \\
&\times 2^{-1}u^{-1}(gu)^{\ell_{s,i}}{|X_{P_i}|\, |X_{P_{i+1}}| \over ( d_{X_{P_i}, X_{P_{i+1}}} )^{\alpha}} \, (2 gue^2 k)^{\ell_{r,i}} e^{\ell_{s,i}} e^{(|X_{P_i}|+|X_{P_{i+1}}|)/k} \nonumber \\
&\le \sum_P \prod_{i=1}^{n-1} {2^{-1} u^{-1}}\sum_{s_i=0}^{\infty} (2 n\beta g u e^2 k e^{n \beta g u k  })^{s_i} \sum_{\ell_{s.i}} (1/2 ke)^{\ell_{s,i}} {|X_{P_i}|\, |X_{P_{i+1}}| \over ( d_{X_{P_i}, X_{P_{i+1}}} )^{\alpha}} e^{(|X_{P_i}|+|X_{P_{i+1}}|)/k} \nonumber \\
&\le \sum_P \prod_{i=1}^{n-1} u^{-1} 
 {|X_{P_i}|\, |X_{P_{i+1}}| \over ( d_{X_{P_i}, X_{P_{i+1}}} )^{\alpha}} e^{(|X_{P_i}|+|X_{P_{i+1}}|)/k} 
 {1\over 1 - 2 n\beta g u e^2 k e^{n \beta g u k } } \, ,
\end{align}
where $\delta (x,y)$ in the first line is the Kronecker's delta function. We have used the convergence condition $\beta < \beta_c$. This completes the proof of Theorem \ref{cumulantpowerlawtheorem}.

\subsection{Proof of Proposition \ref{cumulant-lemma}.}
\begin{proof}
We show the proof of Proposition \ref{cumulant-lemma}.
 We use the method of induction. For $n=2$, let us choose $i_1 =1$ and $i_2=2$. We have 
\begin{align}
\tilde{O}_{X_1,X_2}^{(1,2)} & =O_{X_1}^{(1)} O_{X_2}^{(1)} - {\cal D}_{2}  \, , \\
{\cal D}_2 &= \tilde{O}_{X_1}^{(1)}\tilde{O}_{X_2}^{(2)} \, .
\end{align} 
It is clear that ${\rm tr}(\tilde{\rho} \tilde{O}_{X_1,X_2}^{(1,2)}) = \langle O_{X_1}O_{X_2} \rangle_{c}$. Also, we can check that the space index $(1)$ and $(2)$ can be replaced by arbitrary $(i_1)$ and $(i_2)$ as long as $i_1 < i_2$. The operator ${\cal D}_2$ is simply the product of the first-order cumulants. In general, In general, ${\cal D}_n$ is defined as the summation of all $n$th-order representations decomposed into products of cumulant operators up to the $(n-1)$th order, involving the operators $O_{X_1}, \cdots , O_{X_n}$ defined in the space $(i_1, \cdots , i_n)$. This is an operator version analogous to the c-number expansion of $D_n$ in the relation (\ref{momentofromcumulant2}). Let us consider the case of $n=3$ choosing $(i_1, i_2, i_3)=(1,2,3)$:
\begin{align}
\tilde{O}_{X_1,X_2,X_3}^{(1,2,3)} & =O_{X_1}^{(1)} O_{X_2}^{(1)}(O_{X_3}^{(1)} + O_{X_3}^{(2,3)}) - {\cal D}_3 \, , \\
{\cal D}_3 &=
\tilde{O}_{X_1}^{(1)}\tilde{O}_{X_2,X_3}^{(2,3)} 
+ \tilde{O}_{X_1,X_3}^{(1,3)}\tilde{O}_{X_2}^{(2)}
+\tilde{O}_{X_1,X_2}^{(1,2)}\tilde{O}_{X_3}^{(3)} 
+\tilde{O}_{X_1}^{(1)}\tilde{O}_{X_2}^{(2)}\tilde{O}_{X_3}^{(3)} \,  \nonumber \\
&=
{O}_{X_1}^{(1)}{O}_{X_2}^{(2)} {O}_{X_3}^{(2,3)} +{O}_{X_1}^{(1)}{O}_{X_3}^{(1,3)} {O}_{X_2}^{(2)} 
+{O}_{X_1}^{(1)}{O}_{X_2}^{(1,2)}{O}_{X_3}^{(3)}+ {O}_{X_1}^{(1)} {O}_{X_2}^{(2)} {O}_{X_3}^{(3)}  \, .
\end{align}
Again, we can check that ${\rm tr}(\tilde{\rho}\, \tilde{O}_{X_1,X_2,X_3}^{(1,2,3)} )=\langle O_{X_1}O_{X_2}O_{X_3} \rangle_c $. Note again that our argument is in parallel to the relation (\ref{momentofromcumulant2}). We also observe the relation 
\begin{align}
{\cal D}_3 &= {\cal D}_2 (O_{X_3}^{(1,3)} + O_{X_3}^{(2,3)}) + ( \tilde{O}_{X_1,X_2}^{(1,2)} +
{\cal D}_2 
)\tilde{O}_{X_3}^{(3)}. \label{d3d2relation}
\end{align}
One can check easily that the space index $(1), (2)$ and $(3)$ can be replaced by arbitrary $(i_1)$, $(i_2)$ and $(i_3)$ as long as $i_1 < i_2 < i_3$. 

Now we suppose that the following hold up to $n$:
\begin{align}
\tilde{O}_{X_1, X_2 , \cdots , X_n}^{(1, 2, \cdots , n)} &=\tilde{O}_{X_1,X_2,\cdots , X_n}^{\rm u} - {\cal D}_{n} \, , \label{defn1}\\
\tilde{O}_{X_1,X_2,\cdots , X_n}^{\rm u}&:=O_{X_1}^{(1)}O_{X_2}^{(1)}\prod_{j=3}^n (O_{X_j}^{(1)} + \sum_{k=2}^{j-1} O_{X_j}^{(k,j)}) \, ,\label{defn2}
\end{align}
where the operator $\tilde{O}_{X_1,X_2,\cdots , X_n}^{\rm u}$ yields the $n$th order mixed moment as
\begin{align}
{\rm tr} (\tilde{\rho} \tilde{O}_{X_1,X_2,\cdots , X_n}^{\rm u} )&={\rm tr}( \tilde{\rho} O_{X_1}^{(1)}O_{X_2}^{(1)}\cdots O_{X_n}^{(1)} )= {\rm tr}(\rho  O_{X_1}O_{X_2}\cdots O_{X_n} ) =
\langle O_{X_1}O_{X_2}\cdots O_{X_n} \rangle \, . \label{ouaverage}
\end{align}
Note again that operator ${\cal D}_n$ is defined as the summation of all $n$th-order representations decomposed into products of cumulant operators up to the $(n-1)$th order, involving the operators $O_{X_1}, \cdots , O_{X_n}$ defined in the space $(1, \cdots , n)$, and we have ${\rm tr} (\tilde{\rho} {\cal D}_n )=D_n$. Here, $D_n$ is the summation of all decompositions with the joint cumulants up to the $(n-1)$th order (Note again the relation (\ref{momentofromcumulant2})). Using (\ref{ouaverage}), we obtain ${\rm tr} (\tilde{\rho} \tilde{O}_{X_1, X_2 , \cdots , X_n}^{(i_1, i_2, \cdots , i_n)}) =\langle O_{X_1} O_{X_2}\cdots O_{X_n}\rangle_c$. 
By assuming (\ref{defn1}) and (\ref{defn2}) up to $n$, we show the same expression holds for the case of $(n+1)$ as follows, where we use the lemma \ref{ddecolemma} on ${\cal D}_{n+1}$ proven below:
\begin{align}
\tilde{O}_{X_1,X_2,\cdots , X_{n+1}}^{\rm u}\!\! -\!{\cal D}_{n+1} &= \tilde{O}_{X_1,X_2,\cdots , X_{n+1}}^{\rm u}-{\cal D}_{n}\!\sum_{\ell=1}^{n} O_{X_{n+1}}^{(\ell, n+1)}  - (\tilde{O}_{X_1, X_2 , \cdots , X_n}^{(1, 2, \cdots , n)}\! + \!{\cal D}_n ) O_{X_{n+1}}^{(n+1)}
 \nonumber \\
&= \tilde{O}_{X_1,X_2,\cdots , X_{n+1}}^{\rm u}-{\cal D}_{n}\!\sum_{\ell=1}^{n} O_{X_{n+1}}^{(\ell, n+1)} 
 - \tilde{O}_{X_1,X_2,\cdots , X_{n}}^{\rm u}O_{X_{n+1}}^{(n+1)} \nonumber \\
&=(\tilde{O}_{X_1,X_2,\cdots , X_{n}}^{\rm u} - {\cal D}_{n})\sum_{\ell=1}^{n} O_{X_{n+1}}^{(\ell, n+1)} \nonumber \\
&=\tilde{O}_{X_1, X_2 , \cdots , X_n}^{(1, 2, \cdots , n)}\sum_{\ell=1}^{n} O_{X_{n+1}}^{(\ell, n+1)} = \tilde{O}_{X_1, X_2 , \cdots , X_{n+1}}^{(1, 2, \cdots , n+1)} \, . 
\end{align}
Here, we use Lemma \ref{ddecolemma} in the first line, and Eqs.(\ref{defn1}) and (\ref{defn2}) are used in the second and third lines.
\end{proof}

{\lemma{\label{ddecolemma} The operator ${\cal D}_{n+1}$ is written with the operator ${\cal D}_{n}$ as 
\begin{align}
{\cal D}_{n+1} &= {\cal D}_{n} \sum_{\ell=1}^{n} O_{X_{n+1}}^{(\ell, n+1)}  + (\tilde{O}_{X_1, X_2 , \cdots , X_n}^{(1, 2, \cdots , n)} + {\cal D}_n ) O_{X_{n+1}}^{(n+1)} \, . \label{dn1n}
\end{align}
}}
\begin{proof}
We show the proof of Lemma \ref{ddecolemma}. First, let us consider the second part in Eq.(\ref{dn1n}). This part contains the single $O_{X_{n+1}}^{(n+1)}$ and hence this part generates the decomposition with the first order cumulant operator $\tilde{O}_{X_{n+1}}^{(n+1)}$. See Eq.(\ref{d3d2relation}) as an illustration. Hence, the second part accounts for all decompositions with cumulant operators including only one $\tilde{O}_{X_{n+1}}^{(n+1)}$. On the other hand, the first part in Eq.(\ref{dn1n}) generates all decompositions including cumulant operators that include the cumulant operators involving $O_{X_{n+1}}$ with the second order or higher. See Eq.(\ref{d3d2relation}) again as an illustration. Let us consider the first part. We note that the decomposition ${\cal D}_n$ contains all decompositions decomposed into products of cumulant operators up to the $(n-1)$th order involving the operators $O_{X_1}, \cdots , O_{X_n}$ defined in the space $(1, \cdots, n)$. Let us pick up the following term with $m (\le n)$ decomposition in ${\cal D}_n$:
\begin{align}
\tilde{O}_{X_{i({1})},\cdots ,X_{i({\ell_1})} }^{(i({1}), \cdots , i({\ell_1}))}
\tilde{O}_{X_{i({\ell_1+1})},\cdots ,X_{i({\ell_2})} }^{(i({\ell_1+1}), \cdots , i({\ell_2}))}\cdots 
\tilde{O}_{X_{i ( \ell_{m-1} +1)}, \cdots ,X_{i (\ell_m) }}^{(i ( \ell_{m-1} +1), \cdots , i( \ell_m)  )} \, , 
~~~i({\ell_k +j})< i({\ell_k +j+1}) \, ,  \label{pickup}
\end{align}
where $k=0,\cdots, m-1$, with $\ell_0=0$, and $1 \le j \le  \ell_{k+1} - \ell_{k}-1$. From the conservation of the number of operator, $\ell_m=n$, always.
Each cumulant operator can be elevated to the next higher-order cumulant operator by incorporating $O_{X_{n+1}}$, for instance, 
\begin{align}
\tilde{O}_{X_{i({1})},\cdots ,X_{i({\ell_1})} }^{(i({1}), \cdots , i({\ell_1}))} &\to 
\tilde{O}_{X_{i({1})},\cdots ,X_{i({\ell_1})} }^{(i({1}), \cdots , i({\ell_1}))} 
\left( O_{X_{n+1}}^{(i({1}),n+1)} + \cdots +O_{X_{n+1}}^{(i(\ell_1),n+1)} \right) =\tilde{O}_{X_{i({1})},\cdots ,X_{i({\ell_1})},X_{n+1} }^{(i({1}), \cdots , i({\ell_1}),n+1)} \, . 
\end{align}
Totally $m$ upgrades are generated from the term (\ref{pickup}). Summing up all the contributions yields $(\ref{pickup})\times \sum_{\ell=1}^{n} O_{X_{n+1}}^{(\ell,n+1)}$. Hence, all decompositions including cumulant operators involving $O_{X_{n+1}}$ with the second order or higher are given by the first part on the right hand side in (\ref{dn1n}).
\end{proof}

\section{Proof of the main theorem 2 in the main text}
\label{main}

In this section, using Theorem \ref{cumulantpowerlawtheorem}, we prove the upper bound on the equal-time current correlations for one-dimensional prototypical long-range spin systems. 
We provide separate proofs for the long-range transverse Ising, XY and XYZ model. 

\subsection{Transverse Ising model}

In this subsection, we consider the one-dimensional long-range interacting transverse Ising model, which is described by the following Hamiltonian
\begin{equation}
    H= \sum_{n=1}^N \Big[ -J \sum_{r=1}^{N/2} \frac{S^z_{n} S^z_{n+r}}{r^\alpha} + h S^x_n \Big]
    .
\end{equation}
The local energy at site $n$ is defined as 
\begin{equation}
    \hat{h}_n := - \frac{J}{2} \sum_{r=1}^{N/2} \Big[ \frac{S^z_{n} S^z_{n+r}}{r^\alpha} + \frac{S^z_{n} S^z_{n-r} }{r^\alpha} \Big] + h S^x_n
    ,
\end{equation}
and satisfies the following continuity equations
\begin{align}
    \partial_t \hat{h}_n
    &= - (\mathcal{J}_{n+1} - \mathcal{J}_{n})
    \\
    \mathcal{J}_n 
    &= - \sum_{\substack{i \geq n , j < n \\ \abs{i-j}<N/2 }} t_{i \leftarrow j}
    .
\end{align}
The quantity $t_{i \leftarrow j}$ is defined as
\begin{align}
    t_{i \leftarrow j} 
    &= \{ \hat{h}_i , \hat{h}_j \} 
    = \frac{Jh}{2} \frac{ S^z_i S^y_j - S^y_i S^z_j }{ d_{i,j}^\alpha }
    ,
\end{align}
where $\{..., ...\}$ is spin Poisson bracket, defined as
$\{ A_1 , A_2 \} = \sum_i \varepsilon_{\sigma\tau\upsilon} (\partial A_1 / \partial S^\sigma_i ) (\partial A_2 / \partial S^\tau_i) S^\upsilon_i$.
$t_{i \leftarrow j}$ is interpreted as the microscopic energy transmission from the site $j$ to $i$. Hence, the current ${\cal J}_{n}$ is defined as the collection of energy transmissions through the virtual surface between the sites $n-1$ and $n$

From the symmetry of the Hamiltonian, we note that the following properties hold
\begin{align}
    \ave{ S^y_i } 
    &= \ave{ S^z_i } = 0 \label{ising_prop1}
    ,
    \\
    \ave{ S^y_i S^y_j \mathcal{O} } 
    &= \delta_{i,j} \ave{ (S^y_i)^2 \mathcal{O} } \label{ising_prop2}
    .
\end{align}
The relation (\ref{ising_prop1}) is shown that the Hamiltonian is invariant to the canonical transformation $( S^x_i, S^y_i, S^z_i) \rightarrow (S^x_i, -S^y_i, -S^z_i)$ for $i=1,...,N$. The relation (\ref{ising_prop2}) is derived from the absence of $S^y_i$ components in the Hamiltonian.

Using these properties and Theorem \ref{cumulantpowerlawtheorem}, we establish the following theorem:

\begin{theorem} \label{main_ising}
    In the one-dimensional long-range interacting transverse Ising model, under the condition $\beta < \beta_c$, the equal-time current correlation is upper-bounded as
    \begin{align}
        \abs{ \ave{ \mathcal{J}_n \mathcal{J}_0 } }
        \leq C_0 n^{2-2\alpha}
        \label{upp-main_ising},
    \end{align}
    for $\alpha >1 $.
\end{theorem}

\begin{proof}

The equal-time current correlation $\ave{ \mathcal{J}_n \mathcal{J}_0 }$ is given as follows:
\begin{align}
    \ave{ \mathcal{J}_n \mathcal{J}_0 }
    &= \Big(\frac{Jh}{2} \Big)^2 \sum_{\substack{k \geq n \\ \ell < n}} \sum_{\substack{k' \geq 0 \\ \ell' < 0}} \frac{ \ave{ (S^z_k S^y_\ell - S^y_k S^z_\ell ) ( S^z_{k'} S^y_{\ell'} - S^y_{k'} S^z_{\ell'} ) } }{ d_{k,\ell}^\alpha d_{k',\ell'}^\alpha }
    \notag \\
    &= \Big(\frac{Jh}{2} \Big)^2 \sum_{\substack{k \geq n \\ \ell < n}} \sum_{\substack{k' \geq 0 \\ \ell' < 0}} \sum_{ \substack{ \{a,b\}=\{k,\ell\} \\ \{a',b'\}=\{k',\ell'\} }} \varepsilon_{a,b} \varepsilon_{a',b'} \frac{ \ave{ S^z_a S^y_b S^z_{a'} S^y_{b'} } }{ d_{a,b}^\alpha d_{a',b'}^\alpha }
    \notag \\
    &= \Big(\frac{Jh}{2} \Big)^2 \sum_{\substack{k \geq n \\ \ell < n}} \sum_{\substack{k' \geq 0 \\ \ell' < 0}} \sum_{ \substack{ \{a,b\}=\{k,\ell\} \\ \{a',b'\}=\{k',\ell'\} }} \varepsilon_{a,b} \varepsilon_{a',b'} \delta_{b,b'} \frac{ \ave{ S^z_a S^z_{a'} (S^y_b)^2 } }{ d_{a,b}^\alpha d_{a',b}^\alpha }
    \label{curr_ising} ,
\end{align}
where we use the relation (\ref{ising_prop2}). $\varepsilon_{a,b}, \varepsilon_{a',b'}$ is second order Levi-Civita symbol, where we denote 
\begin{align}
    \varepsilon_{a,b} 
    =
    \begin{cases}
    +1 \, , \quad \text{for} \,\, (a,b)=(k,\ell)
    \\
    -1 \, , \quad \text{for} \,\, (a,b)=(\ell,k)
    \end{cases}
\end{align}
and $\varepsilon_{a',b'}$ is also similarly defined.

To enhance clarity, we provide the diagrammatic representation of $\ave{ \mathcal{J}_n \mathcal{J}_0 }$. The term $1 / ( d_{a,b}^\alpha d_{a',b}^\alpha ) $ in (\ref{curr_ising}) is illustrated diagrammatically in Fig \ref{fig_sup_ising_1}. In Fig.\ref{fig_sup_ising_1}, the vertices of the diagram are mapped to the local areas which corresponds to local quantities within the ensembles (i.e., local spin variables). The edges of the diagram represent the power-law components of the distance between these areas, indicating the connections among them.

\begin{figure}[h]
\centering
\includegraphics[width=0.2\textwidth]{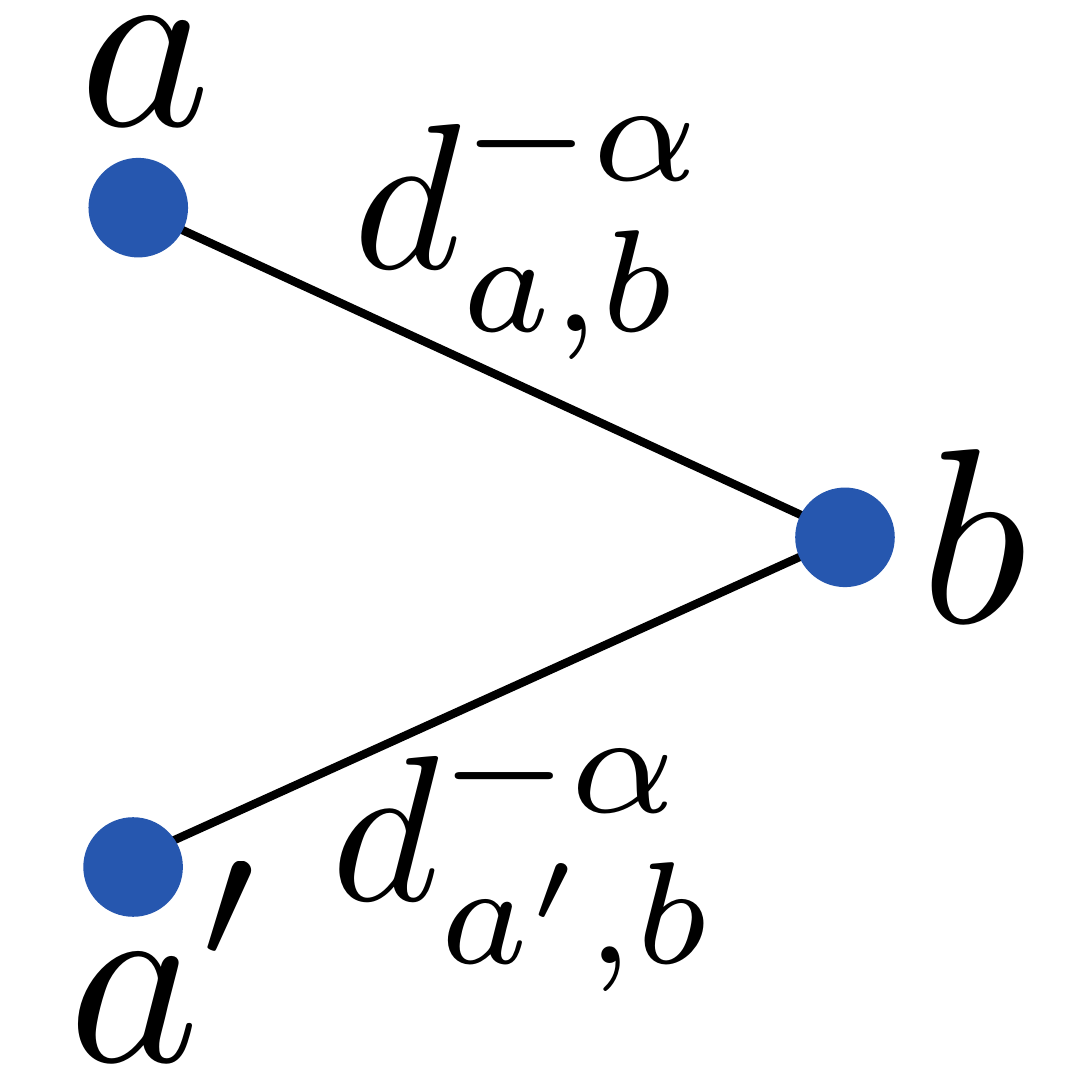}
\caption{The diagrammatic representation of the term $ d_{a,b}^{-\alpha} d_{a',b}^{-\alpha}  $ in (\ref{curr_ising}) }
\label{fig_sup_ising_1}
\end{figure}

The 3-body correlation $\ave{ S^z_a S^z_{a'} (S^y_b)^2 }$ can be decomposed into cumulant products as follows:
\begin{align}
    \ave{ S^z_a S^z_{a'} (S^y_b)^2 } = \ave{ S^z_a S^z_{a'} (S^y_b)^2 }_c + \ave{ S^z_a S^z_{a'} }_c \ave{ (S^y_b)^2 }_c 
    ,
\end{align}
where we use the relation (\ref{ising_prop1}). Next, we use Theorem \ref{cumulantpowerlawtheorem} as follows: 
\begin{align}
    \ave{ S^z_a S^z_{a'} (S^y_b)^2 }_c 
    &\leq \| S^z_a \| \| S^z_{a'} \| \| (S^y_b)^2 \| \Bigg( \frac{ c_1 }{ d_{a,a'}^\alpha d_{a',b}^\alpha } + \frac{ c_1' }{ d_{a,b}^\alpha d_{b,a'}^\alpha } + \frac{c_1''}{ d_{a',a}^\alpha d_{a,b}^\alpha } \Bigg)
    \label{cpc_ising_1}
     ,
    \\
    \ave{ S^z_a S^z_{a'} }_c 
    &\leq \| S^z_a \| \| S^z_{a'} \| \frac{ c_2 }{ d_{a,a'}^\alpha }
    \label{cpc_ising_2}
    .
\end{align}

We now present two trivial inequalities which we use frequently below: 
\begin{align}
     \frac{1}{d_{i,j}^{2\alpha}} 
     < \frac{ 1 }{d_{i,j}^{\alpha}} 
    \label{prop_1}    ,
\end{align}
where we note that the lattice constant is $1$, and
\begin{align}
    \sum_{i} \frac{1}{d_{k,i}^{\alpha}} \frac{1}{d_{i,j}^\alpha} \frac{1}{d_{j,i}^\alpha} \frac{1}{d_{i,\ell}^\alpha} \leq \sum_{i} \frac{1}{d_{k,i}^{\alpha}} \frac{1}{d_{i,j}^\alpha} \sum_{i'} \frac{1}{d_{j,i'}^{\alpha}} \frac{1}{d_{i',\ell}^\alpha} \, , 
    \label{prop_2}    
\end{align}
for all $k,\ell=1,...,N$. 
For the sake of clarity, we provide the diagrammatic representation of Lemma \ref{chain}, (\ref{prop_1}) and (\ref{prop_2}) in Fig \ref{fig_sup_lemma}. Note that in Fig.\ref{fig_sup_lemma}, we omit $\mathcal{O}(1)$ terms and focus exclusively on the terms that depend on the distance between two points.

\begin{figure}[h]
\centering
\includegraphics[width=0.4\textwidth]{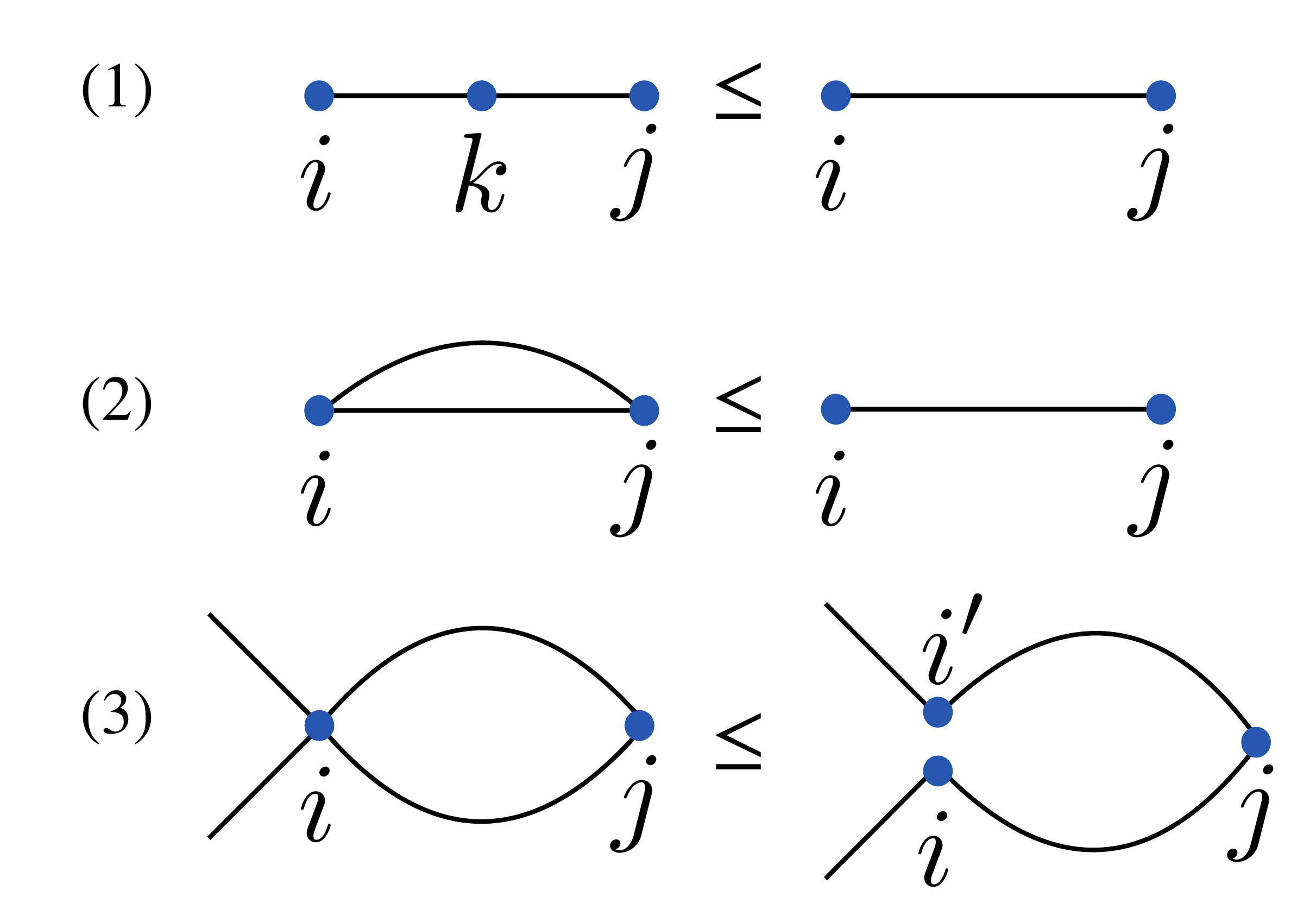}
\caption{Diagrammatic representation of useful inequalities (1): Lemma \ref{chain}, (2): (\ref{prop_1}), and (3): (\ref{prop_2}) }
\label{fig_sup_lemma}
\end{figure}

By using Theorem \ref{cumulantpowerlawtheorem}, Lemma \ref{chain}, and (\ref{cpc_ising_2})-(\ref{prop_2}), we obtain the following upper bound of $\abs{\ave{ \mathcal{J}_n \mathcal{J}_0 }}$:
\begin{align}
    \abs{ \ave{ \mathcal{J}_n \mathcal{J}_0 } }
    &\leq \Big(\frac{Jh}{2} \Big)^2 \sum_{\substack{k \geq n \\ \ell < n}} \sum_{\substack{k' \geq 0 \\ \ell' < 0}} \sum_{ \substack{ \{a,b\}=\{k,\ell\} \\ \{a',b'\}=\{k',\ell'\} }} 
    \delta_{b,b'} \frac{ \abs{ \ave{ S^z_a S^z_{a'} (S^y_b)^2 }_c } + \abs{ \ave{ S^z_a S^z_{a'} }_c \ave{ (S^y_b)^2 }_c } }{ d_{a,b}^\alpha d_{a',b}^\alpha }
    \notag \\
    &\leq \Big(\frac{Jh}{2} \Big)^2 \sum_{\substack{k \geq n \\ \ell < n}} \sum_{\substack{k' \geq 0 \\ \ell' < 0}} \sum_{ \substack{ \{a,b\}=\{k,\ell\} \\ \{a',b'\}=\{k',\ell'\} }} \delta_{b,b'} \frac{ 1 }{ d_{a,b}^\alpha d_{a',b}^\alpha } \Bigg( \frac{ c_1 }{ d_{a,a'}^\alpha d_{a',b}^\alpha } + \frac{ c_1' }{ d_{a,b}^\alpha d_{b,a'}^\alpha } + \frac{c_1''}{ d_{a',a}^\alpha d_{a,b}^\alpha } + \frac{ c_2 }{ d_{a,a'}^\alpha } \Bigg)
    \notag \\
    &\leq \Big(\frac{Jh}{2} \Big)^2 \sum_{\substack{k \geq n \\ \ell < n}} \sum_{\substack{k' \geq 0 \\ \ell' < 0}} \sum_{ \substack{ \{a,b\}=\{k,\ell\} \\ \{a',b'\}=\{k',\ell'\} }} \delta_{b,b'} \frac{ 1 }{ d_{a,b}^\alpha d_{a',b}^\alpha } \Bigg( \frac{ \bar{a} c_1  }{ d_{a,a'}^\alpha } + \sum_{b''} \frac{ c_1' }{ d_{a,b''}^\alpha d_{b'',a'}^\alpha } +\frac{ \bar{a} c_1''  }{ d_{a,a'}^\alpha } + \frac{ c_2 }{ d_{a,a'}^\alpha } \Bigg)
    \notag \\
    &\leq \Big(\frac{Jh}{2} \Big)^2 \{ ( \bar{a} c_1 + \bar{a} c_1'' +c_2) u + c_1' u^2 \} \sum_{\substack{ k \geq n \\ \ell' < 0}} \frac{1}{ d_{k,\ell'}^{2\alpha} }
    \leq C_0 n^{2-2\alpha}
    \label{cal_ising_1},
\end{align}
where $ C_0 =(Jh/2)^2 \{ ( \bar{a} c_1 + \bar{a} c_1'' +c_2) u + c_1' u^2 \}  /\{(1-2\alpha)(2-2\alpha)\}$. 

The diagrammatic derivation of the upper bound of  $\ave{ \mathcal{J}_n \mathcal{J}_0 }$ is illustrated in Fig.\ref{fig_sup_ising_2}, corresponding to (\ref{cal_ising_1}). 
Utilizing (\ref{cpc_ising_1}) and (\ref{cpc_ising_2}), we add the edges connecting between the vertices. By categorizing the cases based on how the cumulant power-law clustering is connected, we identify four terms, as shown in the first line in Fig.\ref{fig_sup_ising_2}, which corresponds to the fourth line in (\ref{cal_ising_1}). Next, by applying (\ref{prop_1}) and (\ref{prop_2}), we obtain the second line in Fig.\ref{fig_sup_ising_2}, which corresponds to the fifth line in (\ref{cal_ising_1}). Finally, by using Lemma \ref{chain}, we form a loop around $k$ and $\ell'$, as depicted in the third line in Fig.\ref{fig_sup_ising_2}, which corresponds to the sixth line in (\ref{cal_ising_1}).

\begin{figure}[h]
\centering
\includegraphics[width=0.55\textwidth]{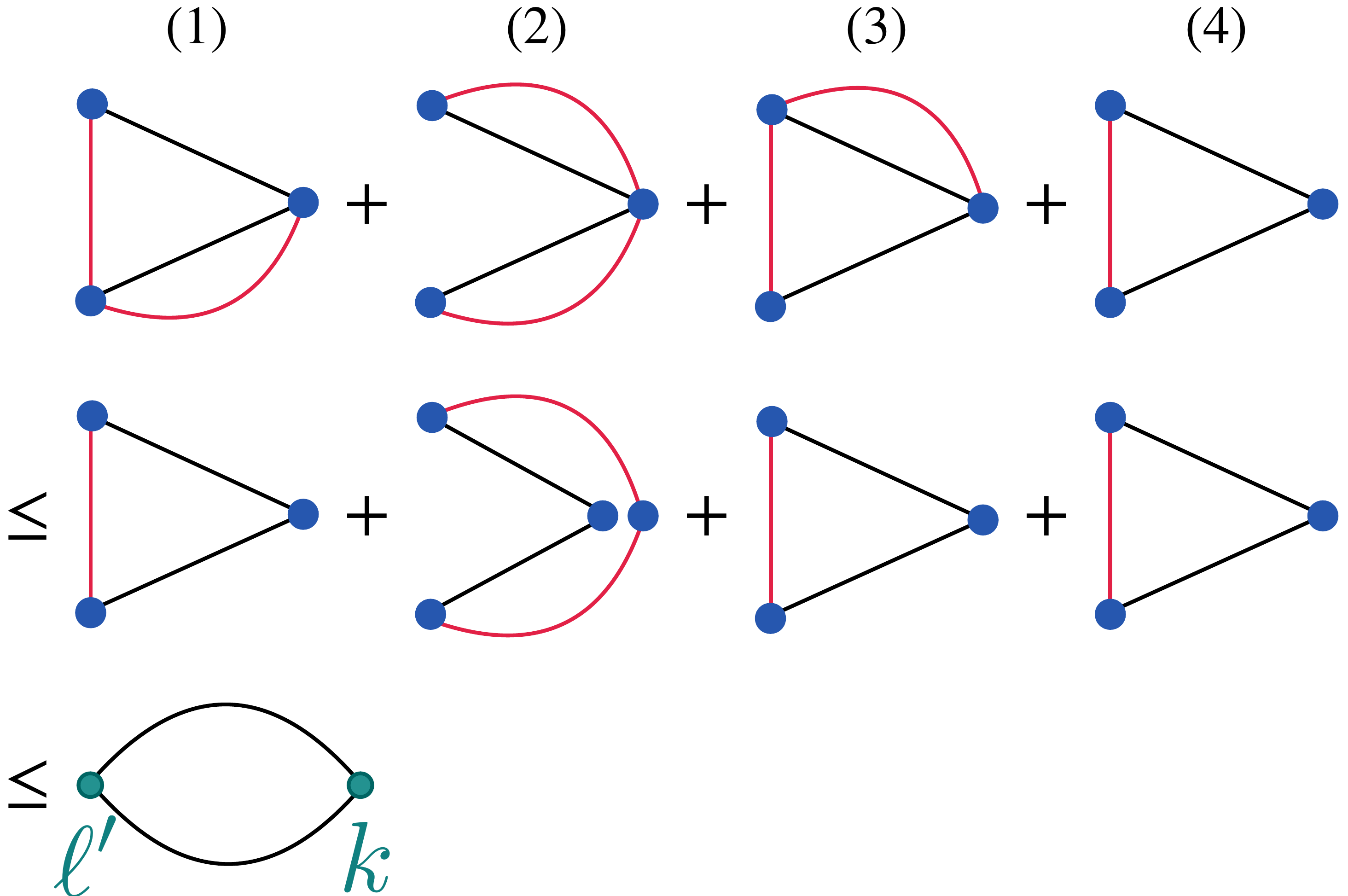}
\caption{The diagrammatic derivation for the upper bound of $\ave{ \mathcal{J}_n \mathcal{J}_0 }$ in (\ref{cal_ising_1})}
\label{fig_sup_ising_2}
\end{figure}

From (\ref{cal_ising_1}) and Figure.\ref{fig_sup_ising_2}, we show the claims of the theorem.    
\end{proof}

\subsection{XY model}

In this subsection, we consider the one-dimensional long-range interacting XY model, which is described by the following Hamiltonian
\begin{equation}
    H= - \sum_{n=1}^N \sum_{r=1}^{N/2} \Big[ J_x \frac{S^x_{n} S^x_{n+r}}{r^\alpha} + J_y \frac{S^y_{n} S^y_{n+r}}{r^\alpha} \Big]
    .
\end{equation}
The local energy at site $n$ is defined as 
\begin{equation}
    \hat{h}_n := - \sum_{r=1}^{N/2} \sum_{\sigma =x,y } \frac{J_\sigma}{2} \Big[ \frac{S^\sigma_{n} S^\sigma_{n+r}}{r^\alpha} + \frac{S^\sigma_{n} S^\sigma_{n-r} }{r^\alpha} \Big]
    ,
\end{equation}
and satisfies the following continuity equations
\begin{align}
    \partial_t \hat{h}_n
    &= - (\mathcal{J}_{n+1} - \mathcal{J}_{n})
    ,
    \\
    \mathcal{J}_n 
    &= - \sum_{\substack{i \geq n , j < n \\ \abs{i-j}<N/2 }} t_{i \leftarrow j}
    .
\end{align}
The energy transmission operator $t_{i \leftarrow j}$ is defined as 
\begin{align}
    t_{i \leftarrow j} 
    &= \{ \hat{h}_i , \hat{h}_j \} 
    = \sum_{m=1}^N \frac{J_{x} J_{y}}{4} \sum_{ \{a,b,c\}=\{i,j,m\} } \tilde{\varepsilon}_{abc} \frac{ S^{x}_{a} S^{y}_{b} S^{z}_{c} }{d_{a,c}^\alpha d_{b,c}^\alpha}
    ,
\end{align}
where $\tilde{\varepsilon}_{abc}$ is defined as follows:
\begin{align}
    \tilde{\varepsilon}_{abc}=
    \begin{cases}
        +1 \, , \quad &\text{for} \, \, 
        (a,b,c)=(i,j,m),(m,j,i),(i,m,j),
        \\
        -1 \, , \quad &\text{for} \, \, 
        (a,b,c)=(j,i,m),(j,m,i),(m,i,j),
    \end{cases}
\end{align}
Note that $\tilde{\varepsilon}_{abc}$ is different from Levi-Civita symbols.

From the symmetry of the Hamiltonian, we note that the following properties hold
\begin{align}
    \ave{ S^x_i } 
    &= \ave{ S^y_i } 
    = \ave{ S^z_i } = 0 \label{xy_prop1}
    ,
    \\
    \ave{ S^\sigma_i S^\tau_j } 
    &= \delta_{\sigma,\tau} \ave{ S^\sigma_i S^\sigma_j } \label{xy_prop2}
    \\
    \ave{ S^z_i S^z_j \mathcal{O} } 
    &= \delta_{i,j} \ave{ (S^z_i)^2 \mathcal{O} } 
    \label{xy_prop3}
    ,
\end{align}
for $i,j=1,...,n$ and $\sigma,\tau \in \{x,y,z\}$. The relation (\ref{xy_prop1}), (\ref{xy_prop2}) is shown from the fact that the Hamiltonian is invariant to the canonical transformation $( S^{\sigma}_i, S^{\tau}_i, S^{\upsilon}_i) \rightarrow (S^{\sigma}_i, -S^{\tau}_i, -S^{\upsilon}_i)$ for $\{\sigma,\tau,\upsilon\}=\{x,y,z\}$ and $i=1,...,N$. The relation (\ref{xy_prop3}) is derived from the absence of components $S^z_i$ in the Hamiltonian.

Using these properties and Theorem \ref{cumulantpowerlawtheorem}, we establish the following theorem:

\begin{theorem} \label{main_xy}
    In the one-dimensional long-range interacting XY model, under the condition $\beta < \beta_c$, the equal-time current correlation is upper-bounded as
    \begin{align}
        \abs{ \ave{ \mathcal{J}_n \mathcal{J}_0 } }
        \leq C_1 n^{2-2\alpha}
        \label{upp-main_xy},
    \end{align}
    for $\alpha>1$.
\end{theorem}

\begin{proof}
The equal-time current correlation $\ave{ \mathcal{J}_n \mathcal{J}_0 }$ is given as follows:
\begin{align}
    \ave{ \mathcal{J}_n \mathcal{J}_0 }
    &= \sum_{\substack{k \geq n \\ \ell < n}} \sum_{\substack{k' \geq 0 \\ \ell' < 0}} \sum_{m,m'=1}^N \frac{ J_{x}^2 J_{y}^2 }{16} \sum_{ \substack{ \{a,b,c\}=\{k,\ell,m\} \\ \{a',b',c'\}=\{k',\ell',m'\} }} \tilde{\varepsilon}_{abc} \tilde{\varepsilon}_{a'b'c'} \frac{ \langle S^{x}_{a} S^{y}_{b} S^{z}_{c} S^{x}_{a'} S^{y}_{b'} S^{z}_{c'} \rangle }{ d_{a,c}^\alpha d_{b,c}^\alpha d_{a',c'}^\alpha d_{b',c'}^\alpha }
    \notag \\
    &= \sum_{\substack{k \geq n \\ \ell < n}} \sum_{\substack{k' \geq 0 \\ \ell' < 0}} \sum_{m,m'=1}^N \frac{ J_{x}^2 J_{y}^2 }{16} \sum_{ \substack{ \{a,b,c\}=\{k,\ell,m\} \\ \{a',b',c'\}=\{k',\ell',m'\} }} \tilde{\varepsilon}_{abc} \tilde{\varepsilon}_{a'b'c'} \delta_{c,c'} \frac{ \langle S^{x}_{a} S^{y}_{b} S^{x}_{a'} S^{y}_{b'} (S^{z}_{c})^2 \rangle }{ d_{a,c}^\alpha d_{b,c}^\alpha d_{a',c}^\alpha d_{b',c}^\alpha } 
    \label{curr_xy}
\end{align}
where we use the relation (\ref{xy_prop3}). 
The term $1 / ( d_{a,c}^\alpha d_{b,c}^\alpha d_{a',c}^\alpha d_{b',c}^\alpha ) $ in (\ref{curr_xy}) is illustrated diagrammatically in Fig.\ref{fig_sup_xy_1}.

\begin{figure}[h]
\centering
\includegraphics[width=0.2\textwidth]{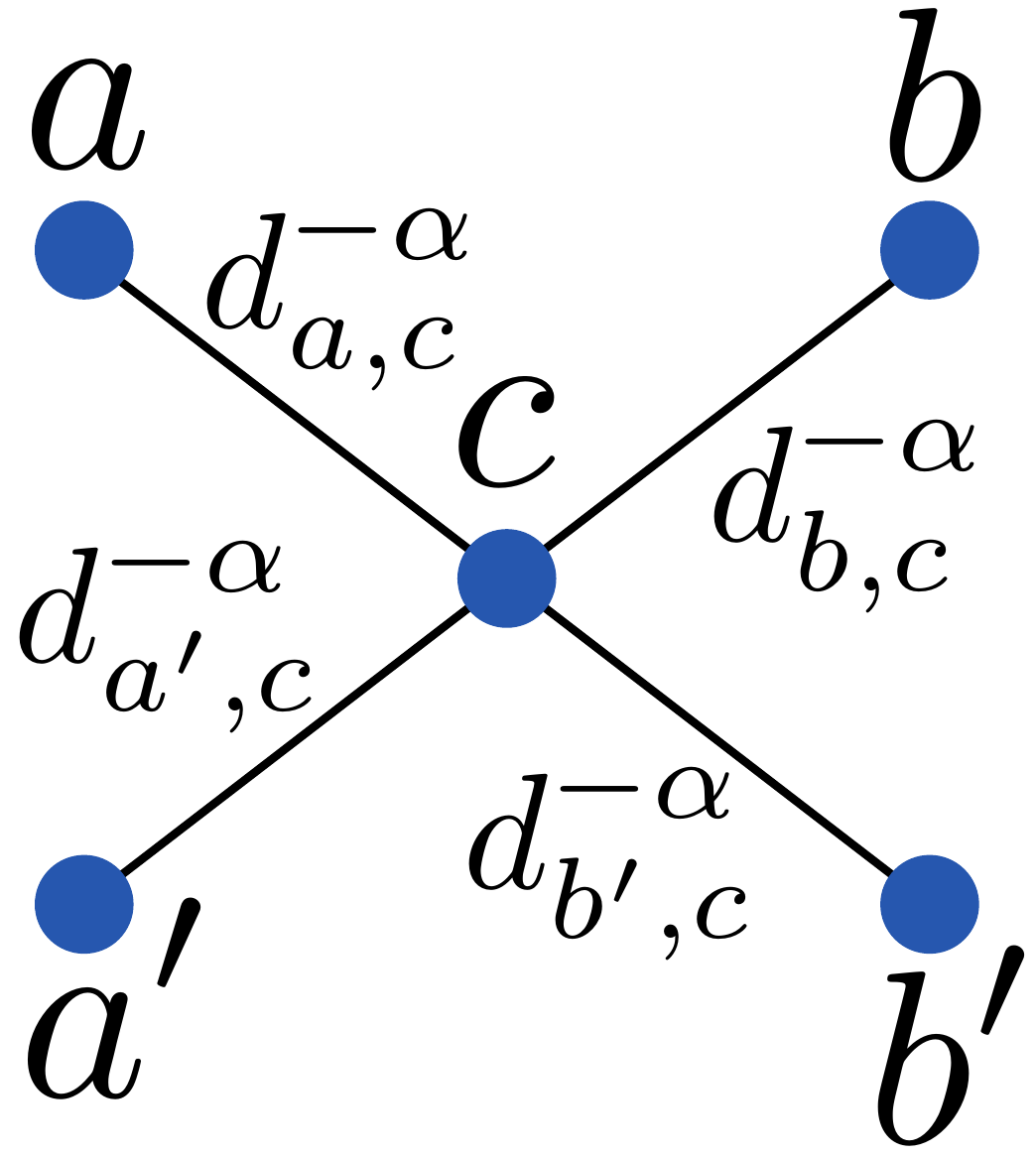}
\caption{The diagrammatic representation of the term $d_{a,c}^{-\alpha} d_{b,c}^{-\alpha} d_{a',c}^{-\alpha} d_{b',c}^{-\alpha} $ in (\ref{curr_xy}) }
\label{fig_sup_xy_1}
\end{figure}

The 5-body correlation $\langle S^{x}_{a} S^{y}_{b} S^{x}_{a'} S^{y}_{b'} (S^{z}_{c})^2 \rangle$ can be decomposed into cumulant products as follows:
\begin{align}
    \langle S^{x}_{a} S^{y}_{b} S^{x}_{a'} S^{y}_{b'} (S^{z}_{c})^2 \rangle
    &= \langle S^{x}_{a} S^{y}_{b} S^{x}_{a'} S^{y}_{b'} (S^{z}_{c})^2 \rangle_c 
    + \ave{ S^{x}_{a} S^{y}_{b} S^{x}_{a'} S^{y}_{b'} }_c \ave{ (S^z_c)^2 } 
    \notag \\ & 
    + \ave{ S^x_{a} S^x_{a'} (S^z_{c})^2 }_c \ave{ S^y_{b} S^y_{b'} }_c
    + \ave{ S^x_{a} S^x_{a'} }_c \ave{ S^y_{b} S^y_{b'} (S^z_{c})^2}_c
    + \ave{ S^x_{a} S^x_{a'} }_c \ave{ S^y_{b} S^y_{b'} }_c \ave{ (S^z_{c})^2 }
\end{align}
where we use the relation (\ref{xy_prop1}), (\ref{xy_prop2}). Next, we use Theorem \ref{cumulantpowerlawtheorem} as follows:
\begin{align}
    \langle S^{x}_{a} S^{y}_{b} S^{x}_{a'} S^{y}_{b'} (S^{z}_{c})^2 \rangle_c  
    &\leq c_1 \|S^{x}_{a} \| 
    ... \| (S^{z}_{c})^2 \| \sum_{ \substack{ \{j_1,...,j_5\}= \\ \{ a,b,a',b',c \} \\ \mathrm{:connected} } } \frac{1}{ d_{j_1,j_2}^\alpha ... d_{j_4,j_5}^\alpha }
    , \label{cpc_xy_1}
    \\
    \langle S^{x}_{a} S^{y}_{b} S^{x}_{a'} S^{y}_{b'} \rangle_c  
    &\leq c_1 \|S^{x}_{a} \| 
    ... \| S^{y}_{b'} \| \sum_{ \substack{ \{j_1,...,j_4\}= \\ \{a,b,a',b'\} \\ \mathrm{:connected} } } \frac{1}{ d_{j_1,j_2}^\alpha d_{j_2,j_3}^\alpha d_{j_3,j_4}^\alpha }
    , \label{cpc_xy_2}
    \\
    \ave{ S^\sigma_{i_1} S^\sigma_{i_2} (S^z_c)^2 }_c
    &\leq c_3 \| S^\sigma_{i_1} \| \| S^\sigma_{i_2} \| \| (S^z_c)^2 \| \sum_{\substack{ \{j_1,j_2,j_3\}= \\ \{i_1,i_2,c\} \\ \mathrm{:connected} }} \frac{1}{ d_{j_1,j_2}^\alpha d_{j_2,j_3}^\alpha }
    , \label{cpc_xy_3}
    \\
    \langle S^\sigma_{i_p} S^\sigma_{i_q} \rangle_c 
    &\leq c_4 \| S^\sigma_{i_p} \| \| S^\sigma_{i_q} \| \frac{1}{ d_{i_p,i_q}^\alpha }
    , \label{cpc_xy_4}
\end{align}
where $\sum_{\substack{ \{j_1,...,j_m\} \\ \mathrm{:connected}}}$ denotes taking the summation over the vertices, which are connected from $j_1$ to $j_m$ along the Hamilton path.
By using (\ref{cpc_xy_1}) - (\ref{cpc_xy_4}), we obtain the following inequality for $\abs{\ave{ \mathcal{J}_n \mathcal{J}_0 }}$:
\begin{align}
    \abs{ \ave{ \mathcal{J}_n \mathcal{J}_0 } }
    &\leq \sum_{\substack{k \geq n \\ \ell < n}} \sum_{\substack{k' \geq 0 \\ \ell' < 0}} \sum_{m,m'=1}^N \frac{ J_{x}^2 J_{y}^2 }{16} \!\!\! \sum_{ \substack{ \{a,b,c\}=\{k,\ell,m\} \\ \{a',b',c'\}=\{k',\ell',m'\} }}  \!\!\! \frac{ 
    \delta_{c,c'} }{ d_{a,c}^\alpha d_{b,c}^\alpha d_{a',c}^\alpha d_{b',c}^\alpha }
    \bigg\{ \abs{ \langle S^{x}_{a} S^{y}_{b} S^{x}_{a'} S^{y}_{b'} (S^{z}_{c})^2 \rangle_c } 
    + \abs{ \ave{ S^{x}_{a} S^{y}_{b} S^{x}_{a'} S^{y}_{b'} }_c \ave{ (S^z_c)^2 } }
    \notag \\ & \hspace{10pt}
    + \abs{ \ave{ S^x_{a} S^x_{a'} (S^z_{c})^2 }_c \ave{ S^y_{b} S^y_{b'} }_c }
    + \abs{ \ave{ S^x_{a} S^x_{a'} }_c \ave{ S^y_{b} S^y_{b'} (S^z_{c})^2}_c }
    + \abs{ \ave{ S^x_{a} S^x_{a'} }_c \ave{ S^y_{b} S^y_{b'} }_c \ave{ (S^z_{c})^2 } }
    \bigg\}
    \notag \\ &
    \leq \sum_{\substack{k \geq n \\ \ell < n}} \sum_{\substack{k' \geq 0 \\ \ell' < 0}} \sum_{m,m'=1}^N \frac{ J_{x}^2 J_{y}^2 }{16} \!\!\! \sum_{ \substack{ \{a,b,c\}=\{k,\ell,m\} \\ \{a',b',c'\}=\{k',\ell',m'\} }} \!\!\! \frac{ \delta_{c,c'} }{ d_{a,c}^\alpha d_{b,c}^\alpha d_{a',c}^\alpha d_{b',c}^\alpha }
    \bigg\{ \sum_{\substack{ \{i_1,...i_5\}= \\ \{a,b,a',b',c\} \\ \mathrm{:connected}}} \frac{ c_1 }{ d_{i_1,i_2}^\alpha ... d_{i_4,i_5}^\alpha } + \sum_{\substack{ \{i_1,...i_4\}= \\ \{a,b,a',b'\} \\ \mathrm{:connected}}} \frac{ c_2 }{ d_{i_1,i_2}^\alpha d_{i_2,i_3}^\alpha d_{i_3,i_4}^\alpha } 
    \notag \\ & \hspace{10pt}
    + \sum_{ \substack{ \{i_1,...,i_4\} = \{a,b,a',b'\} \\ \{i_1,i_2\}=\{a,a'\},\{b,b'\} }} \sum_{\substack{ \{j_1,j_2,j_3\}= \\ \{ i_1,i_2,c\} \\ \mathrm{:connected} }} \frac{2 c_3}{ d_{j_1,j_2}^\alpha d_{j_2,j_3}^\alpha } \frac{ c_4 }{ d_{i_3,i_4}^\alpha } + \frac{ (c_4)^2 }{ d_{a,a'}^\alpha d_{b,b'}^\alpha } \bigg\}
    \label{cal_xy_1}
\end{align}

Let us first consider the contribution of the first term in (\ref{cal_xy_1}): 
\begin{align}
    \sum_{\substack{k \geq n \\ \ell < n}} \sum_{\substack{k' \geq 0 \\ \ell' < 0}} \sum_{m,m'=1}^N \frac{ J_{x}^2 J_{y}^2 }{16} \sum_{ \substack{ \{a,b,c\}=\{k,\ell,m\} \\ \{a',b',c'\}=\{k',\ell',m'\} }} \frac{ \delta_{c,c'} }{ d_{a,c}^\alpha d_{b,c}^\alpha d_{a',c}^\alpha d_{b',c}^\alpha } \sum_{\substack{ \{i_1,...i_5\}= \\ \{a,b,a',b',c\} \\ \mathrm{:connected}}} \frac{ c_1 }{ d_{i_1,i_2}^\alpha ... d_{i_4,i_5}^\alpha }
    \label{cal_xy_1_1}
\end{align}

Figure.\ref{fig_sup_xy_2} enumerates all possible cases of clustering linkages in (\ref{cal_xy_1_1}) and illustrates the derivation of the upper bound for each case. 
Here, diagrams with identical topology are regarded as equivalent.

\begin{figure}[h]
\centering
\includegraphics[width=0.6\textwidth]{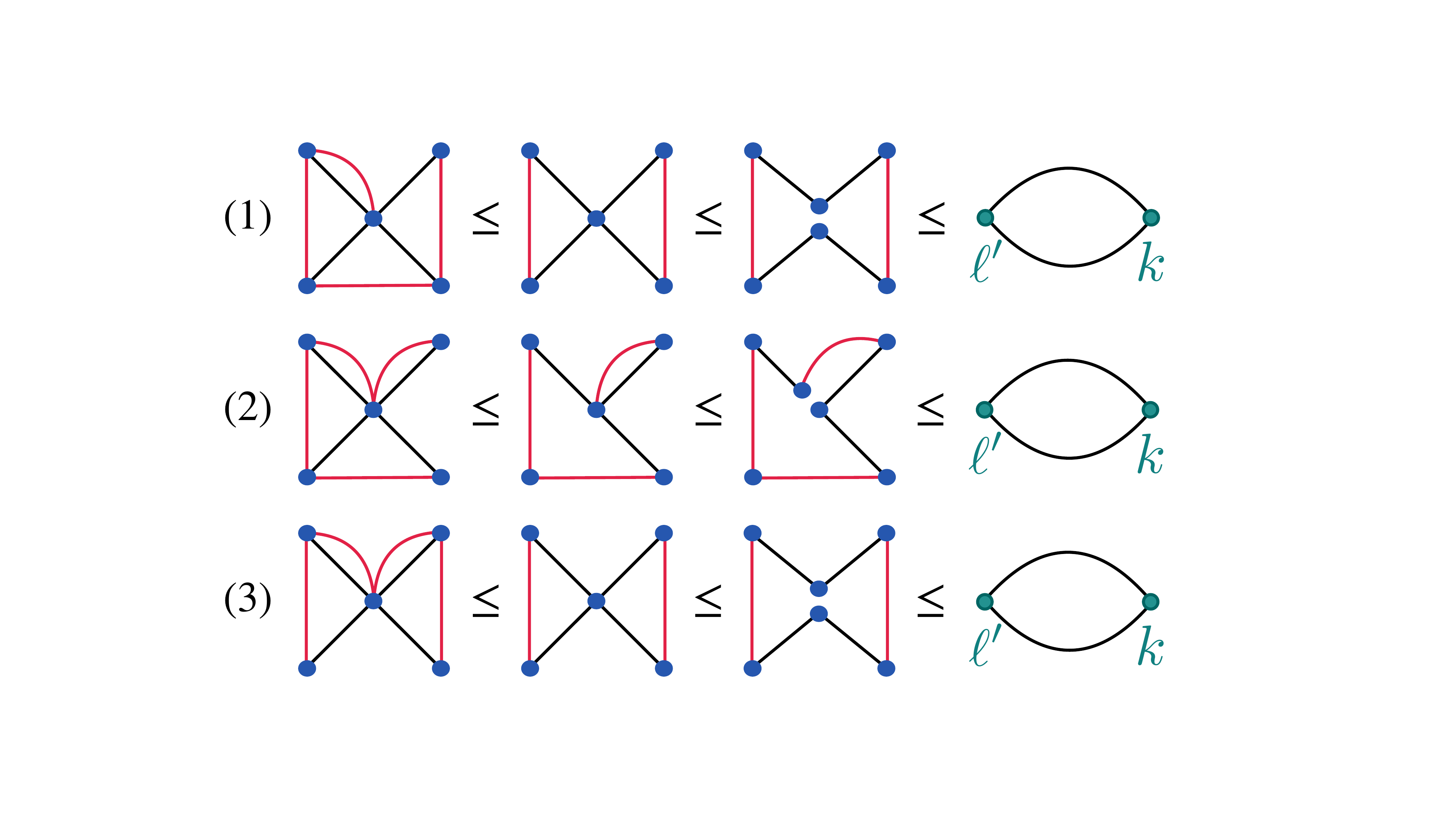}
\caption{The diagrammatic derivation for the upper bound of the contribution of the first term in (\ref{cal_xy_1})}
\label{fig_sup_xy_2}
\end{figure}

Corresponding to the case (1) in Fig.\ref{fig_sup_xy_2}, we obtain the following upper bound of the contribution to $\ave{ \mathcal{J}_n \mathcal{J}_0 }$:
\begin{align}
    &\sum_{\substack{k \geq n \\ \ell < n}} \sum_{\substack{k' \geq 0 \\ \ell' < 0}} \sum_{m,m'=1}^N \frac{ J_{x}^2 J_{y}^2 }{16} \sum_{ \substack{ \{a,b,c\}=\{k,\ell,m\} \\ \{a',b',c'\}=\{k',\ell',m'\} }} \frac{ \delta_{c,c'} }{ d_{a,c}^\alpha d_{b,c}^\alpha d_{a',c}^\alpha d_{b',c}^\alpha } \sum_{\substack{ \{i_2,...i_5\}= \\ \{a,b,a',b'\} \\ \mathrm{:connected}}} \frac{ c_1 }{ d_{c,i_2}^\alpha d_{i_2,i_3}^\alpha ... d_{i_4,i_5}^\alpha }
    \notag \\ &
    \leq \sum_{\substack{k \geq n \\ \ell < n}} \sum_{\substack{k' \geq 0 \\ \ell' < 0}} \sum_{m,m'=1}^N \frac{ J_{x}^2 J_{y}^2 }{16} \sum_{ \substack{ \{a,b,c\}=\{k,\ell,m\} \\ \{a',b',c'\}=\{k',\ell',m'\} }} \frac{ \delta_{c,c'} }{ d_{a,c}^\alpha d_{b,c}^\alpha d_{a',c}^\alpha d_{b',c}^\alpha } \sum_{\substack{ \{i_2,...i_5\}= \\ \{a,b,a',b'\} \\ \mathrm{:connected}}} \frac{ c_1 }{ d_{i_2,i_3}^\alpha d_{i_4,i_5}^\alpha }
    \notag \\ &
    \leq \sum_{\substack{k \geq n \\ \ell < n}} \sum_{\substack{k' \geq 0 \\ \ell' < 0 }} \sum_{m,m'=1}^N \frac{ J_{x}^2 J_{y}^2 }{16} \sum_{ \substack{ \{a,b,c\}=\{k,\ell,m\} \\ \{a',b',c'\}=\{k',\ell',m'\} }} \sum_{c''=1}^N \sum_{\substack{ \{i_2,...i_5\}= \\ \{a,b,a',b'\} \\ \mathrm{:connected} }} \frac{ \delta_{c,c'}  c_1 }{ d_{i_2,i_3}^\alpha d_{i_3,c}^\alpha d_{c,i_4}^\alpha d_{i_4,i_5}^\alpha d_{i_5,c''}^\alpha d_{c'',i_2}^\alpha }
    \notag \\ &
    \leq \frac{ J_{x}^2 J_{y}^2 }{16} u^4  c_1 \sum_{\substack{ k \geq n \\ \ell' < 0 }} \frac{1}{ d_{k,\ell'}^{2\alpha} }
    \leq c^{(1)} n^{2-2\alpha}
    ,
\end{align}
where $c^{(1)} = (J_x^2 J_y^2/16) u^4  c_1 /\{ (1-2\alpha)(2-2\alpha) \}$. 

Corresponding to the case (2) in Fig.\ref{fig_sup_xy_2}, we similarly obtain the following upper bound of the contribution to $\ave{ \mathcal{J}_n \mathcal{J}_0 }$:
\begin{align}
    &\sum_{\substack{k \geq n \\ \ell < n}} \sum_{\substack{k' \geq 0 \\ \ell' < 0}} \sum_{m,m'=1}^N \frac{ J_{x}^2 J_{y}^2 }{16} \sum_{ \substack{ \{a,b,c\}=\{k,\ell,m\} \\ \{a',b',c'\}=\{k',\ell',m'\} }} \frac{ \delta_{c,c'} }{ d_{a,c}^\alpha d_{b,c}^\alpha d_{a',c}^\alpha d_{b',c}^\alpha } \sum_{\substack{ \{i_1,i_3,i_4,i_5\}= \\ \{a,b,a',b'\} \\ \mathrm{:connected}}} \frac{ c_1 }{ d_{i_1,c}^\alpha d_{c,i_3}^\alpha d_{i_3,i_4}^\alpha d_{i_4,i_5}^\alpha }
    \notag \\ &
    \leq \sum_{\substack{k \geq n \\ \ell < n}} \sum_{\substack{k' \geq 0 \\ \ell' < 0}} \sum_{m,m'=1}^N \frac{ J_{x}^2 J_{y}^2 }{16} \sum_{ \substack{ \{a,b,c\}=\{k,\ell,m\} \\ \{a',b',c'\}=\{k',\ell',m'\} }} \sum_{\substack{ \{i_1,i_3,i_4,i_5\}= \\ \{a,b,a',b'\} \\ \mathrm{:connected}}} \frac{  \delta_{c,c'} }{ d_{i_1,c}^\alpha d_{i_2,c}^\alpha d_{i_4,c}^\alpha d_{i_5,c}^\alpha } \frac{  c_1 }{ d_{i_1,c}^\alpha d_{i_3,i_4}^\alpha d_{i_4,i_5}^\alpha }
    \notag \\ &
    \leq \sum_{\substack{k \geq n \\ \ell < n}} \sum_{\substack{k' \geq 0 \\ \ell' < 0}} \sum_{m,m'=1}^N \frac{ J_{x}^2 J_{y}^2 }{16} \sum_{ \substack{ \{a,b,c\}=\{k,\ell,m\} \\ \{a',b',c'\}=\{k',\ell',m'\} }} \sum_{c''=1}^N \sum_{\substack{ \{i_1,i_3,i_4,i_5\}= \\ \{a,b,a',b'\} \\ \mathrm{:connected}}} \frac{ \delta_{c,c'}  c_1 }{ d_{i_1,c}^\alpha d_{c,i_2}^\alpha d_{i_2,i_3}^\alpha d_{i_3,i_4}^\alpha d_{i_4,c''}^\alpha d_{c'',i_1}^\alpha } 
    \notag \\ &
    \leq \frac{ J_{x}^2 J_{y}^2 }{16} u^4  c_1 \sum_{\substack{ k \geq n \\ \ell' < 0 }} \frac{1}{ d_{k,\ell'}^{2\alpha} }
    \leq c^{(1)} n^{2-2\alpha}
    .
\end{align}

Corresponding to the case (3) in Fig.\ref{fig_sup_xy_2}, we similarly obtain the following upper bound of the contribution to $\ave{ \mathcal{J}_n \mathcal{J}_0 }$:
\begin{align}
    &\sum_{\substack{k \geq n \\ \ell < n}} \sum_{\substack{k' \geq 0 \\ \ell' < 0}} \sum_{m,m'=1}^N \frac{ J_{x}^2 J_{y}^2 }{16} \sum_{ \substack{ \{a,b,c\}=\{k,\ell,m\} \\ \{a',b',c'\}=\{k',\ell',m'\} }} \frac{ \delta_{c,c'} }{ d_{a,c}^\alpha d_{b,c}^\alpha d_{a',c}^\alpha d_{b',c}^\alpha } \sum_{\substack{ \{i_1,i_2,i_4,i_5\}= \\ \{a,b,a',b'\} \\ \mathrm{:connected}}} \frac{ c_1 }{ d_{i_1,i_2}^\alpha d_{i_2,c}^\alpha d_{c,i_4}^\alpha d_{i_4,i_5}^\alpha }
    \notag \\ &
    \leq \sum_{\substack{k \geq n \\ \ell < n}} \sum_{\substack{k' \geq 0 \\ \ell' < 0}} \sum_{m,m'=1}^N \frac{ J_{x}^2 J_{y}^2 }{16} \sum_{ \substack{ \{a,b,c\}=\{k,\ell,m\} \\ \{a',b',c'\}=\{k',\ell',m'\} }} \frac{ \delta_{c,c'} }{ d_{a,c}^\alpha d_{b,c}^\alpha d_{a',c}^\alpha d_{b',c}^\alpha } \sum_{\substack{ \{i_1,i_2,i_4,i_5\}= \\ \{a,b,a',b'\} \\ \mathrm{:connected}}} \frac{  c_1 }{ d_{i_1,i_2}^\alpha d_{i_4,i_5}^\alpha }
    \notag \\ &
    \leq \sum_{\substack{k \geq n \\ \ell < n}} \sum_{\substack{k' \geq 0 \\ \ell' < 0}} \sum_{m,m'=1}^N \frac{ J_{x}^2 J_{y}^2 }{16} \sum_{ \substack{ \{a,b,c\}=\{k,\ell,m\} \\ \{a',b',c'\}=\{k',\ell',m'\} }} \sum_{c''=1}^N \sum_{\substack{ \{i_1,i_2,i_4,i_5\}= \\ \{a,b,a',b'\} \\ \mathrm{:connected}}} \frac{ \delta_{c,c'}  c_1 }{ d_{i_1,i_2}^\alpha d_{i_2,c}^\alpha d_{c,i_4}^\alpha d_{i_4,i_5}^\alpha d_{i_5,c''}^\alpha d_{c'',i_1}^\alpha }
    \notag \\ &
    \leq \frac{ J_{x}^2 J_{y}^2 }{16} u^4  c_1 \sum_{\substack{ k \geq n \\ \ell' < 0 }} \frac{1}{ d_{k,\ell'}^{2\alpha} }
    \leq c^{(1)} n^{2-2\alpha}
    .
\end{align}

From the above case analysis, we obtain the following upper bound for the contribution from the first term in Fig.\ref{fig_sup_xy_2}:
\begin{align}
    \sum_{\substack{k \geq n \\ \ell < n}} \sum_{\substack{k' \geq 0 \\ \ell' < 0}} \sum_{m,m'=1}^N \frac{ J_{x}^2 J_{y}^2 }{16} \sum_{ \substack{ \{a,b,c\}=\{k,\ell,m\} \\ \{a',b',c'\}=\{k',\ell',m'\} }} \frac{ \delta_{c,c'} }{ d_{a,c}^\alpha d_{b,c}^\alpha d_{a',c}^\alpha d_{b',c}^\alpha } \sum_{\substack{ \{i_1,...i_5\}= \\ \{a,b,a',b',c\} \\ \mathrm{:connected}}} \frac{ c_1 }{ d_{i_1,i_2}^\alpha ... d_{i_4,i_5}^\alpha } \leq 3 c^{(1)} n^{2-2\alpha}
    \label{cal_xy_2}
\end{align}

Next, let us consider the contribution of the second term in (\ref{cal_xy_1}): 
\begin{align}
    \sum_{\substack{k \geq n \\ \ell < n}} \sum_{\substack{k' \geq 0 \\ \ell' < 0}} \sum_{m,m'=1}^N \frac{ J_{x}^2 J_{y}^2 }{16} \sum_{ \substack{ \{a,b,c\}=\{k,\ell,m\} \\ \{a',b',c'\}=\{k',\ell',m'\} }} \frac{ \delta_{c,c'} }{ d_{a,c}^\alpha d_{b,c}^\alpha d_{a',c}^\alpha d_{b',c}^\alpha } \sum_{\substack{ \{i_1,...i_4\}= \\ \{a,b,a',b'\} \\ \mathrm{:connected}}} \frac{ c_2 }{ d_{i_1,i_2}^\alpha d_{i_2,i_3}^\alpha d_{i_3,i_4}^\alpha }
    \label{cal_xy_1_2}
\end{align}

Figure.\ref{fig_sup_xy_3} illustrates the method of clustering linkage in (\ref{cal_xy_1_2}), which, when considering graphs with equivalent topology as identical, is restricted to a single case. Furthermore, it shows the derivation of the upper bound for this case.

\begin{figure}[h]
\centering
\includegraphics[width=0.6\textwidth]{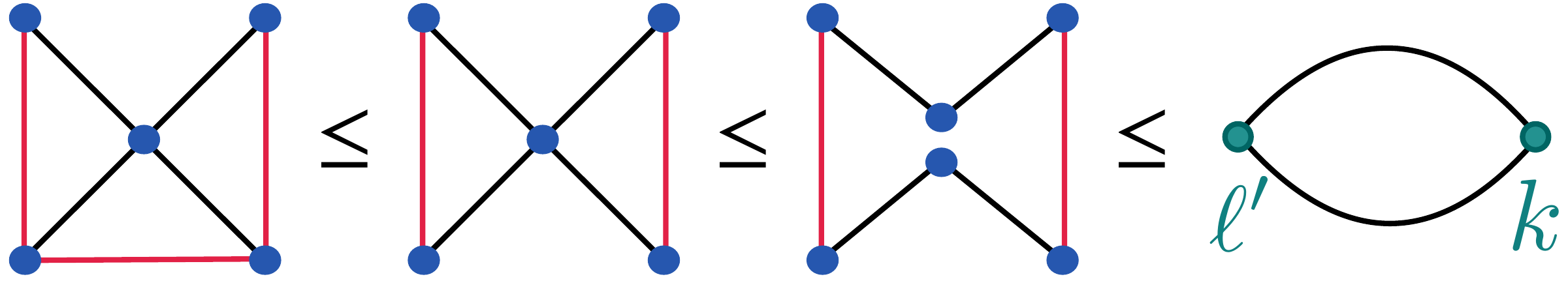}
\caption{The diagrammatic derivation for the upper bound of the contribution of the second term in (\ref{cal_xy_1})}
\label{fig_sup_xy_3}
\end{figure}

In this case, we obtain the following upper bound of the contribution to $\ave{ \mathcal{J}_n \mathcal{J}_0 }$:
\begin{align}
    &\sum_{\substack{k \geq n \\ \ell < n}} \sum_{\substack{k' \geq 0 \\ \ell' < 0}} \sum_{m,m'=1}^N \frac{ J_{x}^2 J_{y}^2 }{16} \sum_{ \substack{ \{a,b,c\}=\{k,\ell,m\} \\ \{a',b',c'\}=\{k',\ell',m'\} }} \frac{ \delta_{c,c'} }{ d_{a,c}^\alpha d_{b,c}^\alpha d_{a',c}^\alpha d_{b',c}^\alpha } \sum_{\substack{ \{i_1,...i_4\}= \\ \{a,b,a',b'\} \\ \mathrm{:connected}}} \frac{ c_2 }{ d_{i_1,i_2}^\alpha d_{i_2,i_3}^\alpha d_{i_3,i_4}^\alpha }
    \notag \\ &
    \leq \sum_{\substack{k \geq n \\ \ell < n}} \sum_{\substack{k' \geq 0 \\ \ell' < 0}} \sum_{m,m'=1}^N \frac{ J_{x}^2 J_{y}^2 }{16} \sum_{ \substack{ \{a,b,c\}=\{k,\ell,m\} \\ \{a',b',c'\}=\{k',\ell',m'\} }} \frac{ \delta_{c,c'} }{ d_{a,c}^\alpha d_{b,c}^\alpha d_{a',c}^\alpha d_{b',c}^\alpha } \sum_{\substack{ \{i_1,...i_4\}= \\ \{a,b,a',b'\} \\ \mathrm{:connected}}} \frac{  c_2 }{ d_{i_1,i_2}^\alpha d_{i_3,i_4}^\alpha }
    \notag \\ &
    \leq \sum_{\substack{k \geq n \\ \ell < n}} \sum_{\substack{k' \geq 0 \\ \ell' < 0}} \sum_{m,m'=1}^N \frac{ J_{x}^2 J_{y}^2 }{16} \sum_{ \substack{ \{a,b,c\}=\{k,\ell,m\} \\ \{a',b',c'\}=\{k',\ell',m'\} }} \sum_{c''=1}^N \sum_{\substack{ \{i_1,...i_4\}= \\ \{a,b,a',b'\} \\ \mathrm{:connected}}} \frac{ \delta_{c,c'}  c_2 }{ d_{i_1,i_2}^\alpha d_{i_2,c}^\alpha d_{c,i_3}^\alpha d_{i_3,i_4}^\alpha d_{i_4,c''}^\alpha d_{c'',i_1}^\alpha }
    \notag \\ &
    \leq \frac{ J_{x}^2 J_{y}^2 }{16} u^4  c_2 \sum_{\substack{ k \geq n \\ \ell' < 0 }} \frac{1}{ d_{k,\ell'}^{2\alpha} }
    \leq c^{(2)} n^{2-2\alpha}
    \label{cal_xy_3}
    ,
\end{align}
where $c^{(2)} = (J_x^2 J_y^2/16) u^4  c_1 /\{ (1-2\alpha)(2-2\alpha) \}$. 

Next, let us consider the contribution of the third term in (\ref{cal_xy_1}): 
\begin{align}
    \sum_{\substack{k \geq n \\ \ell < n}} \sum_{\substack{k' \geq 0 \\ \ell' < 0}} \sum_{m,m'=1}^N \frac{ J_{x}^2 J_{y}^2 }{16} \sum_{ \substack{ \{a,b,c\}=\{k,\ell,m\} \\ \{a',b',c'\}=\{k',\ell',m'\} }} \frac{ \delta_{c,c'} }{ d_{a,c}^\alpha d_{b,c}^\alpha d_{a',c}^\alpha d_{b',c}^\alpha } \sum_{ \substack{ \{i_1,...,i_4\} = \{a,b,a',b'\} \\ \{i_1,i_2\}=\{a,a'\},\{b,b'\} }} \sum_{\substack{ \{j_1,j_2,j_3\}= \\ \{ i_1,i_2,c\} \\ \mathrm{:connected} }} \frac{2 c_3}{ d_{j_1,j_2}^\alpha d_{j_2,j_3}^\alpha } \frac{ c_4 }{ d_{i_3,i_4}^\alpha }
    \label{cal_xy_1_3}
\end{align}

Figure.\ref{fig_sup_xy_4} illustrates the method of clustering linkage in (\ref{cal_xy_1_3}), which, when considering graphs with equivalent topology as identical, is restricted to a single case. Furthermore, it shows the derivation of the upper bound for this case.

\begin{figure}[h]
\centering
\includegraphics[width=0.6\textwidth]{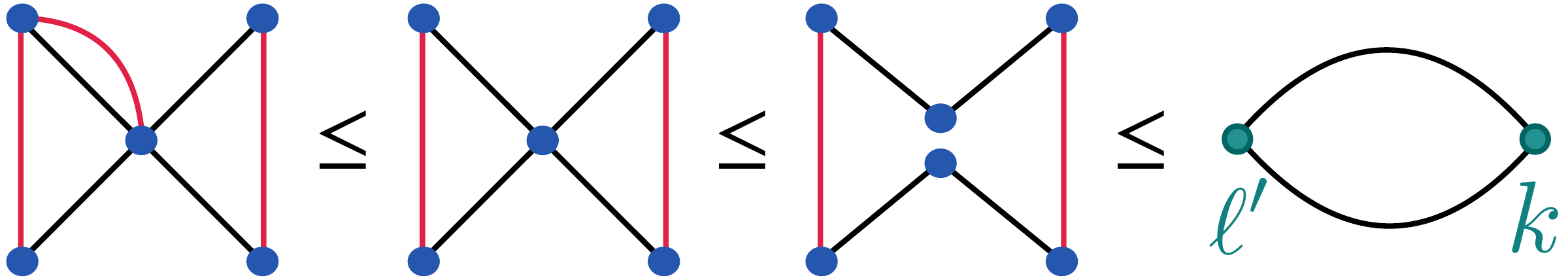}
\caption{The diagrammatic derivation for the upper bound of the contribution of the third term in (\ref{cal_xy_1})}
\label{fig_sup_xy_4}
\end{figure}

In this case, we similarly obtain the following upper bound of the contribution to $\ave{ \mathcal{J}_n \mathcal{J}_0 }$:
\begin{align}
    &\sum_{\substack{k \geq n \\ \ell < n}} \sum_{\substack{k' \geq 0 \\ \ell' < 0}} \sum_{m,m'=1}^N \frac{ J_{x}^2 J_{y}^2 }{16} \sum_{ \substack{ \{a,b,c\}=\{k,\ell,m\} \\ \{a',b',c'\}=\{k',\ell',m'\} }} \frac{ \delta_{c,c'} }{ d_{a,c}^\alpha d_{b,c}^\alpha d_{a',c}^\alpha d_{b',c}^\alpha } \sum_{ \substack{ \{i_1,...,i_4\} = \{a,b,a',b'\} \\ \{i_1,i_2\}=\{a,a'\},\{b,b'\} }} \sum_{\substack{ \{j_1,j_2,j_3\}= \\ \{ i_1,i_2,c\} \\ \mathrm{:connected} }} \frac{2 c_3}{ d_{j_1,j_2}^\alpha d_{j_2,j_3}^\alpha } \frac{ c_4 }{ d_{i_3,i_4}^\alpha }
    \notag \\ &
    \leq \sum_{\substack{k \geq n \\ \ell < n}} \sum_{\substack{k' \geq 0 \\ \ell' < 0}} \sum_{m,m'=1}^N \frac{ J_{x}^2 J_{y}^2 }{16} \sum_{ \substack{ \{a,b,c\}=\{k,\ell,m\} \\ \{a',b',c'\}=\{k',\ell',m'\} }} \frac{ \delta_{c,c'} }{ d_{a,c}^\alpha d_{b,c}^\alpha d_{a',c}^\alpha d_{b',c}^\alpha } \sum_{ \substack{ \{i_1,...,i_4\} = \{a,b,a',b'\} \\ \{i_1,i_2\}=\{a,a'\},\{b,b'\} }}  \frac{2  c_3}{ d_{i_1,i_2}^\alpha } \frac{ c_4 }{ d_{i_3,i_4}^\alpha }
    \notag \\ &
    \leq \sum_{\substack{k \geq n \\ \ell < n}} \sum_{\substack{k' \geq 0 \\ \ell' < 0}} \sum_{m,m'=1}^N \frac{ J_{x}^2 J_{y}^2 }{16} \sum_{ \substack{ \{a,b,c\}=\{k,\ell,m\} \\ \{a',b',c'\}=\{k',\ell',m'\} }} \sum_{c''=1}^N \sum_{\substack{ \{i_1,...i_4\}= \\ \{a,b,a',b'\} \\ \mathrm{:connected}}} \frac{ \delta_{c,c'}  c_3 c_4 }{ d_{i_1,i_2}^\alpha d_{i_2,c}^\alpha d_{c,i_3}^\alpha d_{i_3,i_4}^\alpha d_{i_4,c''}^\alpha d_{c'',i_1}^\alpha }
    \notag \\ &
    \leq \frac{ J_{x}^2 J_{y}^2 }{16} u^4  c_3 c_4 \sum_{\substack{ k \geq n \\ \ell' < 0 }} \frac{1}{ d_{k,\ell'}^{2\alpha} }
    \leq c^{(3)} n^{2-2\alpha}
    \label{cal_xy_4}
    ,
\end{align}
where $c^{(3)} = (J_x^2 J_y^2/16) u^4  c_3 c_4 /\{ (1-2\alpha)(2-2\alpha) \}$.

Lastly, let us consider the contribution of the fourth term in (\ref{cal_xy_1}): 
\begin{align}
    \sum_{\substack{k \geq n \\ \ell < n}} \sum_{\substack{k' \geq 0 \\ \ell' < 0}} \sum_{m,m'=1}^N \frac{ J_{x}^2 J_{y}^2 }{16} \sum_{ \substack{ \{a,b,c\}=\{k,\ell,m\} \\ \{a',b',c'\}=\{k',\ell',m'\} }} \frac{ \delta_{c,c'} }{ d_{a,c}^\alpha d_{b,c}^\alpha d_{a',c}^\alpha d_{b',c}^\alpha } \frac{ (c_4)^2 }{ d_{a,a'}^\alpha d_{b,b'}^\alpha } 
    \label{cal_xy_1_4}
\end{align}

Figure.\ref{fig_sup_xy_5} illustrates the method of clustering linkage in (\ref{cal_xy_1_4}), which, when considering graphs with equivalent topology as identical, is restricted to a single case. Furthermore, it shows the derivation of the upper bound for this case.

\begin{figure}[h]
\centering
\includegraphics[width=0.45\textwidth]{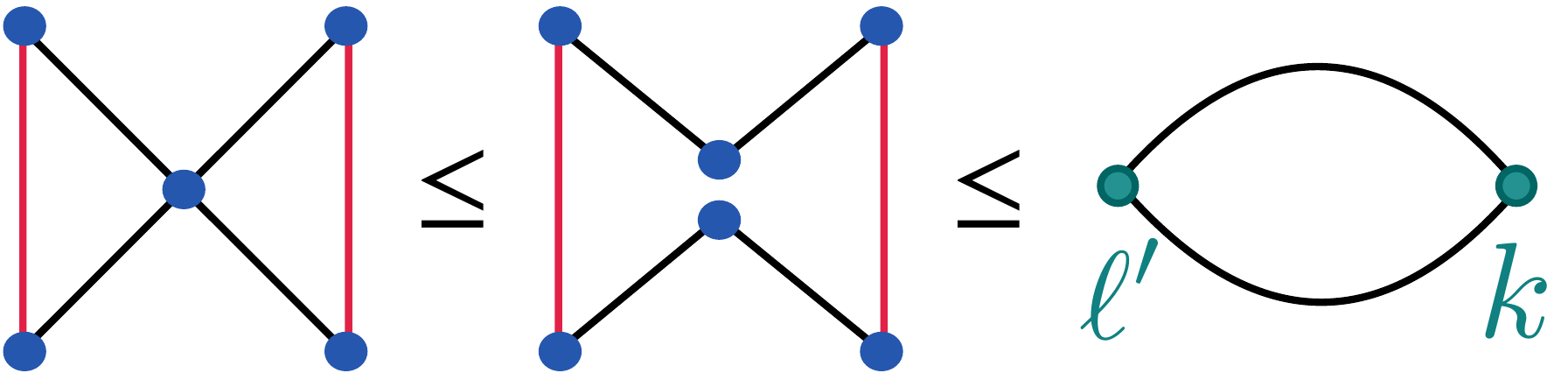}
\caption{The diagrammatic derivation for the upper bound of the contribution of the fourth term in (\ref{cal_xy_1})}
\label{fig_sup_xy_5}
\end{figure}

We similarly obtain the following upper bound of the contribution to $\ave{ \mathcal{J}_n \mathcal{J}_0 }$:
\begin{align}
    &\sum_{\substack{k \geq n \\ \ell < n}} \sum_{\substack{k' \geq 0 \\ \ell' < 0}} \sum_{m,m'=1}^N \frac{ J_{x}^2 J_{y}^2 }{16} \sum_{ \substack{ \{a,b,c\}=\{k,\ell,m\} \\ \{a',b',c'\}=\{k',\ell',m'\} }} \frac{ \delta_{c,c'} }{ d_{a,c}^\alpha d_{b,c}^\alpha d_{a',c}^\alpha d_{b',c}^\alpha } \frac{ (c_4)^2 }{ d_{a,a'}^\alpha d_{b,b'}^\alpha } 
    \notag \\ &
    \leq \sum_{\substack{k \geq n \\ \ell < n}} \sum_{\substack{k' \geq 0 \\ \ell' < 0}} \sum_{m,m'=1}^N \frac{ J_{x}^2 J_{y}^2 }{16} \sum_{ \substack{ \{a,b,c\}=\{k,\ell,m\} \\ \{a',b',c'\}=\{k',\ell',m'\} }} \sum_{c''=1}^N \frac{ \delta_{c,c'} (c_4)^2 }{ d_{a,a'}^\alpha d_{a',c}^\alpha d_{c,b'}^\alpha d_{b',b}^\alpha d_{b,c''}^\alpha d_{c'',a}^\alpha }
    \notag \\ &
    \leq \frac{ J_{x}^2 J_{y}^2 }{16} u^4 (c_4)^2 \sum_{\substack{ k \geq n \\ \ell' < 0 }} \frac{1}{ d_{k,\ell'}^{2\alpha} }
    \leq c^{(4)} n^{2-2\alpha}
    \label{cal_xy_5}
    ,
\end{align}
where $c^{(4)} = (J_x^2 J_y^2/16) u^4 (c_4)^2 /\{ (1-2\alpha)(2-2\alpha) \}$.

From (\ref{cal_xy_2}), (\ref{cal_xy_3}), (\ref{cal_xy_4}), (\ref{cal_xy_5}), we demonstrate that all terms in (\ref{cal_xy_1}) can be bounded as follows: 
\begin{align}
    \abs{\ave{ \mathcal{J}_n \mathcal{J}_0 }} \leq C_1 n^{2-2\alpha}
    ,
\end{align}
where $C_1 = 3 c^{(1)} + c^{(2)} + c^{(3)} + c^{(4)}$. 
Thus, we establish the claims of the theorem.
\end{proof}

\subsection{XYZ model}

In this subsection, we consider the one-dimensional long-range interacting XYZ model, which is described by the following Hamiltonian
\begin{equation}
    H= - \sum_{n=1}^N \sum_{r=1}^{N/2} \Big[ J_x \frac{S^x_{n} S^x_{n+r}}{r^\alpha} + J_y \frac{S^y_{n} S^y_{n+r}}{r^\alpha} + J_z \frac{S^z_{n} S^z_{n+r}}{r^\alpha} \Big]
    .
\end{equation}
The local energy at site $n$ is defined as 
\begin{equation}
    \hat{h}_n := - \sum_{r=1}^{N/2} \sum_{\sigma =x,y,z } \frac{J_\sigma}{2} \Big[ \frac{S^\sigma_{n} S^\sigma_{n+r}}{r^\alpha} + \frac{S^\sigma_{n} S^\sigma_{n-r} }{r^\alpha} \Big]
    ,
\end{equation}
and satisfies the following continuity equations
\begin{align}
    \partial_t \hat{h}_n
    &= - (\mathcal{J}_{n+1} - \mathcal{J}_{n})
    ,
    \\
    \mathcal{J}_n 
    &= - \sum_{\substack{i \geq n , j < n \\ \abs{i-j}<N/2 }} t_{i \leftarrow j}
    .
\end{align}
Energy transmission operator $t_{i \leftarrow j}$ is defined as 
\begin{align}
    t_{i \leftarrow j} 
    &= \{ \hat{h}_i , \hat{h}_j \} 
    = \sum_{m=1}^N \sum_{ \sigma \in P_{\mathrm{cyc}}^{xyz} } \frac{J_{\sigma_1} J_{\sigma_2}}{4} \sum_{ \{a,b,c\}=\{i,j,m\} } \tilde{\varepsilon}_{abc} \frac{ S^{\sigma_1}_{a} S^{\sigma_2}_{b} S^{\sigma_3}_{c} }{d_{a,c}^\alpha d_{b,c}^\alpha}
    ,
\end{align}
where $\sum_{\sigma \in P^{xyz}_{\mathrm{cyc}}}$ denotes sums over all circular permutation of the variables $(x,y,z)$. In other words, $\sigma=(\sigma_1,\sigma_2,\sigma_3)$ represents one of $(x,y,z), (y,z,x)$ or $(z,x,y)$, and $\sum_{\sigma \in P^{xyz}_{\mathrm{cyc}}}$ is defined as the summation over these permutations. $\tilde{\varepsilon}_{abc}$ is defined as follows:
\begin{align}
    \tilde{\varepsilon}_{abc}=
    \begin{cases}
        +1 \, , \quad &\text{for} \, \, 
        (a,b,c)=(i,j,m),(m,j,i),(i,m,j),
        \\
        -1 \, , \quad &\text{for} \, \, 
        (a,b,c)=(j,i,m),(j,m,i),(m,i,j),
    \end{cases}
\end{align}
Note that $\tilde{\varepsilon}_{abc}$ is different from Levi-Civita symbols.

From the symmetry of the Hamiltonian, we note that the following properties hold
\begin{align}
    \ave{ S^x_i } 
    &= \ave{ S^y_i } 
    = \ave{ S^z_i } = 0 \label{xyz_prop1}
    ,
    \\
    \ave{ S^\sigma_i S^\tau_j } 
    &= \delta_{\sigma,\tau} \ave{ S^\sigma_i S^\sigma_j } \label{xyz_prop2}
    ,
\end{align}
for $i,j=1,...,n$ and $\sigma,\tau \in \{x,y,z\}$. The relation (\ref{xyz_prop1}), (\ref{xyz_prop2}) is shown that the Hamiltonian is invariant to the canonical transformation $( S^{\sigma}_i, S^{\tau}_i, S^{\upsilon}_i) \rightarrow (S^{\sigma}_i, -S^{\tau}_i, -S^{\upsilon}_i)$ for $\{\sigma,\tau,\upsilon\}=\{x,y,z\}$ and $i=1,...,N$.

Using these properties and Theorem \ref{cumulantpowerlawtheorem}, we establish the following theorem:

\begin{theorem} \label{main_xyz}
    In the one-dimensional long-range interacting XYZ model, under the condition $\beta < \beta_c$, the equal-time current correlation is upper-bounded as
    \begin{align}
        \abs{\ave{ \mathcal{J}_n \mathcal{J}_0 }}
        \leq C_2 n^{2-2\alpha}
        \label{upp-main_xyz},
    \end{align}
    for $1<\alpha \leq 2$, and
    \begin{align}
        \abs{\ave{ \mathcal{J}_n \mathcal{J}_0 }}
        \leq C'_2 n^{-\alpha}
        ,
    \end{align}
    for $\alpha > 2$.
\end{theorem}

\begin{proof}

The equal-time current correlation $\ave{ \mathcal{J}_n \mathcal{J}_0 }$ is given as follows:
\begin{align}
    \ave{ \mathcal{J}_n \mathcal{J}_0 }
    &= \sum_{\substack{k \geq n \\ \ell < n}} \sum_{\substack{k' \geq 0 \\ \ell' < 0}} \sum_{m,m'=1}^N \sum_{ \sigma,\tau \in P_{\mathrm{cyc}}^{xyz} } \frac{ J_{\sigma_1} J_{\sigma_2} J_{\tau_1} J_{\tau_2} }{16} \sum_{ \substack{ \{a,b,c\}=\{k,\ell,m\} \\ \{a',b',c'\}=\{k',\ell',m'\} }} \tilde{\varepsilon}_{abc} \tilde{\varepsilon}_{a'b'c'} \frac{ \langle S^{\sigma_1}_{a} S^{\sigma_2}_{b} S^{\sigma_3}_{c} S^{\tau_1}_{a'} S^{\tau_2}_{b'} S^{\tau_3}_{c'} \rangle }{ d_{a,c}^\alpha d_{b,c}^\alpha d_{a',c'}^\alpha d_{b',c'}^\alpha }
    \label{curr_xyz}
\end{align}

The term $1 / ( d_{a,c}^\alpha d_{b,c}^\alpha d_{a',c'}^\alpha d_{b',c'}^\alpha ) $ in (\ref{curr_xyz}) is illustrated diagrammatically in Fig.\ref{fig_sup_xyz_1}.

\begin{figure}[h]
\centering
\includegraphics[width=0.2\textwidth]{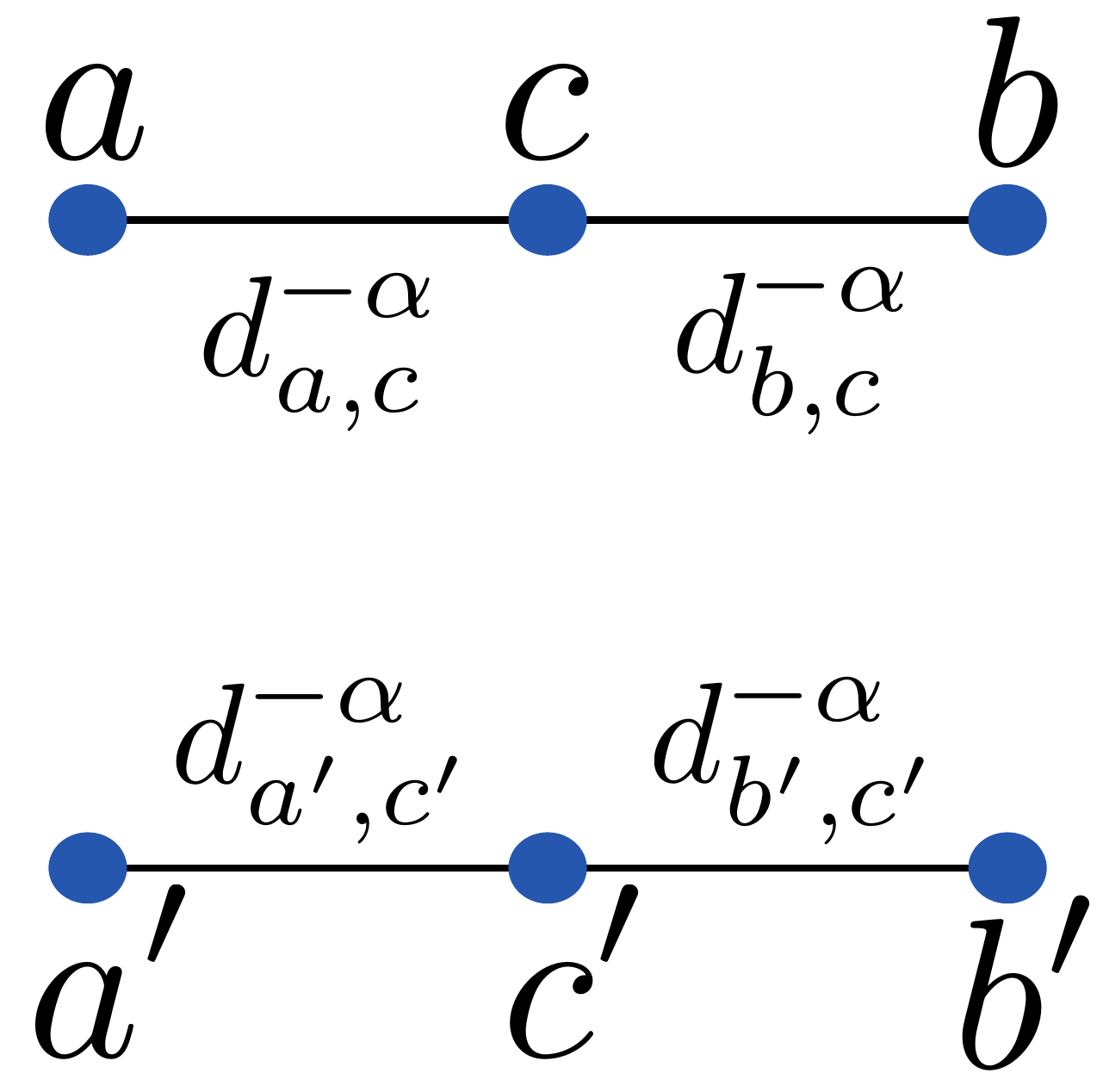}
\caption{The diagrammatic representation of the term $d_{a,c}^{-\alpha} d_{b,c}^{-\alpha} d_{a',c'}^{-\alpha} d_{b',c'}^{-\alpha} $ in (\ref{curr_xyz}) $\ave{ \mathcal{J}_n \mathcal{J}_0 }$}
\label{fig_sup_xyz_1}
\end{figure}

The 6-body correlation $\langle S^{\sigma_1}_{a} S^{\sigma_2}_{b} S^{\sigma_3}_{c} S^{\tau_1}_{a'} S^{\tau_2}_{b'} S^{\tau_3}_{c'} \rangle$ can be decomposed into cumulant products as follows:
\begin{align}
    &\langle S^{\sigma_1}_{a} S^{\sigma_2}_{b} S^{\sigma_3}_{c} S^{\tau_1}_{a'} S^{\tau_2}_{b'} S^{\tau_3}_{c'} \rangle
    \notag \\
    &= \langle S^{\sigma_1}_{a} S^{\sigma_2}_{b} S^{\sigma_3}_{c} S^{\tau_1}_{a'} S^{\tau_2}_{b'} S^{\tau_3}_{c'} \rangle_c 
    + \sum_{ \substack{ \{i_1,...i_6\} = \\ 
    \{a,b,c,a',b',c'\} \\
    \mu \in \{x,y,z\} , \nu_k \neq \mu } }
    \ave{ S^\mu_{i_1} S^\mu_{i_2} }_c \ave{ S^{\nu_1}_{i_3} S^{\nu_2}_{i_4} S^{\nu_3}_{i_5} S^{\nu_4}_{i_6} }_c 
    + \ave{ S^x_{i_1} S^x_{i_2} }_c \ave{ S^y_{i_3} S^y_{i_4} }_c \ave{ S^z_{i_5} S^z_{i_6} }_c
\end{align}
where we use the relation (\ref{xyz_prop1}), (\ref{xyz_prop2}). Next, we use Theorem \ref{cumulantpowerlawtheorem} as follows: 
\begin{align}
    \langle S^{\sigma_1}_{a} S^{\sigma_2}_{b} S^{\sigma_3}_{c} S^{\tau_1}_{a'} S^{\tau_2}_{b'} S^{\tau_3}_{c'} \rangle_c 
    &\leq c_1 \|S^{\sigma_1}_{a} \| 
    ... \| S^{\tau_3}_{c'} \| \sum_{ \substack{
    \{ j_1,...,j_6 \}=  \\ \{a,b,c,a',b',c'\} \\ \mathrm{:connected} } } \frac{1}{ d_{j_1,j_2}^\alpha ... d_{j_5,j_6}^\alpha }
    \label{cpc_xyz_1}
    \\
    \ave{ S^{\nu_1}_{i_3} S^{\nu_2}_{i_4} S^{\nu_3}_{i_5} S^{\nu_4}_{i_6} }_c 
    &\leq c_2 \| S^{\nu_1}_{i_3} \|... \| S^{\nu_4}_{i_6} \| \sum_{\substack{ \{j_1,...,j_4\} \\ =\{ i_1,...i_4 \} \\ \mathrm{:connected} }} \frac{1}{ d_{j_1,j_2}^\alpha d_{j_2,j_3}^\alpha d_{j_3,j_4}^\alpha }
    \label{cpc_xyz_2}
    \\
    \langle S^\mu_{i_p} S^\mu_{i_q} \rangle_c 
    &\leq c_3 \| S^\mu_{i_p} \| \| S^\mu_{i_q} \| \frac{1}{ d_{i_p,i_q}^\alpha }
    \label{cpc_xyz_3}
\end{align}
By using Theorem \ref{cumulantpowerlawtheorem} and (\ref{cpc_xyz_1})-(\ref{cpc_xyz_3}), we obtain the following inequality for $\abs{\ave{ \mathcal{J}_n \mathcal{J}_0 }}$:
\begin{align}
    \abs{\ave{ \mathcal{J}_n \mathcal{J}_0 }}
    & \leq \sum_{\substack{k \geq n \\ \ell < n}} \sum_{\substack{k' \geq 0 \\ \ell' < 0}} \sum_{m,m'=1}^N \sum_{ \sigma,\tau \in P_{\mathrm{cyc}}^{xyz} } \frac{ J_{\sigma_1} J_{\sigma_2} J_{\tau_1} J_{\tau_2} }{16} \!\!\! \sum_{ \substack{ \{a,b,c\}=\{k,\ell,m\} \\ \{a',b',c'\}=\{k',\ell',m'\} }} 
    \frac{1}{ d_{a,c}^\alpha d_{b,c}^\alpha d_{a',c'}^\alpha d_{b',c'}^\alpha } 
    \notag \\ & 
    \times \bigg\{ \abs{ \ave{ S^{\sigma_1}_{a} S^{\sigma_2}_{b} S^{\sigma_3}_{c} S^{\tau_1}_{a'} S^{\tau_2}_{b'} S^{\tau_3}_{c'} }_c } 
    + \!\!\! \sum_{ \substack{ \{i_1,...i_6\} = \\ 
    \{a,b,c,a',b',c'\} \\
    \mu \in\{x,y,z\}, \nu_k \neq \mu } }
    \abs{ \ave{ S^\mu_{i_1} S^\mu_{i_2} }_c \ave{ S^{\nu_1}_{i_3} S^{\nu_2}_{i_4} S^{\nu_3}_{i_5} S^{\nu_4}_{i_6} }_c } 
    + \abs{ \ave{ S^x_{i_1} S^x_{i_2} }_c \ave{ S^y_{i_3} S^y_{i_4} }_c \ave{ S^z_{i_5} S^z_{i_6} }_c } \bigg\}
    \notag \\
    & \leq \sum_{\substack{k \geq n \\ \ell < n}} \sum_{\substack{k' \geq 0 \\ \ell' < 0}} \sum_{m,m'=1}^N \sum_{ \sigma,\tau \in P_{\mathrm{cyc}}^{xyz} } \frac{ J_{\sigma_1} J_{\sigma_2} J_{\tau_1} J_{\tau_2} }{16} \!\!\! \sum_{ \substack{ \{a,b,c\}=\{k,\ell,m\} \\ \{a',b',c'\}=\{k',\ell',m'\} }} \frac{1}{ d_{a,c}^\alpha d_{b,c}^\alpha d_{a',c'}^\alpha d_{b',c'}^\alpha } 
    \sum_{\substack{ \{i_1,...,i_6\}= \\ \{ a,b,c,a',b',c' \} }}
    \notag \\ & 
    \times
    \bigg\{ \sum_{ \substack{i_1,...,i_6 \\ \mathrm{:connected} } } \frac{ c_1 }{ d_{i_1,i_2}^\alpha ... d_{i_5,i_6}^\alpha } 
    + \sum_{\substack{ i_1 \in \{a,b,c\} \\ i_2 \in \{a',b',c'\} }} \frac{ c_3 }{ d_{i_1,i_2}^\alpha }  \sum_{\substack{ i_3,...,i_6 \\ \mathrm{:connected} }} \frac{ c_2 }{ d_{i_3,i_4}^\alpha d_{i_4,i_5}^\alpha d_{i_5,i_6}^\alpha } 
    + \sum_{ \substack{ \{i_1,i_3,i_5\}=\{a,b,c\} \\ \{i_2,i_4,i_6\}=\{a',b',c'\} }}  \frac{ c_3 }{ d_{i_1,i_2}^\alpha } \frac{ c_3 }{ d_{i_3,i_4}^\alpha } \frac{ c_3 }{ d_{i_5,i_6}^\alpha } \bigg\}
    \label{cal_xyz_1}
\end{align}

Let us first consider the contribution of the first term in (\ref{cal_xyz_1}):
\begin{align}
    \sum_{\substack{k \geq n \\ \ell < n}} \sum_{\substack{k' \geq 0 \\ \ell' < 0}} \sum_{m,m'=1}^N \sum_{ \sigma,\tau \in P_{\mathrm{cyc}}^{xyz} } \frac{ J_{\sigma_1} J_{\sigma_2} J_{\tau_1} J_{\tau_2} }{16} \!\!\! \sum_{ \substack{ \{a,b,c\}=\{k,\ell,m\} \\ \{a',b',c'\}=\{k',\ell',m'\} }} \frac{1}{ d_{a,c}^\alpha d_{b,c}^\alpha d_{a',c'}^\alpha d_{b',c'}^\alpha } 
    \sum_{\substack{ \{i_1,...,i_6\}= \\ \{ a,b,c,a',b',c' \} \\ \mathrm{:connected} }} \frac{ c_1 }{ d_{i_1,i_2}^\alpha ... d_{i_5,i_6}^\alpha } 
    \label{cal_xyz_1_1}
\end{align}

Figure.\ref{fig_sup_xyz_2} comprehensively covers all the cases of clustering linkage in (\ref{cal_xyz_1_1}). Here, diagrams with identical topology are regarded as equivalent.
The existence of 51 topologically distinct graphs can be verified using Burnside’s lemma (also known as Cauchy–Frobenius lemma).

The methods of clustering linkage can be classified into two cases: (1) when there are two or more connected components linking the upper and lower graphs in Fig.\ref{fig_sup_xyz_1}, and (2) when there is only one connected component. 

\begin{figure}[h]
\centering
\includegraphics[width=0.8\textwidth]{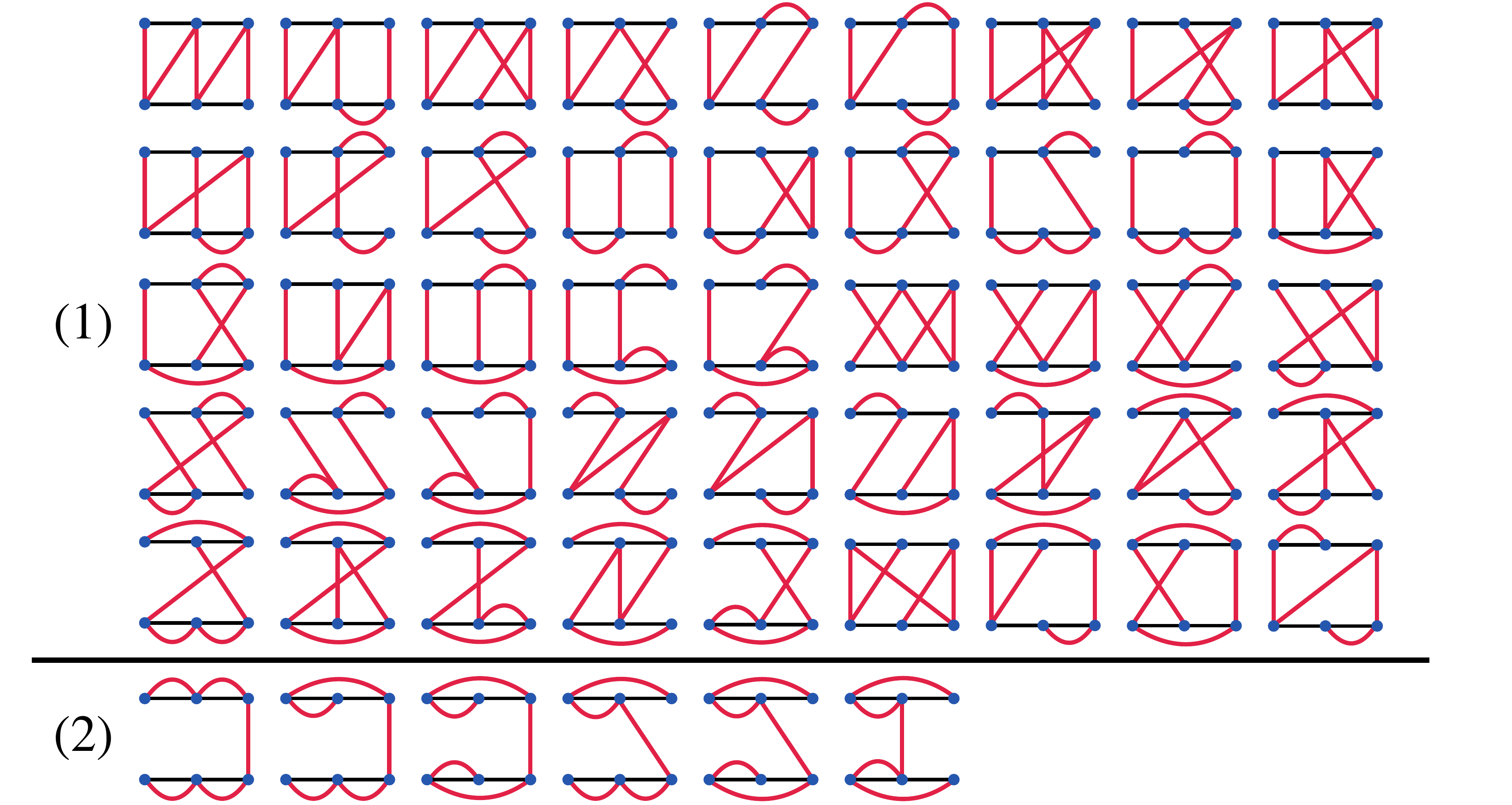}
\caption{The diagrammatic representation of the contribution of the first term in (\ref{cal_xyz_1}) to the upper bound of $\ave{ \mathcal{J}_n \mathcal{J}_0 }$. It can be classified into two cases: (1) when there are two or more connected components linking the upper and lower graphs in Fig.\ref{fig_sup_xyz_1}, and (2) when there is only one connected component. }
\label{fig_sup_xyz_2}
\end{figure} 

In the case (1) in Fig.\ref{fig_sup_xyz_2}, let us first compute the upper bound of the contribution to $\ave{ \mathcal{J}_n \mathcal{J}_0 }$ in the following typical examples: (1)-(i) the first row, fifth column from the left in Fig.\ref{fig_sup_xyz_2} (1), and (1)-(ii) the fourth row, ninth column from the left in Fig.\ref{fig_sup_xyz_2} (1). 

In the case (1)-(i), we obtain the following upper bound of the contribution to $\ave{ \mathcal{J}_n \mathcal{J}_0 }$:
\begin{align}
    &\sum_{\substack{k \geq n \\ \ell < n}} \sum_{\substack{k' \geq 0 \\ \ell' < 0}} \sum_{m,m'=1}^N \sum_{ \sigma,\tau \in P_{\mathrm{cyc}}^{xyz} } \frac{ J_{\sigma_1} J_{\sigma_2} J_{\tau_1} J_{\tau_2} }{16} \sum_{ \substack{ \{a,b,c\}=\{k,\ell,m\} \\ \{a',b',c'\}=\{k',\ell',m'\} }} \frac{1}{ d_{a,c}^\alpha d_{b,c}^\alpha d_{a',c'}^\alpha d_{b',c'}^\alpha } \frac{ c_1 }{ d_{a,a'}^\alpha d_{a',c}^\alpha d_{c,b}^\alpha d_{b,c'}^\alpha d_{c',b}^\alpha }
    \notag \\ &
    \leq \sum_{\substack{k \geq n \\ \ell < n}} \sum_{\substack{k' \geq 0 \\ \ell' < 0}} \sum_{m,m'=1}^N \sum_{ \sigma,\tau \in P_{\mathrm{cyc}}^{xyz} } \frac{ J_{\sigma_1} J_{\sigma_2} J_{\tau_1} J_{\tau_2} }{16} \sum_{ \substack{ \{a,b,c\}=\{k,\ell,m\} \\ \{a',b',c'\}=\{k',\ell',m'\} }} \frac{1}{ d_{a,c}^\alpha d_{b,c}^\alpha d_{a',c'}^\alpha d_{b',c'}^\alpha } \frac{ \bar{a}^2 c_1 }{ d_{a,a'}^\alpha d_{b',c'}^\alpha d_{c',b}^\alpha }
    \notag \\ &
    \leq \sum_{\substack{k \geq n \\ \ell < n}} \sum_{\substack{k' \geq 0 \\ \ell' < 0}} \sum_{m,m'=1}^N \sum_{ \sigma,\tau \in P_{\mathrm{cyc}}^{xyz} } \frac{ J_{\sigma_1} J_{\sigma_2} J_{\tau_1} J_{\tau_2} }{16} \sum_{ \substack{ \{a,b,c\}=\{k,\ell,m\} \\ \{a',b',c'\}=\{k',\ell',m'\} }} \sum_{c''=1}^N \frac{ \bar{a}^{2} c_1 }{ d_{a,a'}^\alpha d_{a',c'}^\alpha d_{c',b'}^\alpha d_{b',c''}^\alpha d_{c'',b}^\alpha d_{b,c}^\alpha d_{c,a}^\alpha } 
    \notag \\ &
    \leq \sum_{ \sigma,\tau \in P_{\mathrm{cyc}}^{xyz} } \frac{ J_{\sigma_1} J_{\sigma_2} J_{\tau_1} J_{\tau_2} }{16} u^5 \bar{a}^{2} c_1 \sum_{\substack{k \geq n \\ \ell' < 0}} \frac{ 1 }{ d_{k,\ell'}^{2\alpha} }
    \leq c n^{2-2\alpha},
\end{align}
where $c=u^5 \bar{a}^{2} c_1 / \{ (1-2\alpha)(2-2\alpha) \} \cdot \sum_{ \sigma,\tau \in P_{\mathrm{cyc}}^{xyz} } J_{\sigma_1} J_{\sigma_2} J_{\tau_1} J_{\tau_2}/16 $. The corresponding diagrammatic derivation is presented in the upper half of Figure.\ref{fig_sup_xyz_2_1}.

In the case (1)-(ii), we similarly obtain the following upper bound of the contribution to $\ave{ \mathcal{J}_n \mathcal{J}_0 }$:
\begin{align}
    &\sum_{\substack{k \geq n \\ \ell < n}} \sum_{\substack{k' \geq 0 \\ \ell' < 0}} \sum_{m,m'=1}^N \sum_{ \sigma,\tau \in P_{\mathrm{cyc}}^{xyz} } \frac{ J_{\sigma_1} J_{\sigma_2} J_{\tau_1} J_{\tau_2} }{16} \sum_{ \substack{ \{a,b,c\}=\{k,\ell,m\} \\ \{a',b',c'\}=\{k',\ell',m'\} }} \frac{1}{ d_{a,c}^\alpha d_{b,c}^\alpha d_{a',c'}^\alpha d_{b',c'}^\alpha } \frac{ c_1 }{ d_{a,b}^\alpha d_{b,a'}^\alpha d_{a',c'}^\alpha d_{c',c}^\alpha d_{c,b'}^\alpha }
    \notag \\ &
    \leq \sum_{\substack{k \geq n \\ \ell < n}} \sum_{\substack{k' \geq 0 \\ \ell' < 0}} \sum_{m,m'=1}^N \sum_{ \sigma,\tau \in P_{\mathrm{cyc}}^{xyz} } \frac{ J_{\sigma_1} J_{\sigma_2} J_{\tau_1} J_{\tau_2} }{16} \sum_{ \substack{ \{a,b,c\}=\{k,\ell,m\} \\ \{a',b',c'\}=\{k',\ell',m'\} }} \frac{1}{ d_{a,c}^\alpha d_{b,c}^\alpha d_{a',c'}^\alpha d_{b',c'}^\alpha } \frac{ \bar{a} c_1 }{ d_{a,b}^\alpha d_{a',c'}^\alpha d_{c',c}^\alpha d_{c,b'}^\alpha }
    \notag \\ &
    \leq \sum_{\substack{k \geq n \\ \ell < n}} \sum_{\substack{k' \geq 0 \\ \ell' < 0}} \sum_{m,m'=1}^N \sum_{ \sigma,\tau \in P_{\mathrm{cyc}}^{xyz} } \frac{ J_{\sigma_1} J_{\sigma_2} J_{\tau_1} J_{\tau_2} }{16} \sum_{ \substack{ \{a,b,c\}=\{k,\ell,m\} \\ \{a',b',c'\}=\{k',\ell',m'\} }} \sum_{c''=1}^N \frac{ \bar{a} c_1 }{ d_{a,c}^\alpha d_{c,c'}^\alpha d_{c',a'}^\alpha d_{a',c''}^\alpha d_{c'',b'}^\alpha d_{b',c}^\alpha d_{c,b}^\alpha d_{b,a}^\alpha }
    \notag \\ &
    \leq \sum_{ \sigma,\tau \in P_{\mathrm{cyc}}^{xyz} } \frac{ J_{\sigma_1} J_{\sigma_2} J_{\tau_1} J_{\tau_2} }{16} u^6 \bar{a} c_1 \sum_{\substack{k \geq n \\ \ell' < 0}} \frac{ 1 }{ d_{k,\ell'}^{2\alpha} }
    \leq c' n^{2-2\alpha},
\end{align}
where $c'=u^6 \bar{a} c_1 / \{ (1-2\alpha)(2-2\alpha) \} \cdot \sum_{ \sigma,\tau \in P_{\mathrm{cyc}}^{xyz} } J_{\sigma_1} J_{\sigma_2} J_{\tau_1} J_{\tau_2}/16 $. The corresponding diagrammatic derivation is presented in the lower half of Figure.\ref{fig_sup_xyz_2_1}.

\begin{figure}[h]
\centering
\includegraphics[width=0.65\textwidth]{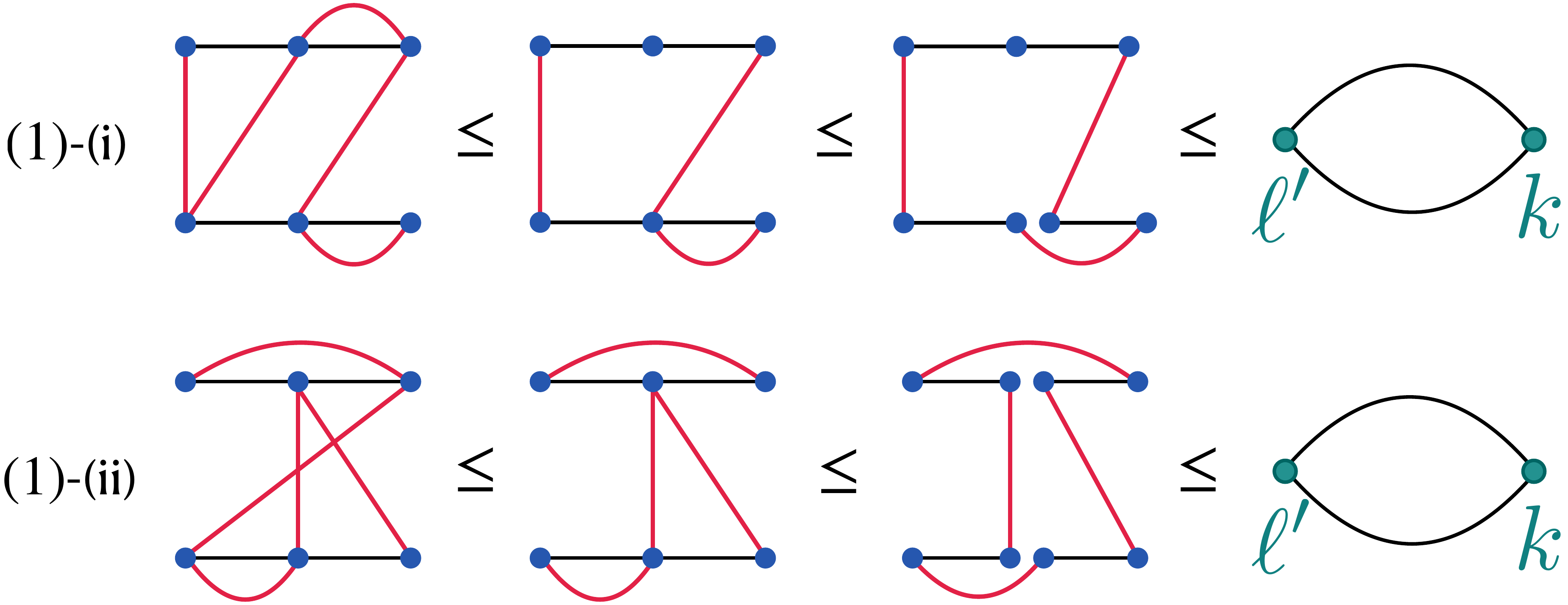}
\caption{Typical examples of the diagrammatic derivation for the  upper bound in the case Figure.\ref{fig_sup_xyz_2} (1) }
\label{fig_sup_xyz_2_1}
\end{figure}

The same upper bound can be obtained for all remaining cases in Fig.\ref{fig_sup_xyz_2} (1) as well:
\begin{align}
    &\sum_{\substack{k \geq n \\ \ell < n}} \sum_{\substack{k' \geq 0 \\ \ell' < 0}} \sum_{m,m'=1}^N \sum_{ \sigma,\tau \in P_{\mathrm{cyc}}^{xyz} } \frac{ J_{\sigma_1} J_{\sigma_2} J_{\tau_1} J_{\tau_2} }{16} \sum_{ \substack{ \{a,b,c\}=\{k,\ell,m\} \\ \{a',b',c'\}=\{k',\ell',m'\} }} \frac{1}{ d_{a,c}^\alpha d_{b,c}^\alpha d_{a',c'}^\alpha d_{b',c'}^\alpha } 
    \sum_{\substack{ \{i_1,...,i_6\}= \\ \{ a,b,c,a',b',c' \} \\ \text{:the case (1) in Fig.\ref{fig_sup_xyz_2} } }} \frac{ c_1 }{ d_{i_1,i_2}^\alpha ... d_{i_5,i_6}^\alpha } 
    \notag \\ &
    \leq \mathcal{O}(1) n^{2-2\alpha}
    \label{cal_xyz_2_1}
    .
\end{align}

Next, we consider the case (2) in Fig.\ref{fig_sup_xyz_2}.
In this case, $(k,\ell,m)$ and $(k',\ell',m')$ are connected via a single edge. We construct closed loops for both $(k,\ell,m)$ and $(k',\ell',m')$. Furthermore, considering the dominant contribution in power-law clustering, the ordering relation between $k,\ell,m,k',\ell',m'$ is as follows: 
\begin{align}
    \forall u \in \{ k,\ell,m \} , \forall v \in \{k',\ell',m'\}, \, u>v.
\end{align}
We define $p,p' \in \{1,...,N\}$ as 
    $p = \min \{ k,\ell,m \} , \, p' = \max \{ k',\ell',m' \}$. 
Here, $p \in \{\ell,m\}$ and $p' \in \{k',m'\}$, which implies that $k,\ell',p,p'$ are arranged as 
    $k \geq n > p > p' \geq 0 > \ell'$. 

Since we obtain the following inequality for the contribution of the case (2) in Fig.\ref{fig_sup_xyz_2}:
\begin{align}
    &\sum_{\substack{k,i,i',\ell' \\ : k \geq n > p > p' \geq 0 > \ell' }} \sum_{q,q'} \sum_{ \sigma,\tau \in P_{\mathrm{cyc}}^{xyz} } \frac{ J_{\sigma_1} J_{\sigma_2} J_{\tau_1} J_{\tau_2} }{16} \sum_{ \substack{ \{a,b,c\}=\{k,p,q\} \\ \{a',b',c'\}=\{k',p',q'\} }} \frac{1}{ d_{a,c}^\alpha d_{b,c}^\alpha d_{a',c'}^\alpha d_{b',c'}^\alpha } 
    \sum_{\substack{ \{i_1,...,i_6\}= \\ \{ a,b,c,a',b',c' \} \\ \text{:the case (2) in Fig.\ref{fig_sup_xyz_2}} }} \frac{ c_1 }{ d_{i_1,i_2}^\alpha ... d_{i_5,i_6}^\alpha } 
    \notag \\ &
    \leq \mathcal{O}(1) \sum_{\substack{k,p,p',\ell' \\ : k \geq n > p > p' \geq 0 > \ell' }} \frac{1}{d_{k,p}^{2\alpha} } \frac{1}{ d_{p,p'}^\alpha } \frac{1}{ d_{p',\ell'}^{2\alpha} }
    \label{cal_xyz_2_2}
    , 
\end{align}
where we denote $q,q'$ as $q=\{\ell,m\} \setminus \{p\}$ and $q'=\{k',m'\} \setminus \{p'\}$. 
The corresponding schematic diagram is presented in Fig.\ref{fig_sup_xyz_2_2}. 

\begin{figure}[h]
\centering
\includegraphics[width=0.5\textwidth]{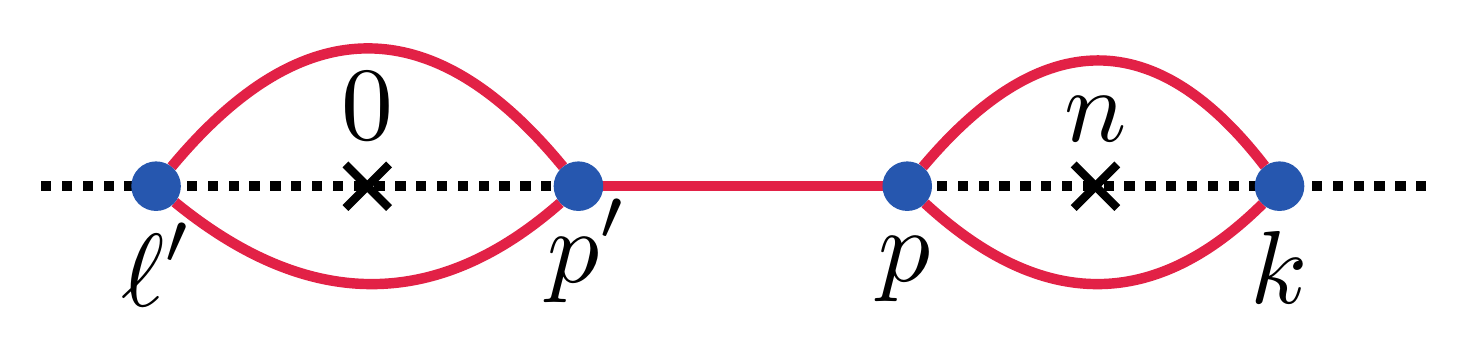}
\caption{The schematic diagram of the case (2) in Fig.\ref{fig_sup_xyz_2} }
\label{fig_sup_xyz_2_2}
\end{figure}

Let us demonstrate (\ref{cal_xyz_2_2}) in a typical example: the second column from the left in Fig.\ref{fig_sup_xyz_2} (2).
For this example, we obtain the following inequality:
\begin{align}
    &\sum_{\substack{k,p,p',\ell' \\ : k \geq n > p > p' \geq 0 > \ell' }}  \sum_{q,q'} \sum_{ \sigma,\tau \in P_{\mathrm{cyc}}^{xyz} } \frac{ J_{\sigma_1} J_{\sigma_2} J_{\tau_1} J_{\tau_2} }{16} \sum_{ \substack{ \{a,b,c\}=\{k,p,q\} \\ \{a',b',c'\}=\{\ell',p',q'\} }} \frac{1}{ d_{a,c}^\alpha d_{b,c}^\alpha d_{a',c'}^\alpha d_{b',c'}^\alpha } 
    \frac{c_1}{ d_{c,a}^\alpha d_{a,b}^\alpha d_{b,b'}^\alpha d_{b',c'}^\alpha d_{c',a'}^\alpha }
    \notag \\ &
    \leq \sum_{\substack{k,p,p',\ell' \\ : k \geq n > p > p' \geq 0 > \ell' }}  \sum_{q,q'} \sum_{ \sigma,\tau \in P_{\mathrm{cyc}}^{xyz} } \frac{ J_{\sigma_1} J_{\sigma_2} J_{\tau_1} J_{\tau_2} }{16} \sum_{ \substack{ \{a,b,c\}=\{k,p,q\} \\ \{a',b',c'\}=\{\ell',p',q'\} }} \frac{1}{ d_{a,c}^\alpha d_{b,c}^\alpha d_{a',c'}^\alpha d_{b',c'}^\alpha } 
    \sum_{c''=1}^N \frac{ \bar{a} c_1}{ d_{a,b}^\alpha d_{b,b'}^\alpha d_{b',c''}^\alpha d_{c'',a'}^\alpha }
    \notag \\ &
    \leq \sum_{\substack{k,p,p',\ell' \\ : k \geq n > p > p' \geq 0 > \ell' }}  \sum_{ \sigma,\tau \in P_{\mathrm{cyc}}^{xyz} } \frac{ J_{\sigma_1} J_{\sigma_2} J_{\tau_1} J_{\tau_2} }{16} u^3 \bar{a} c_1 \frac{1}{d_{k,p}^{2\alpha} } \frac{1}{ d_{p,p'}^\alpha } \frac{1}{ d_{p',\ell'}^{2\alpha} }
    = c' \sum_{\substack{k,p,p',\ell' \\ : k \geq n \geq p > p' >0 \geq \ell' }} \frac{1}{d_{k,p}^{2\alpha} } \frac{1}{ d_{p,p'}^\alpha } \frac{1}{ d_{p',\ell'}^{2\alpha} }
    ,
\end{align}
where $c' = u^3 \bar{a} c_1 \sum_{ \sigma,\tau \in P_{\mathrm{cyc}}^{xyz} } J_{\sigma_1} J_{\sigma_2} J_{\tau_1} J_{\tau_2}/16 $. The corresponding diagrammatic derivation is presented in Fig.\ref{fig_sup_xyz_2_3}. 

\begin{figure}[h]
\centering
\includegraphics[width=0.7\textwidth]{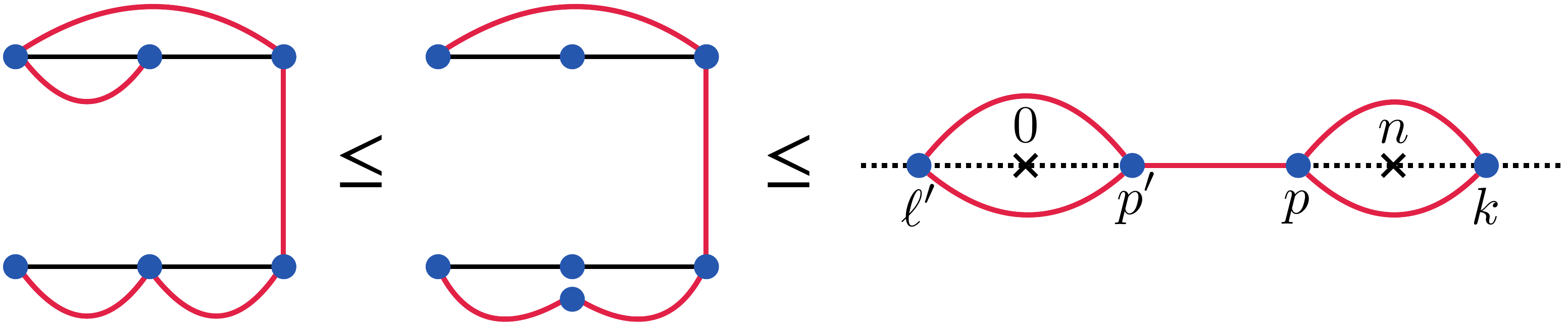}
\caption{A typical example of the diagrammatic derivation for the upper bound in the case Figure.\ref{fig_sup_xyz_2} (2) }
\label{fig_sup_xyz_2_3}
\end{figure}

The same inequality can also be obtained for all remaining cases in Fig.\ref{fig_sup_xyz_2} (2). Thus, (\ref{cal_xyz_2_2}) has been demonstrated.

In the case (2), the following upper bound for the contribution to $\ave{ \mathcal{J}_n \mathcal{J}_0 }$ is derived from (\ref{cal_xyz_2_2}): 
\begin{align}
    &\sum_{\substack{k,p,p',\ell' \\ : k \geq n > p > p' \geq 0 > \ell' }} \frac{1}{d_{k,p}^{2\alpha} } \frac{1}{ d_{p,p'}^\alpha } \frac{1}{ d_{p',\ell'}^{2\alpha} }
    = \sum_{\substack{k,p,p',\ell' \\ : k \geq n > p > p' \geq 0 > \ell' }} \frac{1}{(d_{k,n}+d_{n,p})^{2\alpha} } \frac{1}{ d_{p,p'}^\alpha } \frac{1}{ (d_{p',0}+d_{0,\ell'})^{2\alpha} }
    \notag \\ &
    \leq \sum_{\substack{p,p' \\ : n>p>p'\geq 0 }} \sum_{x,y=1}^{N/2} \frac{1}{(d_{n,p}+x)^{2\alpha} } \frac{1}{ d_{p,p'}^\alpha } \frac{1}{ (d_{p',0}+y)^{2\alpha} }
    \leq \frac{1}{(1-2\alpha)^2} \sum_{\substack{p,p' \\ : n>p>p'\geq 0 }} \frac{1}{ d_{n,p}^{2\alpha-1} } \frac{1}{ d_{p,p'}^\alpha } \frac{1}{ d_{p,0}^{2\alpha-1} }
    \notag \\ &
    \leq \frac{1}{(1-2\alpha)^2} \sum_{\substack{p,p' \\ : n>p>p'\geq 0 }} \frac{1}{ d_{n,p}^{\alpha} } \frac{1}{ d_{p,p'}^\alpha } \frac{1}{ d_{p,0}^{\alpha} }
    \leq \frac{u^2}{ (1-2\alpha)^2 } \frac{1}{d_{n,0}^\alpha} 
    = \frac{u^2}{ (1-2\alpha)^2 } n^{-\alpha}.
\end{align}
This bound decays faster than $n^{2-2\alpha}$ when $\alpha \leq 2$. Hence, the term shown in Fig.\ref{fig_sup_xyz_2} (2) does not constitute the leading term for $\ave{ \mathcal{J}_n \mathcal{J}_0 }$ for $1<\alpha \leq 2$. However, this term should be dominant for $\alpha>2$.

Next, let us consider the contribution of the second term in (\ref{cal_xyz_1}):
\begin{align}
    &\sum_{\substack{k \geq n \\ \ell < n}} \sum_{\substack{k' \geq 0 \\ \ell' < 0}} \sum_{m,m'=1}^N \sum_{ \sigma,\tau \in P_{\mathrm{cyc}}^{xyz} } \frac{ J_{\sigma_1} J_{\sigma_2} J_{\tau_1} J_{\tau_2} }{16}
    \sum_{ \substack{ \{a,b,c\}=\{k,\ell,m\} \\ \{a',b',c'\}=\{k',\ell',m'\} }}   
    \frac{1}{ d_{a,c}^\alpha d_{b,c}^\alpha d_{a',c'}^\alpha d_{b',c'}^\alpha }
    \sum_{\substack{ \{i_1,...,i_6\}= \\ \{ a,b,c,a',b',c' \} \\ i_1 \in \{a,b,c\} \\ i_2 \in \{a',b',c'\} \\  i_3,...,i_6 \mathrm{:connected} }} \frac{ c_3 }{ d_{i_1,i_2}^\alpha } \frac{ c_2 }{ d_{i_3,i_4}^\alpha d_{i_4,i_5}^\alpha d_{i_5,i_6}^\alpha } 
    \label{cal_xyz_1_2}
\end{align}

Figure.\ref{fig_sup_xyz_3} comprehensively covers all the cases of clustering linkage in (\ref{cal_xyz_1_2}). Here, diagrams with identical topology are regarded as equivalent.
The existence of 17 topologically distinct graphs can be verified using Burnside’s lemma (also known as Cauchy–Frobenius lemma).

\begin{figure}[h]
\centering
\includegraphics[width=0.65\textwidth]{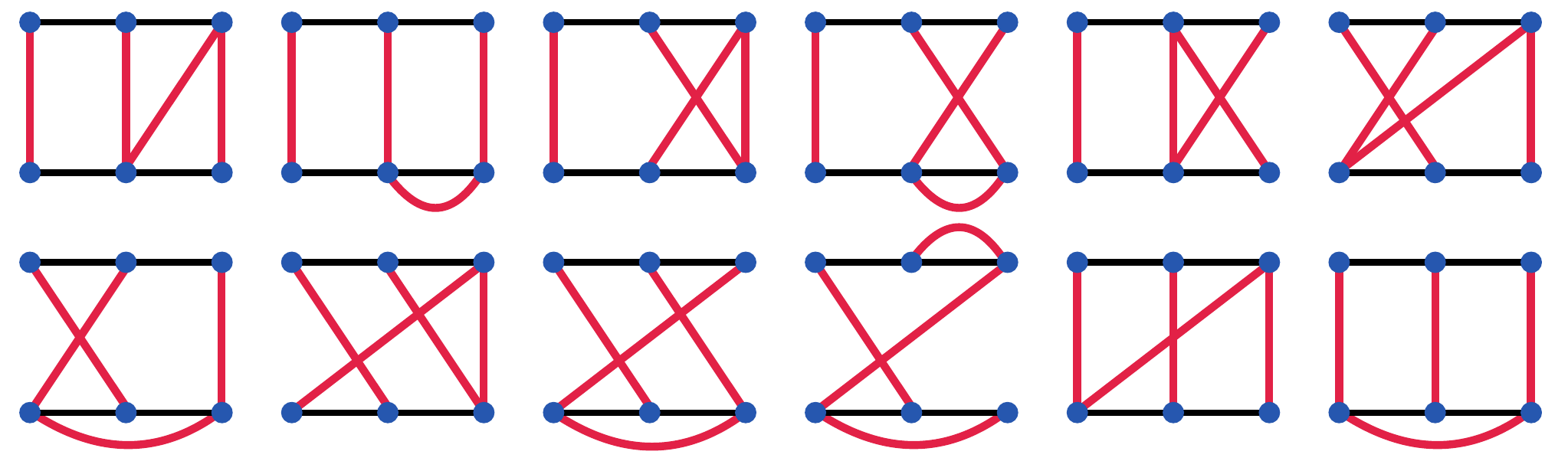}
\caption{The diagrammatic representation of the contribution of the second term in (\ref{cal_xyz_1}) to the upper bound of $\ave{ \mathcal{J}_n \mathcal{J}_0 }$ }
\label{fig_sup_xyz_3}
\end{figure}

Let us compute the upper bound of the contribution to $\ave{ \mathcal{J}_n \mathcal{J}_0 }$ in the following typical examples: (iii) the first row, fifth column from the left in Fig.\ref{fig_sup_xyz_3}, and (iv) the first row, sixth column from the left in Fig.\ref{fig_sup_xyz_3}.

In the case (iii), we obtain the following upper bound of the contribution to $\ave{ \mathcal{J}_n \mathcal{J}_0 }$:
\begin{align}
    &\sum_{\substack{k \geq n \\ \ell < n}} \sum_{\substack{k' \geq 0 \\ \ell' < 0}} \sum_{m,m'=1}^N \sum_{ \sigma,\tau \in P_{\mathrm{cyc}}^{xyz} } \frac{ J_{\sigma_1} J_{\sigma_2} J_{\tau_1} J_{\tau_2} }{16} \sum_{ \substack{ \{a,b,c\}=\{k,\ell,m\} \\ \{a',b',c'\}=\{k',\ell',m'\} }} \frac{1}{ d_{a,c}^\alpha d_{b,c}^\alpha d_{a',c'}^\alpha d_{b',c'}^\alpha } \frac{ c_3 }{ d_{a,a'}^\alpha } \frac{ c_2 }{ d_{b,c'}^\alpha d_{c',c}^\alpha d_{c,b'}^\alpha }
    \notag \\ &
    \leq \sum_{\substack{k \geq n \\ \ell < n}} \sum_{\substack{k' \geq 0 \\ \ell' < 0}} \sum_{m,m'=1}^N \sum_{ \sigma,\tau \in P_{\mathrm{cyc}}^{xyz} } \frac{ J_{\sigma_1} J_{\sigma_2} J_{\tau_1} J_{\tau_2} }{16} \sum_{ \substack{ \{a,b,c\}=\{k,\ell,m\} \\ \{a',b',c'\}=\{k',\ell',m'\} }} \frac{1}{ d_{a,c}^\alpha d_{b,c}^\alpha d_{a',c'}^\alpha d_{b',c'}^\alpha } \frac{ c_3 }{ d_{a,a'}^\alpha } \sum_{e,f=1}^N \frac{ c_2 }{ d_{b,e}^\alpha d_{e,f}^\alpha d_{f,b'}^\alpha }
    \notag \\ &
    \leq \sum_{ \sigma,\tau \in P_{\mathrm{cyc}}^{xyz} } \frac{ J_{\sigma_1} J_{\sigma_2} J_{\tau_1} J_{\tau_2} }{16} u^6 c_2 c_3 \sum_{\substack{k \geq n \\ \ell' < 0}} \frac{ 1 }{ d_{k,\ell'}^{2\alpha} }
    \leq c'' n^{2-2\alpha}, 
\end{align}
where $c''=u^6 c_2 c_3 / \{ (1-2\alpha)(2-2\alpha) \} \cdot \sum_{ \sigma,\tau \in P_{\mathrm{cyc}}^{xyz} } J_{\sigma_1} J_{\sigma_2} J_{\tau_1} J_{\tau_2}/16 $. The corresponding diagrammatic derivation is presented in the upper half of Figure.\ref{fig_sup_xyz_3_1}.

In the case (iv), we similarly obtain the following upper bound of the contribution to $\ave{ \mathcal{J}_n \mathcal{J}_0 }$:
\begin{align}
    &\sum_{\substack{k \geq n \\ \ell < n}} \sum_{\substack{k' \geq 0 \\ \ell' < 0}} \sum_{m,m'=1}^N \sum_{ \sigma,\tau \in P_{\mathrm{cyc}}^{xyz} } \frac{ J_{\sigma_1} J_{\sigma_2} J_{\tau_1} J_{\tau_2} }{16} \sum_{ \substack{ \{a,b,c\}=\{k,\ell,m\} \\ \{a',b',c'\}=\{k',\ell',m'\} }} \frac{1}{ d_{a,c}^\alpha d_{b,c}^\alpha d_{a',c'}^\alpha d_{b',c'}^\alpha } \frac{ c_3 }{ d_{a,b'}^\alpha } \frac{ c_2 }{ d_{c,a'}^\alpha d_{a',b}^\alpha d_{b,b'}^\alpha }
    \notag \\ &
    \leq \sum_{\substack{k \geq n \\ \ell < n}} \sum_{\substack{k' \geq 0 \\ \ell' < 0}} \sum_{m,m'=1}^N \sum_{ \sigma,\tau \in P_{\mathrm{cyc}}^{xyz} } \frac{ J_{\sigma_1} J_{\sigma_2} J_{\tau_1} J_{\tau_2} }{16} \sum_{ \substack{ \{a,b,c\}=\{k,\ell,m\} \\ \{a',b',c'\}=\{k',\ell',m'\} }} \frac{ \bar{a}^2 }{ d_{a,c}^\alpha d_{b',c'}^\alpha } \frac{ c_3 }{ d_{a,b'}^\alpha } \frac{ c_2 }{ d_{c,a'}^\alpha d_{a',b}^\alpha d_{b,b'}^\alpha }
    \notag \\ &
    \leq \sum_{ \sigma,\tau \in P_{\mathrm{cyc}}^{xyz} } \frac{ J_{\sigma_1} J_{\sigma_2} J_{\tau_1} J_{\tau_2} }{16} u^4 \bar{a}^2 c_2 c_3 \sum_{\substack{k \geq n \\ \ell' < 0}} \frac{ 1 }{ d_{k,\ell'}^{2\alpha} }
    \leq c''' n^{2-2\alpha}, 
\end{align}
where $c'''=u^4 \bar{a}^2 c_2 c_3 / \{ (1-2\alpha)(2-2\alpha) \} \cdot \sum_{ \sigma,\tau \in P_{\mathrm{cyc}}^{xyz} } J_{\sigma_1} J_{\sigma_2} J_{\tau_1} J_{\tau_2}/16 $. The corresponding diagrammatic derivation is presented in the lower half of Figure.\ref{fig_sup_xyz_3_1}.

\begin{figure}[h]
\centering
\includegraphics[width=0.5\textwidth]{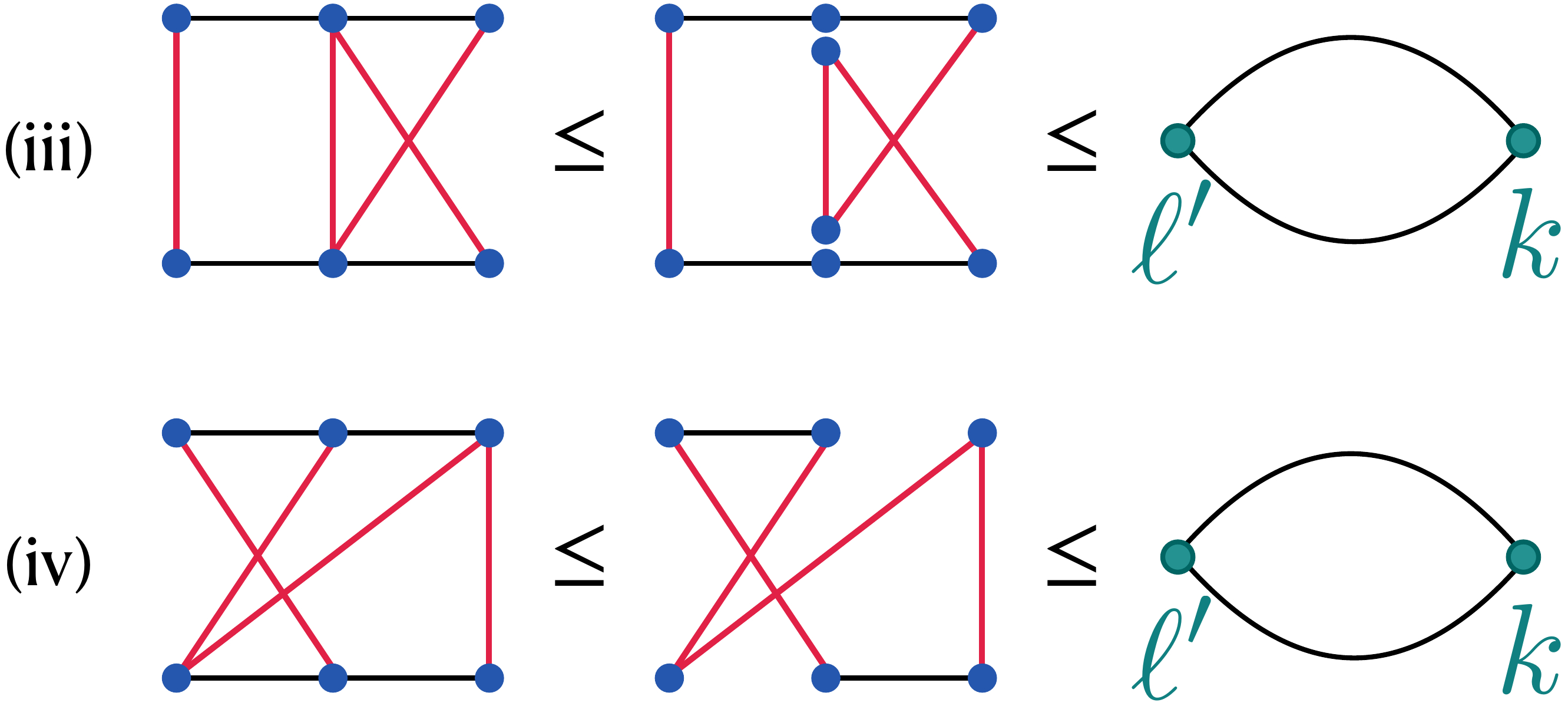}
\caption{Typical examples of the diagrammatic derivation for the  upper bound in the case Figure.\ref{fig_sup_xyz_3} }
\label{fig_sup_xyz_3_1}
\end{figure}

The same upper bound can be obtained for all remaining cases in Fig.\ref{fig_sup_xyz_3} as well:
\begin{align}
    &\sum_{\substack{k \geq n \\ \ell < n}} \sum_{\substack{k' \geq 0 \\ \ell' < 0}} \sum_{m,m'=1}^N \sum_{ \sigma,\tau \in P_{\mathrm{cyc}}^{xyz} } \frac{ J_{\sigma_1} J_{\sigma_2} J_{\tau_1} J_{\tau_2} }{16}
    \sum_{ \substack{ \{a,b,c\}=\{k,\ell,m\} \\ \{a',b',c'\}=\{k',\ell',m'\} }}   
    \frac{1}{ d_{a,c}^\alpha d_{b,c}^\alpha d_{a',c'}^\alpha d_{b',c'}^\alpha }
    \sum_{\substack{ \{i_1,...,i_6\}= \\ \{ a,b,c,a',b',c' \} \\ i_1 \in \{a,b,c\} \\ i_2 \in \{a',b',c'\} \\  i_3,...,i_6 \mathrm{:connected} }} \frac{ c_3 }{ d_{i_1,i_2}^\alpha } \frac{ c_2 }{ d_{i_3,i_4}^\alpha d_{i_4,i_5}^\alpha d_{i_5,i_6}^\alpha }
    \notag \\ &
    \leq \mathcal{O}(1) n^{2-2\alpha}
\end{align}

Lastly, let us consider the contribution of the third term in (\ref{cal_xyz_1}).
\begin{align}
    \sum_{\substack{k \geq n \\ \ell < n}} \sum_{\substack{k' \geq 0 \\ \ell' < 0}} \sum_{m,m'=1}^N \sum_{ \sigma,\tau \in P_{\mathrm{cyc}}^{xyz} } \frac{ J_{\sigma_1} J_{\sigma_2} J_{\tau_1} J_{\tau_2} }{16}
    \sum_{ \substack{ \{a,b,c\}=\{k,\ell,m\} \\ \{a',b',c'\}=\{k',\ell',m'\} }}   
    \frac{1}{ d_{a,c}^\alpha d_{b,c}^\alpha d_{a',c'}^\alpha d_{b',c'}^\alpha }
    \sum_{ \substack{ \{i_1,i_3,i_5\}=\{a,b,c\} \\ \{i_2,i_4,i_6\}=\{a',b',c'\} }}  \frac{ c_3 }{ d_{i_1,i_2}^\alpha } \frac{ c_3 }{ d_{i_3,i_4}^\alpha } \frac{ c_3 }{ d_{i_5,i_6}^\alpha }
    \label{cal_xyz_1_3}
\end{align}

Figure.\ref{fig_sup_xyz_4} encompasses all possible ways of clustering linkage, which are limited to two cases. Here, graphs with the same topology are considered equivalent. 

\begin{figure}[h]
\centering
\includegraphics[width=0.3\textwidth]{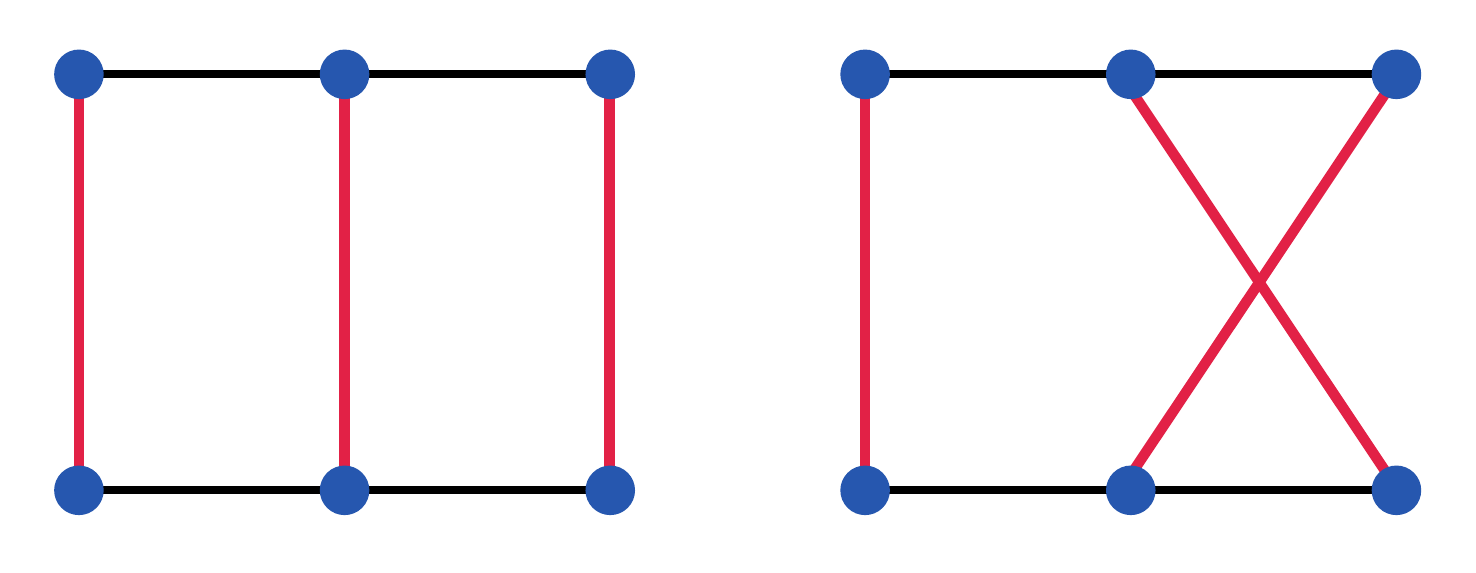}
\caption{The diagrammatic representation of the contribution of the third term in (\ref{cal_xyz_1}) to the upper bound of $\ave{ \mathcal{J}_n \mathcal{J}_0 }$ }
\label{fig_sup_xyz_4}
\end{figure}

For the diagram on the left in Fig.\ref{fig_sup_xyz_4}, as in the case of Figure.\ref{fig_sup_xyz_2}(1) and Figure.\ref{fig_sup_xyz_3}, we can demonstrate that it is bounded by 
$\mathcal{O}(1) n^{2-2\alpha}$, using (\ref{prop_1}), (\ref{prop_2}), and Lemma \ref{chain}.

Next, let us consider the diagram on the right. We classify the cases as follows: (a) when neither $b$ nor $b'$ (the two points at the right end of the graph) are equal to $k$ or $\ell'$, and (b) when $b=k$ and $b'=k'$. 

First, for the case (a), let $p \in \{b,b'\}$ be the point such that $p \neq k,\ell' $. Without loss of generality, we can assume $p=b$. Thus, we obtain the following upper bound for (\ref{cal_xyz_1_3}): 
\begin{align}
    &\sum_{\substack{k \geq n \\ \ell < n}} \sum_{\substack{k' \geq 0 \\ \ell' < 0}} \sum_{m,m'=1}^N \sum_{ \sigma,\tau \in P_{\mathrm{cyc}}^{xyz} } \frac{ J_{\sigma_1} J_{\sigma_2} J_{\tau_1} J_{\tau_2} }{16}
    \sum_{ \substack{ \{a,p,c\}=\{k,\ell,m\} \\ \{a',b',c'\}=\{k',\ell',m'\} \\ p \neq k,\ell'}}   
    \frac{1}{ d_{a,c}^\alpha d_{p,c}^\alpha d_{a',c'}^\alpha d_{b',c'}^\alpha } \frac{ c_3 }{ d_{a,a'}^\alpha } \frac{ c_3 }{ d_{p,c'}^\alpha } \frac{ c_3 }{ d_{b,c'}^\alpha }
    \notag \\ &
    \leq \sum_{\substack{k \geq n \\ \ell < n}} \sum_{\substack{k' \geq 0 \\ \ell' < 0}} \sum_{m,m'=1}^N \sum_{ \sigma,\tau \in P_{\mathrm{cyc}}^{xyz} } \frac{ J_{\sigma_1} J_{\sigma_2} J_{\tau_1} J_{\tau_2} }{16} \sum_{a,c,a',b',c'} \frac{ (c_3)^3 }{ d_{a,a'}^\alpha d_{a',c'}^\alpha d_{c',b'}^\alpha d_{b',c}^\alpha d_{c,a}^\alpha } \frac{ u }{ d_{c,c'}^\alpha }
    \notag \\ &
    \leq \sum_{\substack{k \geq n \\ \ell < n}} \sum_{\substack{k' \geq 0 \\ \ell' < 0}} \sum_{m,m'=1}^N \sum_{ \sigma,\tau \in P_{\mathrm{cyc}}^{xyz} } \frac{ J_{\sigma_1} J_{\sigma_2} J_{\tau_1} J_{\tau_2} }{16} \sum_{a,c,a',b',c'} \frac{ u \bar{a} (c_3)^3 }{ d_{a,a'}^\alpha d_{a',c'}^\alpha d_{c',b'}^\alpha d_{b',c}^\alpha d_{c,a}^\alpha }
    \notag \\ &
    \leq \sum_{ \sigma,\tau \in P_{\mathrm{cyc}}^{xyz} } \frac{ J_{\sigma_1} J_{\sigma_2} J_{\tau_1} J_{\tau_2} }{16} u^4 \bar{a} (c_3)^3 \sum_{\substack{k \geq n \\ \ell' < 0}} \frac{ 1 }{ d_{k,\ell'}^{2\alpha} }
    \leq c_a n^{2-2\alpha},
\end{align}
where $c_a = u^4 \bar{a} (c_3)^3 / \{ (1-2\alpha)(2-2\alpha) \} \cdot \sum_{ \sigma,\tau \in P_{\mathrm{cyc}}^{xyz} } J_{\sigma_1} J_{\sigma_2} J_{\tau_1} J_{\tau_2}/16 $. The corresponding diagrammatic derivation is illustrated in the upper half of the Figure.\ref{fig_sup_xyz_4_1}. 

Next, for the case (b): $b=k,b'=\ell'$, the following upper bound is derived:
\begin{align}
    &\sum_{\substack{k \geq n \\ \ell < n}} \sum_{\substack{k' \geq 0 \\ \ell' < 0}} \sum_{m,m'=1}^N \sum_{ \sigma,\tau \in P_{\mathrm{cyc}}^{xyz} } \frac{ J_{\sigma_1} J_{\sigma_2} J_{\tau_1} J_{\tau_2} }{16}
    \sum_{ \substack{ \{a,c\}=\{\ell,m\} \\ \{a',c'\}=\{k',m'\} }}   
    \frac{1}{ d_{a,c}^\alpha d_{k,c}^\alpha d_{a',c'}^\alpha d_{\ell',c'}^\alpha } \frac{ c_3 }{ d_{a,a'}^\alpha } \frac{ c_3 }{ d_{k,c'}^\alpha } \frac{ c_3 }{ d_{b,\ell'}^\alpha }
    \notag \\ &
    \leq \sum_{\substack{k \geq n \\ \ell < n}} \sum_{\substack{k' \geq 0 \\ \ell' < 0}} \sum_{m,m'=1}^N \sum_{ \sigma,\tau \in P_{\mathrm{cyc}}^{xyz} } \frac{ J_{\sigma_1} J_{\sigma_2} J_{\tau_1} J_{\tau_2} }{16}
    \sum_{c,c'} \frac{ (c_3)^3 }{ d_{k,c'}^\alpha d_{c',\ell'}^\alpha d_{\ell',c}^\alpha d_{c,k}^\alpha } \frac{ u^2 }{ d_{c,c'}^\alpha }
    \notag \\ &
    \leq \sum_{\substack{k \geq n \\ \ell < n}} \sum_{\substack{k' \geq 0 \\ \ell' < 0}} \sum_{m,m'=1}^N \sum_{ \sigma,\tau \in P_{\mathrm{cyc}}^{xyz} } \frac{ J_{\sigma_1} J_{\sigma_2} J_{\tau_1} J_{\tau_2} }{16}
    \sum_{c,c'} \frac{ u^2 \bar{a} (c_3)^3 }{ d_{k,c'}^\alpha d_{c',\ell'}^\alpha d_{\ell',c}^\alpha d_{c,k}^\alpha }
    \notag \\ &
    \leq \sum_{ \sigma,\tau \in P_{\mathrm{cyc}}^{xyz} } \frac{ J_{\sigma_1} J_{\sigma_2} J_{\tau_1} J_{\tau_2} }{16} u^4 \bar{a} (c_3)^3 \sum_{\substack{k \geq n \\ \ell' < 0}} \frac{ 1 }{ d_{k,\ell'}^{2\alpha} }
    \leq c_a n^{2-2\alpha}
    ,
\end{align}
with the corresponding diagrammatic derivation shown in the lower half of the Figure.\ref{fig_sup_xyz_4_1}. 

\begin{figure}[h]
\centering
\includegraphics[width=0.6\textwidth]{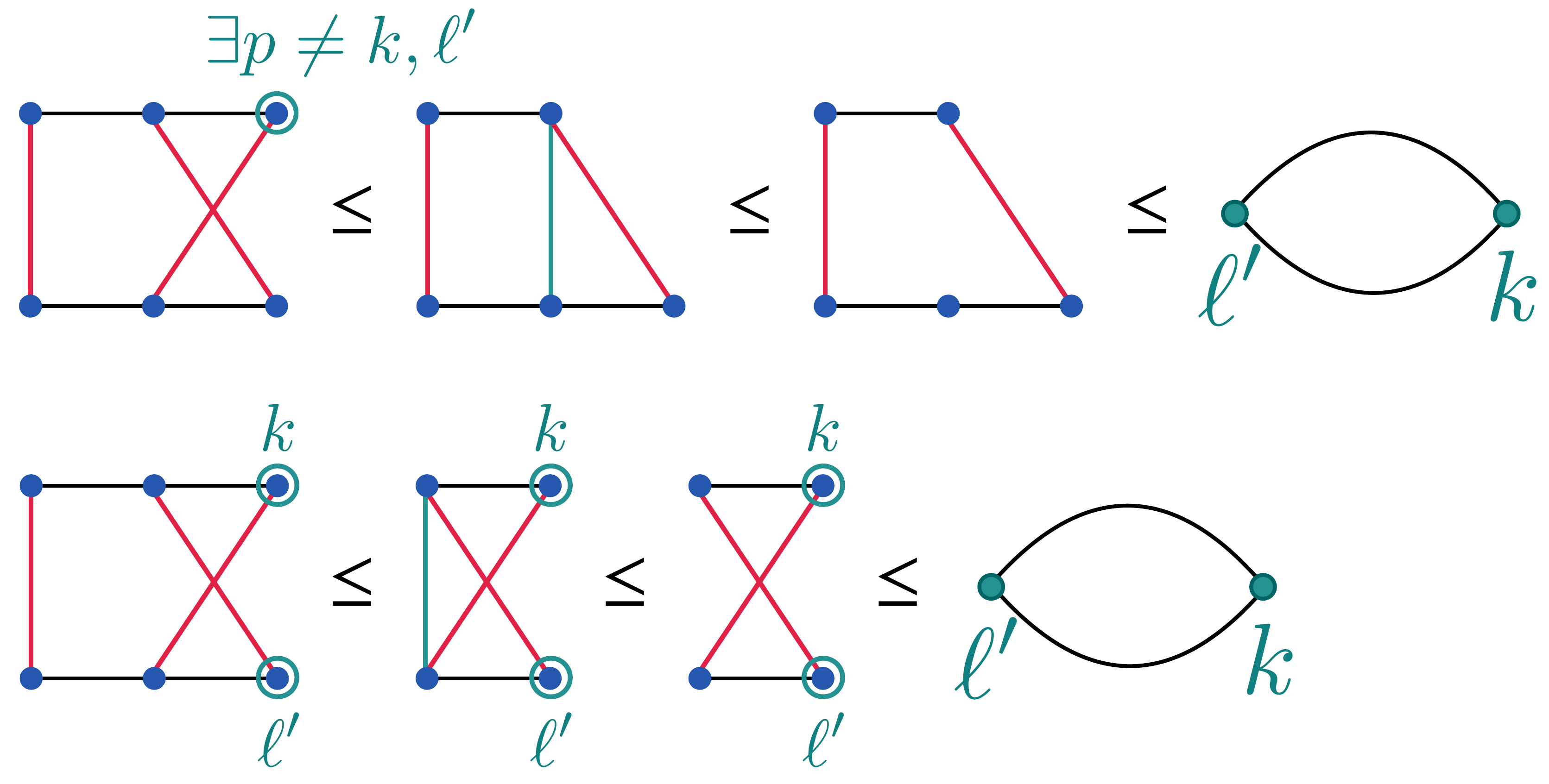}
\caption{A typical example of the diagrammatic derivation of the upper bound of the third term in (\ref{cal_xyz_1})}
\label{fig_sup_xyz_4_1}
\end{figure}

From these results, we can also derive the following upper bound for the contribution of the third term in  (\ref{cal_xyz_1}):
\begin{align}
    &\sum_{\substack{k \geq n \\ \ell < n}} \sum_{\substack{k' \geq 0 \\ \ell' < 0}} \sum_{m,m'=1}^N \sum_{ \sigma,\tau \in P_{\mathrm{cyc}}^{xyz} } \frac{ J_{\sigma_1} J_{\sigma_2} J_{\tau_1} J_{\tau_2} }{16}
    \sum_{ \substack{ \{a,b,c\}=\{k,\ell,m\} \\ \{a',b',c'\}=\{k',\ell',m'\} }}   
    \frac{1}{ d_{a,c}^\alpha d_{b,c}^\alpha d_{a',c'}^\alpha d_{b',c'}^\alpha }
    \sum_{ \substack{ \{i_1,i_3,i_5\}=\{a,b,c\} \\ \{i_2,i_4,i_6\}=\{a',b',c'\} }}  \frac{ c_3 }{ d_{i_1,i_2}^\alpha } \frac{ c_3 }{ d_{i_3,i_4}^\alpha } \frac{ c_3 }{ d_{i_5,i_6}^\alpha }
    \notag \\ &
    \leq \mathcal{O}(1) n^{2-2\alpha}
\end{align}

Based on the above discussion, it has been demonstrated that all terms in (\ref{cal_xyz_1}) can be bounded as follows:
\begin{align}
    \abs{\ave{ \mathcal{J}_n \mathcal{J}_0 }} < \mathcal{O}(1) n^{2-2\alpha} + \mathcal{O}(1) n^{-\alpha}
\end{align}
Hence, for $1<\alpha \leq 2$, the equal-time current correlation $\abs{\ave{ \mathcal{J}_n \mathcal{J}_0 }} $ is bounded by  $\mathcal{O}(1) n^{2-2\alpha}$, while for $\alpha>2$, $\abs{\ave{ \mathcal{J}_n \mathcal{J}_0 }} $ is bounded by $\mathcal{O}(1) n^{-\alpha}$. 
\end{proof}

\section{Higher dimensional case}

In this section, we provide the sufficient condition of normal energy diffusion for prototypical long-range spin systems in $D$ dimensions ($D \geq 2$). The formulation of the Green-Kubo formula in $D$-dimensional long-range systems is more complicated than in the one-dimensional case. Hence, first, we derive the Green-Kubo formula and the current expression in $D$ dimensions. Next, using the cumulant power-law clustering theorem, we prove the upper bound for the equal-time total current correlations in $D$-dimensional case. We provide separate proofs for the long-range transverse Ising, XY and XYZ model.

\subsection{Formula of the diffusion constant in higher dimensions}
In this subsection, we formulate the diffusion constant of the energy using the Green-Kubo formula of thermal conductivity. We first consider the continuity equation with respect to energy. The energy current is defined by the collection of energy transmissions across the hyperplane. To this end, it becomes necessary to define subspaces that partition the $D$-dimensional space into two halves, for which we provide a constructive definition. The continuity equation for the local energy $\hat{h}_{\bm{n}}$ at site $\bm{n}$ can be written as follows:
\begin{align}
    \partial_t \hat{h}_{\bm{n}} &= \{ \hat{h}_{\bm{n}} \, , H \} 
    = \sum_{\bm{m}} \{ \hat{h}_{\bm{n}} \, , \hat{h}_{\bm{m}} \} 
    = \sum_{\bm{m} } t_{\bm{n} \leftarrow \bm{m}}
    = \sum_{\bm{r} \neq \bm{0} } t_{ \bm{n} \leftarrow \bm{n}+\bm{r} }
    \notag \\
    &= \sum_{\bm{r} \in \mathcal{D}^+ } t_{ \bm{n} \leftarrow  \bm{n}+\bm{r} } - t_{ \bm{n} -\bm{r} \leftarrow \bm{n} }
    = - \sum_{\bm{r} \in \mathcal{D}^+ }
    \nabla_{\bm{n};\bm{r}} \mathcal{J}_{\bm{n};\bm{r}}
    \label{conti_D},
\end{align}
where energy transmission operator $t_{i \leftarrow j}$ is defined as 
\begin{align}
    t_{\bm{i} \leftarrow \bm{j} } := \{ \hat{h}_{\bm{i}} ,  \hat{h}_{\bm{j}} \}
    ,
\end{align}
and 
\begin{align}
    \nabla_{\bm{n};\bm{r}} A_{\bm{n};\bm{r}} &:= A_{\bm{n}+\bm{r};\bm{r}} - A_{\bm{n};\bm{r}}
    ,
    \\
    \mathcal{J}_{\bm{n};\bm{r}} &:= - t_{ \bm{n}-\bm{r} \leftarrow \bm{n} }
    .
\end{align}
The set $\mathcal{D}^+$ is defined as one of the two point-symmetric divisions of all lattice points in the space excluding the point $\bm{0}$. Here, we provide a constructive definition as follows:

Let the entire space be $D$-demensional torus $\mathbb{T}^D := (\mathbb{Z}/N\mathbb{Z})^D$ and let us specify an ordering $\{i_1,i_2,...,i_D\} = \{1,...,D\}$. We take a coordinate setting ${\bm n}$ to the origin, and consider the vector ${\bm x}$. 
First, we divide the entire space into 3 regions as follows:
\begin{align}
    \mathcal{D}_{i_1}^0 &:= \{ \bm{x} \in \mathbb{T}^D \mid x_{i_1} = 0 \} \, , 
    \\
    \mathcal{D}_{i_1}^+ &:= \{ \bm{x} \in \mathbb{T}^D \mid x_{i_1} > 0 \} \, , 
    \\
    \mathcal{D}_{i_1}^- &:= \{ \bm{x} \in \mathbb{T}^D \mid x_{i_1} < 0 \} \, .
\end{align}
Next, we divide $\mathcal{D}^0_{i_1}$ into 3 regions as follows:
\begin{align}
    \mathcal{D}_{i_1;i_2}^0 &:= \{ \bm{x} \in \mathcal{D}_{i_1}^0 \mid x_{i_2} = 0 \} \, , 
    \\
    \mathcal{D}_{i_1;i_2}^+ &:= \{ \bm{x} \in \mathcal{D}_{i_1}^0 \mid x_{i_2} > 0  \} \, , 
    \\
    \mathcal{D}_{i_1;i_2}^- &:= \{ \bm{x} \in \mathcal{D}_{i_1}^0 \mid x_{i_2} < 0 \} \, .
\end{align}
Furthermore, we devide $\mathcal{D}_{i_1;i_2}^0$ into 3 regions as follows:
\begin{align}
    \mathcal{D}_{i_1;i_2;i_3}^0 &:= \{ \bm{x} \in \mathcal{D}_{i_1;i_2}^0 \mid x_{i_3} = 0 \} \, ,   \\
    \mathcal{D}_{i_1;i_2;i_3}^+ &:= \{ \bm{x} \in \mathcal{D}_{i_1;i_2}^0 \mid x_{i_3} > 0 \} \, ,   \\
    \mathcal{D}_{i_1;i_2;i_3}^- &:= \{ \bm{x} \in \mathcal{D}_{i_1;i_2}^0 \mid x_{i_3} < 0 \} \, .
\end{align}
By repeating this process, we define $\mathcal{D}_{i_1; ... ;i_e}^{+,0,-} \, (e=1,...,D)$, which satisfies the following properties:
\begin{align}
    \mathcal{D}_{i_1;i_2;...;i_D}^0 &= \{ \bm{0} \}
    \\
    \bigcup_{ \substack{ a=+,- \\ e=1,...,D } } \mathcal{D}_{i_1;...;i_e}^a 
    &= \mathbb{T}^D \setminus \{ \bm{0} \}
    \\
    \mathcal{D}_{i_1;...;i_e}^a \cap \mathcal{D}_{i_1;...;i_f}^b 
    &= \varnothing  , \quad ( a,b \in \{+,-\} , \, e,f \in \{ 1,...,D \} )
\end{align}
Since $\mathcal{D}^{+}_{i_1;...;i_e}$ and $\mathcal{D}^{-}_{i_1;...;i_e}$ are point-symmetric about $\bm{0}$, let us define 
$
    \mathcal{D}^+ := \bigcup_{e=1,...,D} \mathcal{D}^+_{i_1;...;i_e}, 
    \mathcal{D}^- := \bigcup_{e=1,...,D} 
    \mathcal{D}^-_{i_1;...;i_e}
    $, and then we have 
\begin{align}
    \mathcal{D}^+ \cup \mathcal{D}^- &= \mathbb{T}^D \setminus \{ \bm{0} \}
    , \quad
    \mathcal{D}^+ \cap \mathcal{D}^- = \varnothing
\end{align}
and $\mathcal{D}^+$ and $\mathcal{D}^-$ are point-symmetric about the origin $\bm{0}$. Thus, we obtain a concrete construction of $\mathcal{D}^+$. Note that $\mathcal{D}^+$ is labeled by $(i_1,...,i_D)$: $\mathcal{D}^+ = \mathcal{D}^+(i_1,...,i_D)$.

For the two-dimensional case, we demonstrate an example of the explicit construction of $\mathcal{D}^+$ and $\mathcal{D}^-$ in Fig.\ref{fig_sup_d_GK_1}. We denote $\mathcal{D}^+_{i_1} = \{ \bm{x} \mid x_{i_1} >0 \}$, $\mathcal{D}^-_{i_1} = \{ \bm{x} \mid x_{i_1} <0 \}$, $\mathcal{D}^+_{i_1;i_2} = \{ \bm{x} \mid x_{i_1}=0, x_{i_2} >0 \}$ and $\mathcal{D}^-_{i_1;i_2} = \{ \bm{x} \mid x_{i_1}=0, x_{i_2} <0 \}$ and define $\mathcal{D}^a \, (a=+,-)$ as $\mathcal{D}^a = \mathcal{D}^a_{i_1} \cup \mathcal{D}^a_{i_1;i_2}$. As shown in Fig.\ref{fig_sup_d_GK_1}, $\mathcal{D}^+$ and $\mathcal{D}^-$ divide the entire lattice into two halves, except for the origin $\bm{0}$. 
\\

\begin{figure}[h]
\centering
\includegraphics[width=0.5\textwidth]{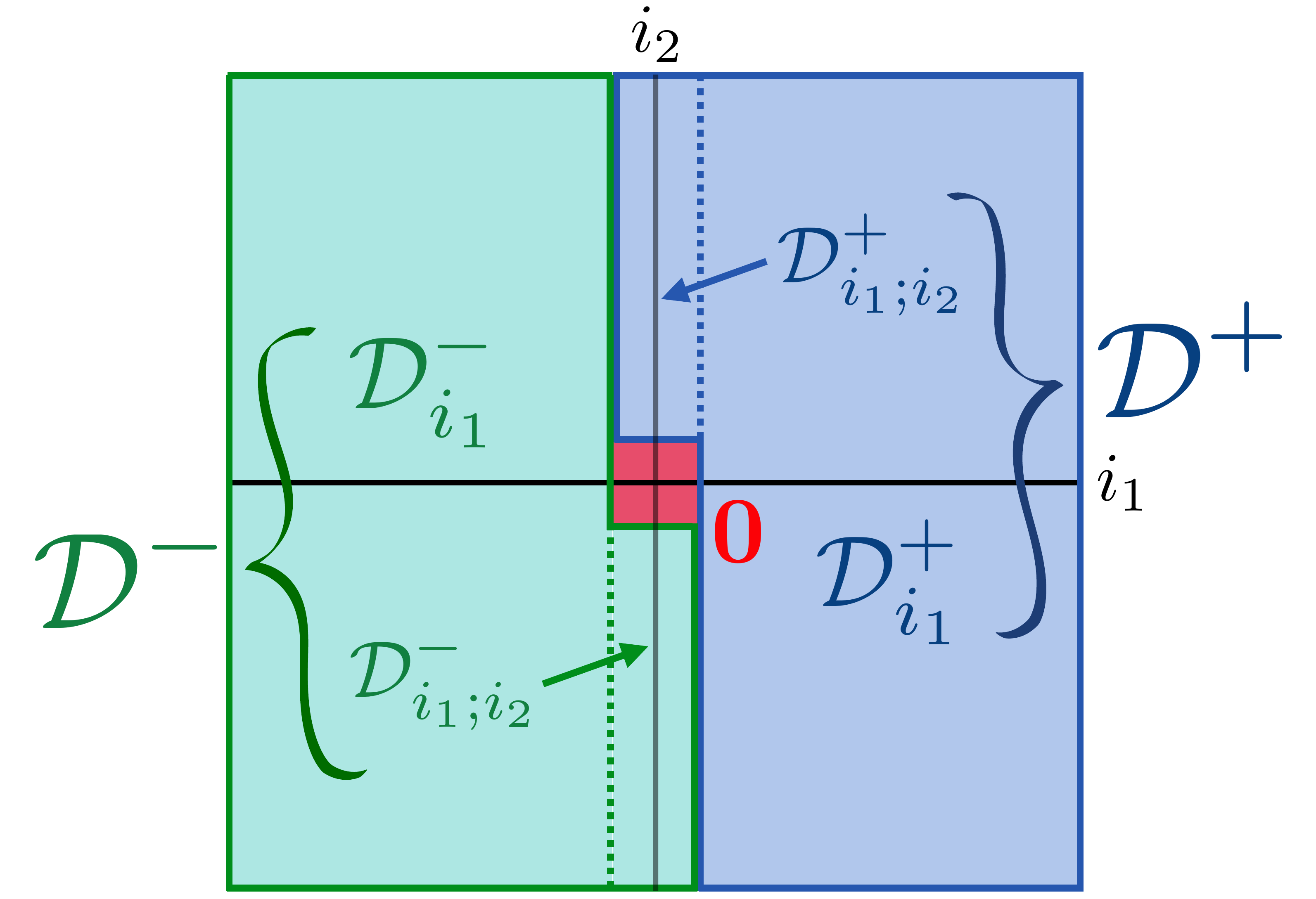}
\caption{The schematic definition of $\mathcal{D}^{+} = \mathcal{D}^{+}_{i_1} \cup \mathcal{D}^{+}_{i_1;i_2}$ and $\mathcal{D}^{-} = \mathcal{D}^{-}_{i_1} \cup \mathcal{D}^{-}_{i_1;i_2}$ in $2$-dimensional case. }
\label{fig_sup_d_GK_1}
\end{figure}

In the diffusive regime, the mean squared displacement $S(t)$ is related to the diffusion coefficient $\mathfrak{D}$ as follows.
\begin{align}
    S(t) 
    = \ave{ \sum_{i=1}^D n_i^2(t) } 
    = \sum_{\bm{n}} \sum_{i=1}^D n_i^2 P(\bm{n},t)
    = {1\over k_{\rm B} T^2 c_{\rm V}} \sigma(t)
    \sim 2 \mathfrak{D} t ,
\end{align}
where $P(\bm{n},t)$ is probability density distribution of displacement $\bm{n}$ defined as 
    $P(\bm{n},t) 
    = \mathcal{N}^{-1} \langle \delta \hat{h}_{\bm{n}}(t) \delta \hat{h}_{\bm{0}}(0) \rangle     
    , \mathcal{N}:= \sum_{\bm{n}} \langle \delta \hat{h}_{\bm{n}} \delta \hat{h}_{\bm{0}} \rangle $
, and the spatial profile of energy $\sigma(t)$ is defined as
    $\sigma(t) := \sum_{i=1}^D \sum_{ \bm{n} } n_i^2 \langle \delta \hat{h}_{\bm{n}}(t) \delta \hat{h}_{\bm{0}}(0)  \rangle $. 

We now derive the relation between the spatial profile of the energy and the energy current correlation.
We divide the space by
$\mathcal{D}^+ := \mathcal{D}^+ (i,i_2,...,i_D) $.
By utilizing the continuity equation, we obtain the following expression for $\sigma(t)$: 
\begin{align}
    \sigma(t) 
    &= \sum_{i=1}^D \sum_{\bm{n}} n_i^2 \langle \delta \hat{h}_{\bm{n}}(t) \delta \hat{h}_{\bm{0}}(0) \rangle
    = \sum_{i=1}^D \sum_{\bm{n}} n_i^2 \ave{ - \bigg( \int^t_0 ds \sum_{\bm{r} \in \mathcal{D}^+} \nabla_{\bm{r}} \mathcal{J}_{\bm{n}:\bm{r}}(s) \bigg) \delta \hat{h}_{\bm{0}}(0) }
    \notag \\
    &= \sum_{i=1}^D \sum_{\bm{n}} \sum_{\bm{r} \in \mathcal{D}^+} ( 2 r_i n_i + r_i^2 ) \int^t_0 ds \langle \mathcal{J}_{\bm{n}:\bm{r}}(s) \delta \hat{h}_{\bm{0}}(0) \rangle
    = \sum_{i=1}^D \sum_{\bm{n}} \sum_{\bm{r} \in \mathcal{D}^+ } ( 2 r_i n_i + r_i^2 ) \int^t_0 ds \langle \mathcal{J}_{\bm{n}:\bm{r}}(0) \delta \hat{h}_{\bm{0}}(-s) \rangle
    \notag \\
    &= \sum_{i=1}^D \sum_{\bm{n}} \sum_{\bm{r} \in \mathcal{D}^+} ( 2 r_i n_i + r_i^2 ) \int^t_0 ds \ave{ \mathcal{J}_{\bm{n}:\bm{r}}(0) \bigg( - \sum_{\bm{r}'\in \mathcal{D}^+ } \int^s_0 ds' \nabla_{\bm{r}'} \mathcal{J}_{\bm{0};\bm{r}'}(-s') \bigg) }
    \notag \\
    &= \sum_{i=1}^D \sum_{\bm{n}} \sum_{\bm{r},\bm{r}' \in \mathcal{D}^+ } ( 2 r_i n_i + r_i^2 ) \int^t_0 ds \ave{ - ( \mathcal{J}_{\bm{n};\bm{r}}(0) - \mathcal{J}_{\bm{n}-\bm{r}';\bm{r}}(0) ) \bigg( \int^s_0 ds' \mathcal{J}_{\bm{0};\bm{r}'}(-s') ) }
    \notag \\
    &= \sum_{i=1}^D \sum_{\bm{n}} \sum_{\bm{r},\bm{r}' \in \mathcal{D}^+ } 2 r_i r_{i}' \int^t_0 ds \int^s_0 ds' \ave{ \mathcal{J}_{\bm{n};\bm{r}}(0) \mathcal{J}_{\bm{0};\bm{r}'}(-s') }
    = \sum_{i=1}^D \sum_{\bm{n}} \sum_{\bm{r},\bm{r}' \in \mathcal{D}^+ } 2 r_i r_{i}' \int^t_0 ds \int^s_0 ds' \ave{ \mathcal{J}_{\bm{n};\bm{r}}(s') \mathcal{J}_{\bm{0};\bm{r}'}(0) }
    \notag \\
    &= \sum_{i=1}^D \frac{1}{N^D} \ave{ \bigg( \int^t_0 ds \sum_{\bm{n}} \sum_{\bm{r} \in \mathcal{D}^+_{i}} r_i \mathcal{J}_{\bm{n};\bm{r}}(s) \bigg) \bigg( \int^t_0 ds' \sum_{\bm{n}'} \sum_{\bm{r}'\in \mathcal{D}^+ } r_{i}' \mathcal{J}_{\bm{n}';\bm{r}'}(s') \bigg) }
    \notag \\
    &= \sum_{i=1}^D \frac{1}{N^D} \ave{ \bigg( \int^t_0 ds \sum_{n_i} {\cal J}^{(i)}_{ n_i \mid n_i -1 }(s) \bigg) \bigg( \int^t_0 ds' \sum_{n'_i} \mathcal{J}^{(i)}_{ n'_{i} \mid n'_{i} -1 }(s') \bigg) }
    ,
\end{align}
where we define the current ${\cal J}^{(i)}_{ n_i \mid n_i -1 }$ in the $i$th direction as 
\begin{align}
    {\cal J}^{(i)}_{ n_i \mid n_i -1 } := \sum_{\substack{ \bm{k}: k_i \geq n_i \\ \bm{\ell}: l_i < n_i }} t_{ \bm{k} \leftarrow \bm{\ell} }
    .
\end{align}
The current ${\cal J}^{(i)}_{ n_i \mid n_i -1 }$ is interpreted as the collection of all energy transmission $t_{ \bm{k} \leftarrow \bm{\ell} }$ that cross the boundary (hyperplane) between $\bm{n}-\bm{e}_i$ and $\bm{n}$, as shown in Fig.\ref{fig_sup_d_GK_2}. 

\begin{figure}[h]
\centering
\includegraphics[width=0.4\textwidth]{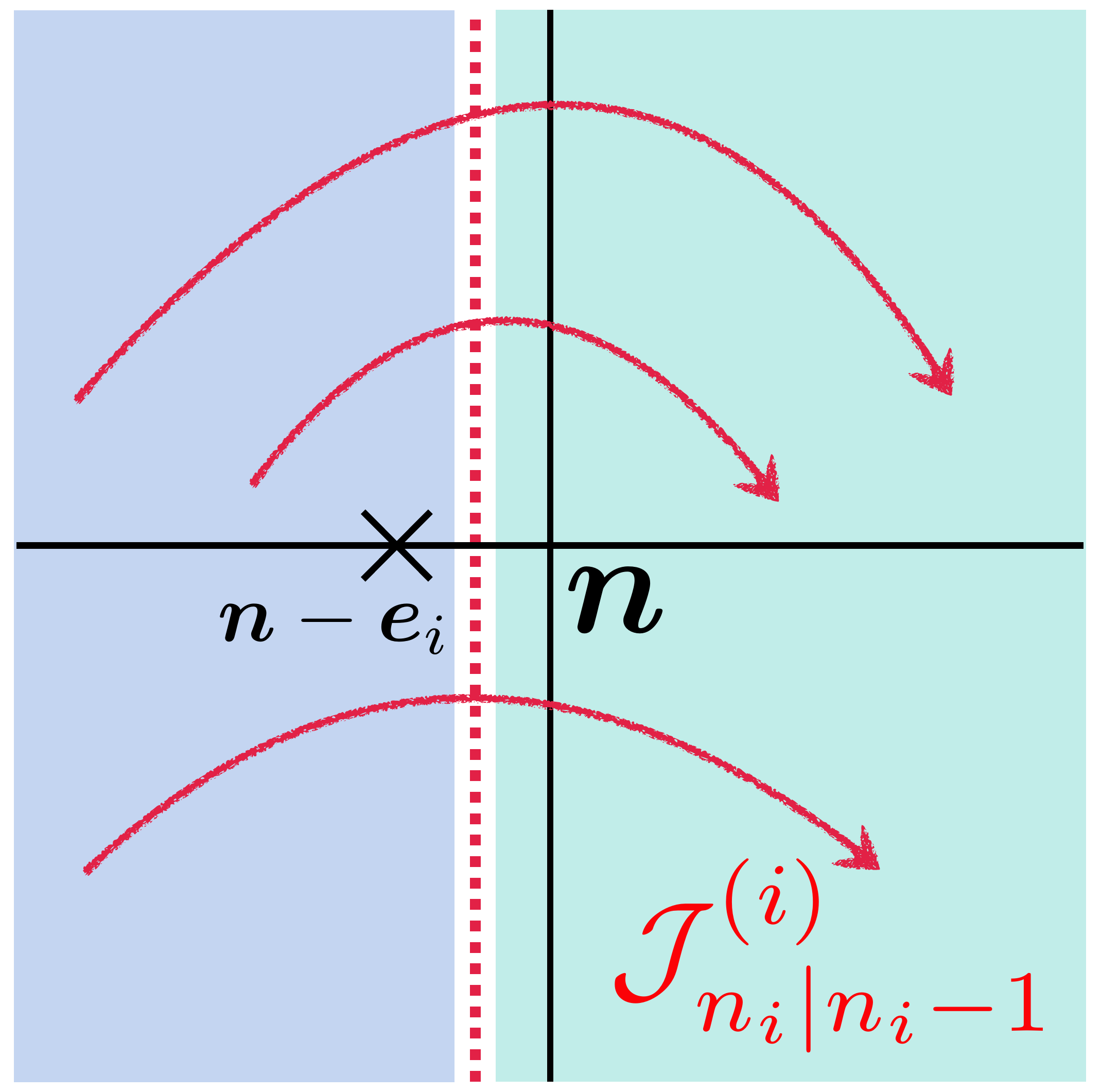}
\caption{The schematic definition of the current ${\cal J}^{(i)}_{ n_i \mid n_i -1 }$ in $i$th direction. }
\label{fig_sup_d_GK_2}
\end{figure}

Thus, when Fourier's law holds, we obtain the following Green-Kubo formula in $D$ dimensions:
\begin{align}
    \mathfrak{D} 
    &= \lim_{ t \rightarrow \infty } \lim_{ N \rightarrow \infty } {1\over k_{\rm B} T^2 c_{\rm V}} \int^t_0 ds C^{(D)}_N(s) ,
    \\
    C^{(D)}_N (s) 
    &= N^{1-D} \sum_{i=1}^D \sum_{n_i} \ave{ {\cal J}^{(i)}_{ n_i \mid n_i -1 } (s) {\cal J}^{(i)}_{ 0 \mid -1 } (0)}
    .
\end{align}

In $D$ dimensions, the prototypical long-range spin systems are also believed to be non-integrable, and hence, we assume $C^{(D)}_N(t)$ decays rapidly in time. 
Based on the premise that the amplitude of equal-time current correlations determines the energy diffusion, we below investigate the amplitude of $C^{(D)}_N(0)$, which is interpreted as the equal-time total current correlation.

\subsection{The amplitude of the equal-time total current correlations}

In this subsection, using Theorem \ref{cumulantpowerlawtheorem}, we prove the upper bound for the equal-time total current correlation $C^{(D)}_N(0)$ in $D$ dimensions. With the assumption that $c_V$ is finite and $C^{(D)}_N(t)$ decay rapidly in time, the amplitude of $C^{(D)}_N(0)$ determines the sufficient condition for normal energy diffusion.  

\subsubsection{Transverse Ising model}

We consider the $D$-dimensional long-range interacting transverse Ising model, which is described by the following Hamiltonian
\begin{equation}
    H=  - \sum_{\bm{i} \neq \bm{j} } J \frac{ S^z_{\bm{i}} S^z_{\bm{j}} }{d_{ \bm{i},\bm{j} }^\alpha} + \sum_{\bm{i}} h S^x_{\bm{i}} 
    .
\end{equation}
The local energy at site $\bm{n}$ is defined as 
\begin{equation}
    \hat{h}_{\bm{n}} := - \sum_{\bm{j} \neq \bm{n} } \frac{J}{2}  \frac{S^z_{\bm{n}} S^z_{\bm{j}}}{d_{\bm{n},\bm{j}}^\alpha} + h S^x_{\bm{n}}
    ,
\end{equation}
and current ${\cal J}^{(i)}_{ n_i \mid n_i -1 }$ is described as follows
\begin{align}
    {\cal J}^{(i)}_{ n_i \mid n_i -1 } 
    &= \sum_{\substack{ \bm{k}: k_i \geq n_i \\ \bm{\ell}: \ell_i < n_i }} t_{ \bm{k} \leftarrow \bm{\ell} }
    \\
    t_{\bm{k} \leftarrow \bm{\ell}} 
    &= \{ \hat{h}_{\bm{k}} , \hat{h}_{\bm{\ell}} \} 
    = \frac{Jh}{2} \frac{ S^z_{\bm{k}} S^y_{\bm{\ell}} - S^y_{\bm{k}} S^z_{\bm{\ell}} }{ d_{\bm{k},\bm{\ell}}^\alpha }
    ,
\end{align}
 where $\{..., ...\}$ is spin Poisson bracket, defined as
$\{ A_1 , A_2 \} = \sum_i \varepsilon_{\sigma\tau\upsilon} (\partial A_1 / \partial S^\sigma_{\bm{i}} ) (\partial A_2 / \partial S^\tau_{\bm{i}}) S^\upsilon_{\bm{i}}$.

\begin{theorem}
    \label{thm_d_ising}
    In the $D$-dimensional long-range interacting transverse Ising model ($D \geq 2$), under the condition 
        $\alpha > D/2 +1$, $\alpha >D$ and $\beta < \beta_c$,
    $C^{(D)}_N (0)$ converges to the finite value independent of the system size $N$.
    It immediately follows that $C^{(D)}_N (0)$ is finite as long as $\alpha >D$. 
\end{theorem}

\begin{proof}

$C^{(D)}_N (0)$ is expressed as follows:
\begin{align}
    C^{(D)}_N (0) = N^{1-D} \sum_{i=1}^D \sum_{n_i} \ave{ {\cal J}^{(i)}_{ n_i \mid n_i -1 } {\cal J}^{(i)}_{ 0 \mid -1 } } 
\end{align}

In exactly the same manner as the one-dimensional case, utilizing Theorem \ref{cumulantpowerlawtheorem} and the symmetry derived from the Hamiltonian, the following upper bound for the current correlation can be obtained.
\begin{align}
    \abs{ \ave{ {\cal J}^{(i)}_{ n_i \mid n_i -1 } {\cal J}^{(i)}_{ 0 \mid -1 } } }
    \leq \mathcal{O}(1) \sum_{\substack{ \bm{k}: k_i \geq n_i \\ \bm{\ell}': \ell'_i < 0}} \frac{1}{ d_{k,\ell'}^{2\alpha} }
\end{align}

For $D \geq 2$, the right-handed side is bounded as follows: 
\begin{align}
    \sum_{\substack{ \bm{k}: k_i \geq n_i \\ \bm{\ell}': \ell'_i < 0}} \frac{1}{ d_{k,\ell'}^{2\alpha} }
    &\leq \int_{\substack{ x_i \geq n_i \\ y_i < 0 }} dx_i dy_i \prod_{k \neq i} \int_{-N/2}^{N/2} dx_k dy_k \Big\{ ( x_i -y_i )^2 + \sum_{j \neq i} ( x_j -y_j )^2 \Big\}^{-2\alpha/2}
    \notag \\
    &\leq \Gamma_{D-1} \int^{N/2}_{-N/2} d^{D-1} \bm{u} \int_{\substack{ x_i \geq n_i \\ y_i <0 }} dx_i dy_i \int^\infty_0 d\xi \xi^{D-2} \{ (x_i -y_i )^2 + \xi^2 \}^{-2\alpha/2}
    \notag \\
    &\leq \frac{ \Gamma_{D-1} B(\frac{D-1}{2},\alpha-\frac{D-1}{2}) }{ 2 } N^{D-1} \int^{N/2}_0 dx dy (n_i +x +y)^{D-1-2\alpha} 
    \notag \\
    &\leq \frac{ \Gamma_{D-1} B(\frac{D-1}{2},\alpha-\frac{D-1}{2}) }{ 2 (D-2\alpha)(D+1-2\alpha) } N^{D-1} n_i^{D+1-2\alpha},
\end{align}
where $u_j := (x_j + y_j)/2$ for $j \neq i$, $\xi := \{ \sum_{j \neq i} (x_j+y_j)^2 \}^{1/2}$, $\Gamma_{D-1}$ is the coefficient of $(D-2)$-dimensional spherical surface area, and $B(x,y)$ is Beta function. 

Thus, the system size dependence of $\sigma_{ii}$ can be obtained straightforwardly as follows:
\begin{align}
    \abs{ C^{(D)}_N (0) } \leq \mathcal{O}(1) N^{D+2-2\alpha} + {\rm const.}
\end{align}

Hence, we show that $C^{(D)}_N (0)$ converges to a finite value under the condition $\alpha > D/2 + 1$.
\end{proof}

\subsubsection{XY model}
We consider the one-dimensional long-range interacting XY model, which is described by the following Hamiltonian
\begin{equation}
    H= - \sum_{ \bm{i} \neq \bm{j} } \Big[ J_x \frac{S^x_{\bm{i}} S^x_{\bm{j}}}{d_{\bm{i},\bm{j}}^\alpha} + J_y \frac{S^y_{\bm{i}} S^y_{\bm{j}}}{d_{\bm{i},\bm{j}}^\alpha} \Big]
    .
\end{equation}
The local energy at site $\bm{n}$ is defined as 
\begin{equation}
    \hat{h}_{\bm{n}} := - \sum_{ \bm{j} \neq \bm{n} } \sum_{\sigma =x,y } \frac{J_\sigma}{2} \frac{S^\sigma_{\bm{n}} S^\sigma_{\bm{j}}}{d_{\bm{n},\bm{j}}^\alpha} 
    ,
\end{equation}
and the current ${\cal J}^{(i)}_{ n_i \mid n_i -1 }$ is described as follows
\begin{align}
    {\cal J}^{(i)}_{ n_i \mid n_i -1 }
    &= - \sum_{\substack{\bm{k} : k_i \geq n_i \\ \bm{\ell} :  \ell_i < n_i }} t_{\bm{k} \leftarrow \bm{\ell} }
    \\
    t_{\bm{k} \leftarrow \bm{\ell}} 
    &= \{ \hat{h}_{\bm{k}} , \hat{h}_{\bm{\ell}} \} 
    = \sum_{\bm{m}} \frac{J_{x} J_{y}}{4} \sum_{ \{\bm{a},\bm{b},\bm{c}\}=\{\bm{k},\bm{\ell},\bm{m}\} } \tilde{\varepsilon}_{\bm{abc}} \frac{ S^{x}_{\bm{a}} S^{y}_{\bm{b}} S^{z}_{\bm{c}} }{d_{\bm{a},\bm{c}}^\alpha d_{\bm{b},\bm{c}}^\alpha}
    ,
\end{align}
where $\sum_{\sigma \in P^{xyz}_{\mathrm{cyc}}}$ is defined as sums over all circular permutation of the variables $(x,y,z)$, and $\tilde{\varepsilon}_{\bm{abc}}$ is defined as follows:
\begin{align}
    \tilde{\varepsilon}_{\bm{abc}}=
    \begin{cases}
        +1 \, , \quad &\text{for} \, \,
        (\bm{a},\bm{b},\bm{c})=(\bm{k},\bm{\ell},\bm{m}),(\bm{m},\bm{\ell},\bm{k}),(\bm{k},\bm{m},\bm{\ell}),
        \\
        -1 \, , \quad &\text{for} \, \, 
        (\bm{a},\bm{b},\bm{c})=(\bm{\ell},\bm{k},\bm{m}),(\bm{\ell},\bm{m},\bm{k}),(\bm{m},\bm{k},\bm{\ell})
        .
    \end{cases}
\end{align}
Note that $\tilde{\varepsilon}_{\bm{abc}}$ is different from Levi-Civita symbols.

\begin{theorem}
    In the $D$-dimensional long-range interacting XY model, under the condition $\alpha > D/2 +1$, $\alpha>D$ and $\beta < \beta_c$,
    $C^{(D)}_N (0)$ converges to the finite value independent of the system size $N$. 
    It immediately follows that $C^{(D)}_N (0)$ is finite as long as $\alpha >D$. 
\end{theorem}

\begin{proof}

In exactly the same manner as the one-dimensional case, utilizing Theorem \ref{cumulantpowerlawtheorem} and the symmetry derived from the Hamiltonian, the following upper bound for the current correlation can be obtained:
\begin{align}
    \abs{ \ave{ {\cal J}^{(i)}_{ n_i \mid n_i -1 } {\cal J}^{(i)}_{ 0 \mid -1 } } } 
    \leq \mathcal{O}(1) \sum_{\substack{ \bm{k}: k_i \geq n_i \\ \bm{\ell}': \ell'_i < 0}} \frac{1}{ d_{k,\ell'}^{2\alpha} }
\end{align}

Thus, in exactly the same manner as Theorem \ref{thm_d_ising}, we establish the claim of the theorem.
\end{proof}

\subsubsection{XYZ model}
We consider the one-dimensional long-range interacting XYZ model, which is described by the following Hamiltonian
\begin{equation}
    H= - \sum_{ \bm{i} \neq \bm{j} } \sum_{\sigma=x,y,z}  J_{\sigma} \frac{S^\sigma_{\bm{i}} S^\sigma_{\bm{j}}}{d_{\bm{i},\bm{j}}^\alpha}
    .
\end{equation}
The local energy at site $\bm{n}$ is defined as 
\begin{equation}
    \hat{h}_{\bm{n}} := - \sum_{ \bm{j} \neq \bm{n} } \sum_{\sigma =x,y,z } \frac{J_\sigma}{2} \frac{S^\sigma_{\bm{n}} S^\sigma_{\bm{j}}}{d_{\bm{n},\bm{j}}^\alpha} 
    ,
\end{equation}
and the current ${\cal J}^{(i)}_{ n_i \mid n_i -1 }$ is described as follows
\begin{align}
    {\cal J}^{(i)}_{ n_i \mid n_i -1 }
    &= - \sum_{\substack{\bm{k} : k_i \geq n_i \\ \bm{\ell} :  \ell_i < n_i }} t_{\bm{k} \leftarrow \bm{\ell} }
    \\
    t_{\bm{k} \leftarrow \bm{\ell}} 
    &= \{ \hat{h}_{\bm{k}} , \hat{h}_{\bm{\ell}} \} 
    = \sum_{\bm{m}} \sum_{ \sigma \in P^{xyz}_{\mathrm{cyc}} } \frac{J_{\sigma_1} J_{\sigma_2}}{4} \sum_{ \{\bm{a},\bm{b},\bm{c}\}=\{\bm{k},\bm{\ell},\bm{m}\} } \tilde{\varepsilon}_{\bm{abc}} \frac{ S^{\sigma_1}_{\bm{a}} S^{\sigma_2}_{\bm{b}} S^{\sigma_3}_{\bm{c}} }{d_{\bm{a},\bm{c}}^\alpha d_{\bm{b},\bm{c}}^\alpha}
    ,
\end{align}
where $\tilde{\varepsilon}_{\bm{abc}}$ is defined as follows:
\begin{align}
    \tilde{\varepsilon}_{\bm{abc}}=
    \begin{cases}
        +1 \, , \quad &\text{for} \, \,
        (\bm{a},\bm{b},\bm{c})=(\bm{k},\bm{\ell},\bm{m}),(\bm{m},\bm{\ell},\bm{k}),(\bm{k},\bm{m},\bm{\ell}),
        \\
        -1 \, , \quad &\text{for} \, \, 
        (\bm{a},\bm{b},\bm{c})=(\bm{\ell},\bm{k},\bm{m}),(\bm{\ell},\bm{m},\bm{k}),(\bm{m},\bm{k},\bm{\ell})
        .
    \end{cases}
\end{align}
Note that $\tilde{\varepsilon}_{\bm{abc}}$ is different from Levi-Civita symbols.

\begin{theorem}
    In the $D$-dimensional long-range interacting XYZ model, under the condition $\alpha > D/2 +1$, $\alpha>D$ and $\beta < \beta_c$,
    $C^{(D)}_N (0)$ converges to the finite value independent of the system size $N$.
    It immediately follows that $C^{(D)}_N (0)$ is finite as long as $\alpha >D$. 
\end{theorem}

\begin{proof}

In exactly the same manner as the one-dimensional case, utilizing Theorem \ref{cumulantpowerlawtheorem} and the symmetry derived from the Hamiltonian, the contributions to the current correlations, except for the case (2) in Fig.\ref{fig_sup_xyz_2}, are bounded as follows:
\begin{align}
    \left. \abs{\ave{ {\cal J}^{(i)}_{ n_i \mid n_i -1 } {\cal J}^{(i)}_{ 0 \mid -1 } }} \right|_{\mathrm{except \, Fig.\ref{fig_sup_xyz_2}(2) }}
    \leq \mathcal{O}(1) \sum_{\substack{ \bm{k}: k_i \geq n_i \\ \bm{\ell}': \ell'_i < 0}} \frac{1}{ d_{k,\ell'}^{2\alpha} }
    \label{xyz_d_1}
\end{align}

Next, we consider the contribution of the case (2) in Fig.\ref{fig_sup_xyz_2}.
In this case, we construct closed loops for both $(\bm{k},\bm{\ell},\bm{m})$ and $(\bm{k}',\bm{\ell}',\bm{m}')$, and they are connected via a single edge.
The corresponding schematic diagram is shown in Fig.\ref{fig_sup_d_xyz_1}.

\begin{figure}[h]
\centering
\includegraphics[width=0.5\textwidth]{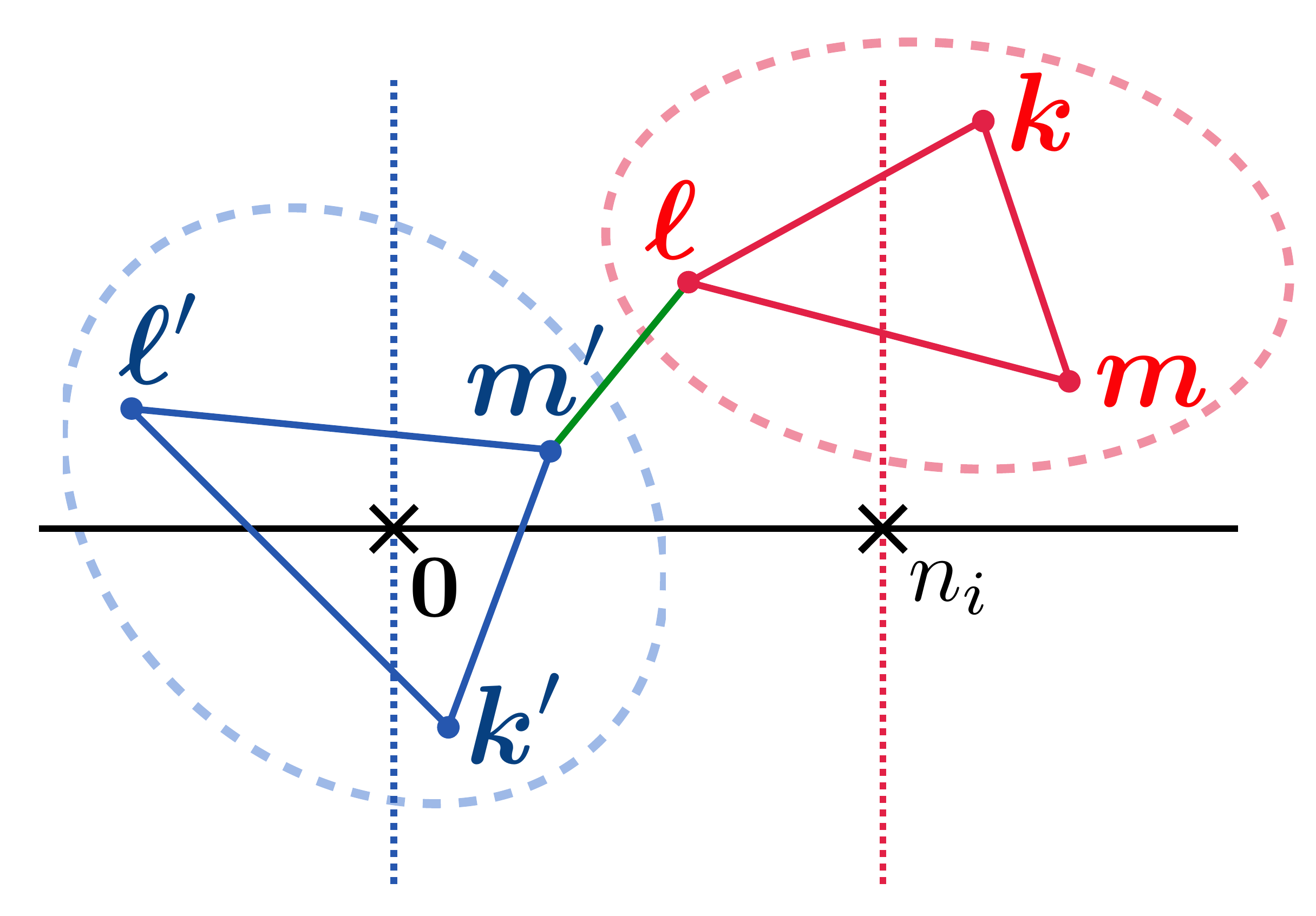}
\caption{The schematic diagram of the case (2) in Fig.\ref{fig_sup_xyz_2} in $D$-dimension }
\label{fig_sup_d_xyz_1}
\end{figure}

Let $\bm{p} \in (\bm{k},\bm{\ell},\bm{m}) , \bm{p}' \in (\bm{k}',\bm{\ell}',\bm{m}')$ be the two endpoints of the single edge connecting $(\bm{k},\bm{\ell},\bm{m})$ and $(\bm{k}',\bm{\ell}',\bm{m}')$. For the example in Fig.\ref{fig_sup_d_xyz_1}, $\bm{p} = \bm{\ell}$ and $\bm{p}' = \bm{m}'$. 

Next, we define $\bm{q} \in \{ \bm{k},\bm{\ell},\bm{m} \} \setminus \{ \bm{p} \}, \bm{q}' \in \{ \bm{k}',\bm{\ell}',\bm{m}' \} \setminus \{ \bm{p}' \} $
 as follows:
\begin{align}
    \bm{q} := 
    \begin{cases}
        \bm{k}  \, \quad &( \bm{p} \in \{ \bm{x} \mid x_i \geq n_i \} ) 
        \\
        \bm{\ell}  \, \quad &( \bm{p} \in \{ \bm{x} \mid x_i < n_i \} ) 
    \end{cases}
    \\
    \bm{q}' := 
    \begin{cases}
        \bm{\ell}'  \quad &( \bm{p} \in \{ \bm{x} \mid x_i \geq 0 \} ) 
        \\
        \bm{k}'  \quad &( \bm{p} \in \{ \bm{x} \mid x_i < 0 \} ) 
    \end{cases}
\end{align} 
By definition, $\bm{p}$ and $\bm{q}$ are on opposite sides of the hyperplane $x_i=n_i$, and $\bm{p}'$ and $\bm{q}'$ are on opposite sides of the hyperplane $x_i=0$. Thus, the condition $(p_i - n_i)(q_i-n_i)\leq 0$ and $p'_i q'_i \leq 0$ are satisfied. For the example in Fig.\ref{fig_sup_d_xyz_1}, $\bm{q} = \bm{k}$ and $\bm{q}' = \bm{\ell}'$. 

In the same manner as one-dimensional case, the contributions of the case (2) in Fig.\ref{fig_sup_xyz_2} are bounded as follows: 
\begin{align}
    \left. \abs{ \ave{ {\cal J}^{(i)}_{ n_i \mid n_i -1 } {\cal J}^{(i)}_{ 0 \mid -1 } } } \right|_{\mathrm{ Fig.\ref{fig_sup_xyz_2}(2) }}
    \leq \mathcal{O}(1) \sum_{ \bm{p},\bm{q},\bm{p}',\bm{q}' } \frac{1}{ d_{\bm{q},\bm{p}}^{2\alpha} } \frac{1}{ d_{\bm{p},\bm{p}'}^\alpha } \frac{1}{ d_{\bm{p}',\bm{q}'}^{2\alpha} }
\end{align}

Firstly, we evaluate the sum over $\bm{q},\bm{q}'$. For the sum over $\bm{q}$,  the following inequality is obtained: 
\begin{align}
    \sum_{\bm{q}: (p_i - n_i)(q_i-n_i)\leq 0 } \frac{1}{ d_{\bm{q},\bm{p}}^{2\alpha} } 
    &\leq \int_{(p_i - n_i)(q_i-n_i)\leq 0 } d^d \bm{q} \frac{1}{ \{ (q_i - p_i )^2 + \sum_{j \neq i} (q_j - p_j)^2 \}^{2\alpha/2} } 
    \notag \\ & 
    \leq \int_{(p_i - n_i)(q_i-n_i)\leq 0 } d q_i \Gamma_{D-1} \int^\infty_0 d \xi \xi^{D-2} \frac{1}{ \{ (q_i -p_i)^2 + \xi^2 \}^{2\alpha/2} }
    \notag \\ &
    \leq \Gamma' \int_{(p_i - n_i)(q_i-n_i)\leq 0 } d q_i \frac{1}{ (q_i -p_i )^{2\alpha-D+1} }
    \notag \\ &
    \leq \Gamma' \int^{N/2}_0 dx \frac{1}{ (x+n_i - p_i )^{2\alpha-D+1} }
    \leq \Gamma'' \frac{1}{ (n_i - p_i)^{2\alpha-D} } = \Gamma'' \frac{1}{ d_{\bar{\bm{n}}_{\bm{p}},\bm{p}}^{2\alpha-D} }
    , 
\end{align}
where $\Gamma' = \Gamma_{D-1} B(\frac{D-1}{2}, \alpha- \frac{D-1}{2})/2$, $\Gamma''=\Gamma'/(2\alpha-D)$ and we denote $\bar{\bm{n}}_{\bm{p}}$ as the foot of the perpendicular dropped from $\bm{p}$ to the hyperplane $x_i =n_i$, as in the leftmost part in Fig.\ref{fig_sup_d_xyz_2}.  

Similarily, for the sum over $\bm{q}'$,  the following inequality is obtained: 
\begin{align}
    \sum_{\bm{q}': p'_i q'_i \leq 0 } \frac{1}{ d_{\bm{p}',\bm{q}'}^{2\alpha} } 
    \leq \Gamma'' \frac{1}{d_{\bm{p}',\bar{\bm{0}}_{\bm{p}'}}^{2\alpha-D} }
\end{align}
where we denote $\bar{\bm{0}}_{\bm{p}'}$ as the foot of the perpendicular dropped from $\bm{p}'$ to the hyperplane $x_i =0$, as in the leftmost part in Fig.\ref{fig_sup_d_xyz_2}.

Thus, we obtain the following upper bound:
\begin{align}
    \left. \abs{ \ave{ {\cal J}^{(i)}_{ n_i \mid n_i -1 } {\cal J}^{(i)}_{ 0 \mid -1 } } } \right|_{\mathrm{ Fig.\ref{fig_sup_xyz_2}(2) }}
    \leq \sum_{ \bm{p},\bm{p}' } \frac{1}{ d_{ \bar{\bm{n}}_{\bm{p}},\bm{p} }^{2\alpha-D} } \frac{1}{ d_{\bm{p},\bm{p}'}^\alpha } \frac{1}{d_{\bm{p}',\bar{\bm{0}}_{\bm{p}'}}^{2\alpha-D} }
    \leq \sum_{ \bm{p},\bm{p}' } \frac{1}{ d_{ \bar{\bm{n}}_{\bm{p}},\bm{p} }^{\alpha} } \frac{1}{ d_{\bm{p},\bm{p}'}^\alpha } \frac{1}{d_{\bm{p}',\bar{\bm{0}}_{\bm{p}'}}^{\alpha} }
    ,
\end{align}
where in advance, we assume $\alpha>D$. 

Next, to further evaluate this upper bound, we present the following lemma.

\begin{lemma}
\label{chain_2}
If $\bm{i},\bm{j}$ depend on $\bm{k}$ and $\tilde{u}_i := \sum_{k} 1/d_{\bm{k},\bm{i}(\bm{k})}^\alpha < \infty , \tilde{u}_j := \sum_{k} 1/d_{\bm{k},\bm{j}(\bm{k})}^\alpha < \infty$ hold, under the condition $\alpha >D$, the following inequality holds:
\begin{align}
    \sum_{\bm{k}} \frac{1}{ d_{\bm{i}(\bm{k}),\bm{k}}^\alpha } \frac{1}{ d_{\bm{j}(\bm{k}),\bm{k}}^\alpha } \leq \frac{ \tilde{u} }{d_{\bar{\bm{i}},\bar{\bm{j}}}^\alpha} ,
\end{align}
where $d_{\bar{\bm{i}},\bar{\bm{j}}}^{-\alpha} := \displaystyle\sup_{\bm{k}} d_{\bm{i}(\bm{k}),\bm{j}(\bm{k})}^{-\alpha}$, and $\tilde{u} := 2^{\alpha-1} (\tilde{u}_i + \tilde{u}_j) $.
\end{lemma}

\begin{proof}[Proof of Lemma \ref{chain_2}]
\begin{align}
    &\sum_{\bm{k}} \frac{1}{d_{\bm{i}(\bm{k}),\bm{k}}^\alpha} \frac{1}{d_{\bm{k},\bm{j}(\bm{k})}^\alpha}
    = \sum_{\bm{k}} \frac{1}{d_{\bm{i}(\bm{k}),\bm{j}(\bm{k})}^\alpha} \frac{ d_{\bm{i}(\bm{k}),\bm{j}(\bm{k})}^\alpha }{ d_{\bm{i}(\bm{k}),\bm{k}}^\alpha d_{\bm{k},\bm{j}(\bm{k})}^\alpha }
    \leq \frac{1}{d_{\bar{\bm{i}},\bar{\bm{j}}}^\alpha} \sum_{\bm{k}} \frac{ (d_{\bm{i}(\bm{k}),\bm{k}} + d_{\bm{k},\bm{j}(\bm{k})})^\alpha }{ d_{\bm{i}(\bm{k}),\bm{k}}^\alpha d_{\bm{k},\bm{j}(\bm{k})}^\alpha }
    \notag \\ &
    \leq 2^{\alpha-1} \frac{1}{d_{\bar{\bm{i}},\bar{\bm{j}}}^\alpha} \sum_{\bm{k}} \frac{ d_{\bm{i}(\bm{k}),\bm{k}}^\alpha + d_{\bm{k},\bm{j}(\bm{k})}^\alpha }{ d_{\bm{i}(\bm{k}),\bm{k}}^\alpha d_{\bm{k},\bm{j}(\bm{k})}^\alpha }
    = 2^{\alpha-1} \frac{1}{d_{\bar{\bm{i}},\bar{\bm{j}}}^\alpha} \sum_{\bm{k}} \Big( \frac{1}{d_{\bm{i}(\bm{k}),\bm{k}}^\alpha} + \frac{1}{d_{\bm{k},\bm{j}(\bm{k})}^\alpha} \Big)
    = \frac{ \tilde{u} }{d_{\bar{\bm{i}},\bar{\bm{j}}}^\alpha} 
    \end{align}
Here, we use $\{(x+y)/2\}^\alpha \leq (x^\alpha + y^\alpha )/2$
, following from the convexity of $f(x)=x^\alpha$ and the convexity inequality.
\end{proof}

\begin{figure}[h]
\centering
\includegraphics[width=1.0\textwidth]{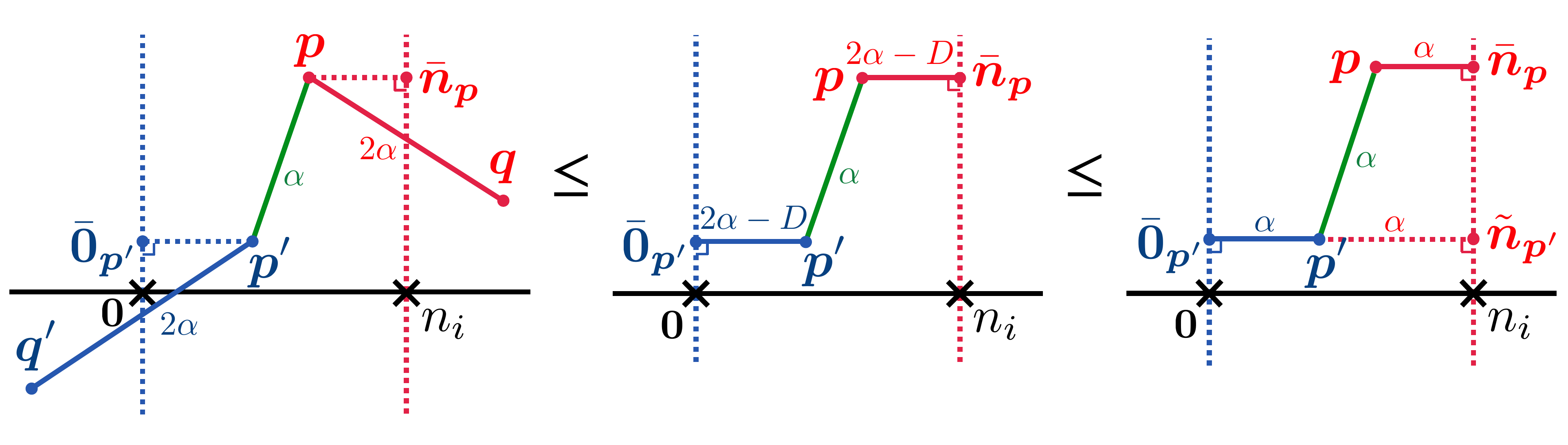}
\caption{The diagrammatic derivation for the upper bound of the case (2) in Fig.\ref{fig_sup_xyz_2} in $D$-dimension }
\label{fig_sup_d_xyz_2}
\end{figure}

Finally, utilizing Lemma.\ref{chain_2}, we obtain the following upper bound:
\begin{align}
    \left. \abs{ \ave{ {\cal J}^{(i)}_{ n_i \mid n_i -1 } {\cal J}^{(i)}_{ 0 \mid -1 } } } \right|_{\mathrm{ Fig.\ref{fig_sup_xyz_2}(2) }}
    &\leq \sum_{ \bm{p},\bm{p}' } \frac{1}{ d_{ \bar{\bm{n}}_{\bm{p}},\bm{p} }^{\alpha} } \frac{1}{ d_{\bm{p},\bm{p}'}^\alpha } \frac{1}{d_{\bm{p}',\bar{\bm{0}}_{\bm{p}'}}^{\alpha} }
    \notag \\ &
    \leq \sum_{\bm{p}'} \frac{1}{d_{\bm{p}',\bar{\bm{0}}_{\bm{p}'}}^{\alpha} } 2^{\alpha-1} ( u + \tilde{u}_{\bar{\bm{n}}_{\bm{p}}}) \sup_{\bm{p}} \frac{1}{ d_{ \bar{\bm{n}}_{\bm{p}}, \bm{p}' }^{\alpha} } 
    = 2^{\alpha-1} ( u + \tilde{u}_{\bar{\bm{n}}_{\bm{p}}}) \sum_{\bm{p}'} \frac{1}{ d_{ \tilde{\bm{n}}_{\bm{p}'}, \bm{p}' }^{\alpha} } \frac{1}{d_{\bm{p}',\bar{\bm{0}}_{\bm{p}'}}^{\alpha} }
    \notag \\ &
    \leq 2^{\alpha-1} ( u + \tilde{u}_{\bar{\bm{n}}_{\bm{p}}})  2^{\alpha-1} ( \tilde{u}_{\bar{\bm{0}}_{\bm{p}'}} + \tilde{u}_{\tilde{\bm{n}}_{\bm{p}'}} ) \sup_{\bm{p}'} \frac{1}{ d_{ \tilde{\bm{n}}_{\bm{p}'}, \bar{\bm{0}}_{\bm{p}'} }^\alpha }
    = \mathcal{O}(1) \frac{1}{n_i^\alpha}
    \label{xyz_d_2},
\end{align}
where $\tilde{u}_{\bar{\bm{n}}_{\bm{p}}} = \sum_{ \bm{p} } d_{ \bar{\bm{n}}_{\bm{p}},\bm{p} }^{-\alpha} \leq \int^{N/2}_1 d^d \bm{x} x_i^{-\alpha} \leq \mathcal{O}(1) N^{D-\alpha} < \infty$ for $\alpha>D$ . Furthermore, let $\tilde{\bm{n}}_{\bm{p}'}$ be the $\bar{\bm{n}}_{\bm{p}}$ that satisfies 
$\sup_{\bm{p}} d_{ \bar{\bm{n}}_{\bm{p}, \bm{p}'} }^{-\alpha} = d_{ \tilde{\bm{n}}_{\bm{p}}, \bm{p}' }^{-\alpha} $. 
As shown on the far right of Fig.\ref{fig_sup_d_xyz_2}, $\bar{\bm{n}}_{\bm{p}}$ coincides with the foot of the perpendicular dropped from $\bm{p}'$ to the hyperplane $x_i=n_i$.
The condition of Lemma \ref{chain_2} is satisfied, as $\tilde{u}_{\bar{\bm{0}}_{\bm{p}'}}= \sum_{\bm{p}'} d_{\bar{\bm{0}}_{\bm{p}'},\bm{p}'}^{\alpha} \leq \int^{N/2}_1 d^d \bm{x} x_i^{-\alpha} <\infty$ and $\tilde{u}_{\tilde{\bm{n}}_{\bm{p}'}}= \sum_{\bm{p}'} d_{\tilde{\bm{n}}_{\bm{p}'},\bm{p}'}^{\alpha} \leq \int^{N/2}_1 d^d \bm{x} x_i^{-\alpha} <\infty$. 

From \eqref{xyz_d_1} and \eqref{xyz_d_2}, we obtain the following upper bound for $C^{(D)}_N (0)$: 
\begin{align}
    \abs{ C^{(D)}_N (0) }
    &\leq \mathcal{O}(1) N^{D+2-2\alpha} + \mathcal{O}(1) N^{2-D-\alpha} + {\rm const.}
    \notag \\
    &\leq \mathcal{O}(1) N^{D+2-2\alpha} + {\rm const.},
\end{align}
for $D \geq 2$. 
Finally, we show that $C^{(D)}_N (0)$ converges to a finite value under the condition $\alpha > D/2 + 1$.
\end{proof}

\section{Fluctuating hydrodynamics and L\'{e}vy diffusion}

\subsection{Derivation of fluctuating hydrodynamics and L\'{e}vy diffusion} 

In this subsection, we construct the microscopic derivation of fluctuating hydrodynamics for long-range interacting spin systems, and show that the space-time energy correlation function exhibits the L\'{e}vy diffusion scaling.

Fluctuating hydrodynamics provides a theoretical framework for describing the evolution of conserved quantities at the mesoscopic scale, where stochastic fluctuations originating from microscopic degrees of freedom become significant. In systems where mass, momentum, and energy are conserved, the emergent macroscopic dynamics of these quantities are encapsulated by the celebrated Navier-Stokes equations. These equations, derived under the assumptions of continuum mechanics, offer a deterministic description of the system's behavior at large scales. However, at the mesoscopic scale, the deterministic formulation of the Navier-Stokes equations is insufficient, necessitating the inclusion of noise terms. These stochastic contributions arise from the microscopic degrees of freedom beyond the conserved quantities, ensuring a consistent description of the system's behavior across scales.

The derivation of such mesoscopic equations typically relies on Zwanzig's projection operator formalism \cite{zwanzig,zubarevmorozof}, a powerful method for systematically separating the dynamics of conserved quantities from the other microscopic variables. After using this method, we consider the corresponding Langevin dynamics introducing stochastic terms that encapsulate the effects of the microscopic fluctuations, thereby bridging the microscopic and mesoscopic descriptions.

In this work, we examine the effective dynamics of a spin system characterized by a single conservation law—energy conservation. To employ the projection operator formalism to capture this system's mesoscopic behavior, we introduce a coarse-grained coordinate $x$, associated with a characteristic length scale $\ell$, to define the spatial resolution of the coarse-grained dynamics:
\begin{align}
\hat{\epsilon}_x &= (1/\ell)\sum_{n=(x-1)\ell +1}^{x \ell} \hat{h}_n \, , 
\end{align}
The symbol $\hat{a}$ indicates that the variable $a$ depends on the phase space $\Gamma := (S_1^x, S_1^y, S_1^z, \cdots, S_N^x, S_N^y, S_N^z)$. The continuity equation for energy is expressed as $\partial_t \hat{\epsilon}_x = -\partial_x \hat{\mathbb{J}}_x$, where the derivative $\partial_x$ is defined as $\partial_x \hat{a}_x := (1/\ell)(\hat{a}_{x+1} - \hat{a}_x)$. To formulate the dynamics, we introduce the projection operator 
\begin{align} ({\cal P} \hat{a})[\Gamma]&:= \Omega^{-1} (\Gamma) \int d\Gamma' \hat{a} (\Gamma') \prod_{x} \delta (\hat{\epsilon}_x(\Gamma') - \hat{\epsilon}_x(\Gamma)) \, , \\
 \Omega (\Gamma) &= \int d\Gamma' \prod_x \delta (\hat{\epsilon}_x (\Gamma') - \hat{\epsilon}_x (\Gamma)) . 
\end{align}
Let $f(\epsilon, t)$ denote the distribution of energy ${\epsilon_x}$ at time $t$, defined as 
\begin{align} f(\epsilon, t) &:= \int d\Gamma \hat{\rho} (\Gamma, t) \prod_x \delta (\hat{\epsilon}_x (\Gamma) - \epsilon_x) \, , 
\end{align} 
where $\hat{\rho} (\Gamma, t)$ represents the density distribution in the phase space $\Gamma$. The density distribution evolves according to the Liouville equation under the Hamiltonian $H$: $\partial_t \hat{\rho}(\Gamma, t) = \mathbb{L} \hat{\rho}(\Gamma, t) = \{ H, \hat{\rho}(\Gamma, t) \}$, with the Poisson bracket $\{A, B\} = \sum_i \sum_{a, b, c} \varepsilon_{abc} (\partial A / \partial S_i^a) (\partial B / \partial S_i^b) S_i^c$, where $\varepsilon_{abc}$ is the Levi-Civita symbol. The time derivative of the energy distribution is given by 
\begin{align} 
\partial_t f(\epsilon, t) &= \int d\Gamma \hat{\rho}(\Gamma, t) \sum_x \partial_x \hat{\mathbb{J}}_x (\Gamma) \frac{\delta}{\delta \epsilon_x} \prod_x \delta (\hat{\epsilon}_x (\Gamma) - \epsilon_x) \nonumber \\
&= \int d\Gamma \left[{\cal P} \hat{\rho}(\Gamma, t) + {\cal Q} \hat{\rho}(\Gamma, t) \right] \sum_x \partial_x \hat{\mathbb{J}}_x (\Gamma) \frac{\delta}{\delta \epsilon_x} \prod_x \delta (\hat{\epsilon}_x (\Gamma) - \epsilon_x)  \label{proj_base1} \\
&= \int d\Gamma {\cal Q} \hat{\rho}(\Gamma, t)  \sum_x \partial_x \hat{\mathbb{J}}_x (\Gamma) \frac{\delta}{\delta \epsilon_x} \prod_x \delta (\hat{\epsilon}_x (\Gamma) - \epsilon_x) \, , \label{proj_base} 
\end{align} 
where we have used that the contribution of ${\cal P}\hat{\rho}$ in \eqref{proj_base1} vanishes due to the time-reversal symmetry. Using this expression, we derive a closed equation in terms of $f(\epsilon, t)$. Writing ${\cal Q}\hat{\rho}$ in terms of ${\cal P}\hat{\rho}$, we arrive at the following form
\begin{align} 
\partial_t f(\epsilon, t) &= \sum_{x,x'} \frac{\delta}{\delta \epsilon_x} \int_{-\infty}^t ds \int {\cal D}\epsilon' \Omega (\epsilon) (\partial_x \partial_{x'} K(x,x'; t-s)) \frac{\delta}{\delta \epsilon'_{x'}} \left( \frac{f(\epsilon',s)}{\Omega(\epsilon')} \right) \, , \\ 
K(x,x'; t-s) &= \int d\Gamma \, \Omega^{-1} (\epsilon) \prod_z \delta (\hat{\epsilon}_z (\Gamma) - \epsilon_z) \left[ \hat{\mathbb{J}}_x (\Gamma) \left( e^{(t-s) {\cal Q}\mathbb{L}} \prod_y \delta(\hat{\epsilon}_y(\Gamma) - \epsilon_y') \mathbb{J}_{x'} (\Gamma) \right) \right] \, . 
\end{align}
We approximate $K(x,x'; t-s)$ as in \cite{saito2021microscopic} assuming that the term decays rapidly in time, while the conserved quantity, i.e., the energy field does not change rapidly:
\begin{align}
K(x,x'; t-s) &\sim \int d\Gamma \, \Omega^{-1} (\epsilon )\prod_{z}\delta (\hat{\epsilon}_z (\Gamma) - \epsilon_z) \bigl[  \hat{\mathbb{J}}_x(\Gamma) (e^{(t-s){\cal Q}\mathbb{L}} \mathbb{J}_{x'} (\Gamma ) ) \bigr]\prod_y \delta(\epsilon_{y} -\epsilon_y' ) \, .
\end{align}
In addition, the Markovian approximation leads to the following expression
\begin{align} 
\partial_t f(\epsilon, t) &= \sum_{x,x'} \frac{\delta}{\delta \epsilon_x} \Omega(\epsilon) \left[(\partial_x \partial_{x'} K(x,x') )\frac{\delta}{\delta \epsilon_x} \left( \frac{f(\epsilon , t)}{\Omega(\epsilon)} \right) \right] \, , \\ 
K(x,x') &= \int_0^{\infty} ds \int d\Gamma \,  \Omega^{-1} (\epsilon) \prod_z \delta (\hat{\epsilon}_z (\Gamma) - \epsilon_z) \left[ \hat{\mathbb{J}}_x (\Gamma) e^{s {\cal Q}\mathbb{L}} \mathbb{J}_{x'} (\Gamma) \right] \, . \label{fp-projection}
\end{align}
This dynamics corresponds to the Langevin equation: 
\begin{align} 
\partial_t \epsilon_x(t) &= -\partial_x \left[ \sum_{x'} K(x,x') \frac{\partial}{\partial x'} \frac{\delta \bar{S}}{\delta \epsilon_{x'}} + \xi_x(t) \right] \, , \\ 
\langle\!\langle \xi_x(t) \xi_{x'}(t') \rangle\!\rangle &= 2 K(x,x') \delta(t-t') \, , \\ 
\bar{S} &= \ln \Omega(\epsilon) \, , 
\end{align} 
where $\langle\!\langle ...\rangle\!\rangle$ stands for the noise average and $\bar{S}$ represents the system entropy divided by $k_{\rm B}$. We use the standard thermodynamic relation $d \bar{S}/ d E = \beta$. Then, we find the Langevin equation in terms of energy as follows
\begin{align} 
\partial_t \epsilon_x(t) &= \partial_x \left[ \sum_{x'} \frac{K(x,x')}{c_{\rm V} k_{\rm B} T^2} \frac{\partial}{\partial x'} \epsilon_{x'}(t) + \xi_x(t) \right] \, . 
\end{align}
In the continuum limit, defining $D(x,x') = K(x,x') / (c_{\rm V} k_{\rm B} T^2)$, we arrive at the fluctuating hydrodynamics: 
\begin{align} 
\partial_t \epsilon(x,t) &= \partial_x \left[ \int dx' D(x,x') \partial_{x'} \epsilon(x',t) + \xi(x,t) \right] \, , \\ 
\langle\!\langle \xi(x,t) \xi(x',t') \rangle \!\rangle &= 2 c_{\rm V} k_{\rm B} T^2 D(x,x') \delta(t-t') \, . 
\label{fht_sup}
\end{align}
The long-range diffusion coefficient $D(x,x')$ is effectively given by the Green-Kubo like relation: 
\begin{align} 
D(x,x') &= \frac{1}{c_{\rm V} k_{\rm B} T^2} \int_0^{\infty} \langle \mathbb{J}(x,t) \mathbb{J}(x',0) \rangle dt \, , 
\label{GK-like}
\end{align} 
where $\mathbb{J}(x,t)$ is the local energy current, determined by the microscopic Hamiltonian dynamics. 
For short-range interactions, $D(x,x')$ generically shows rapid decay in space, i.e., $D(x,x')\sim{\cal D}_0 \delta (x-x')$, and hence, the dynamics reduces to the standard fluctuating hydrodynamics $\partial_t \epsilon(x,t) = \partial_x \left[ {\cal D}_0 \partial_x \epsilon(x,t) + \xi(x,t) \right]$. 

At the stage of (\ref{fp-projection}), the time evolution of the current in the Green-Kubo-like relation is governed by the projected operator, i.e., $e^{t {\cal Q}\mathbb{L}}$. However, we approximate it by replacing with the standard time evolution, i.e., $e^{t \mathbb{L}}$ for the following two reasons. First, ${\cal P}\mathbb{L}$ physically implies extracting the hydrodynamic effect in the dynamics. However, in the present case, the equilibrium current is absent due to the absence of continuous translational symmetry, and hence hydrodynamic effect is not significant. Thus the rapid decay in time is expected to occur. During the initial short-time regime, the effect of the projection operator is shown to be negligible \cite{saito2021microscopic}. Second, in the short-range interacting systems, the bare transport coefficient ${\cal D}_0$ aligns with the macroscopic coefficient derived using the standard Green-Kubo formula which uses the standard time evolution (See e.g., \cite{dhar2019fourier}). These two provide a physical rationale for the replacement.

Now, let us consider the space-time correlation function $C(x,t)$, defined as 
\begin{align} 
C(x,t) &= \langle\!\langle \delta \epsilon(x,t) \delta \epsilon(0,0) \rangle\!\rangle \, , 
\label{space-time_fn}
\end{align} 
where $\delta \epsilon (x,t) = \epsilon (x,t)- \bar{\epsilon}$, where $\bar{\epsilon}$ is the long-time average of the local energy. We have the equation for $C(x,t)$ 
\begin{align} 
\partial_t C(x,t) &= \partial_x \int dx' D(x,x') \partial_{x'} C(x',t) \, . 
\label{space-time_eq}
\end{align}
In the Fourier representation defining $C(x,t):={1\over 2\pi}\int d k C(k,t) e^{i k x}$ and $D(x,x')= {1\over 2\pi}\int d k D(k) e^{i k (x-x')}$, we obtain the expression 
\begin{align}
{\partial \over \partial t} C(k,t) &= -k^2 D(k) C(k,t) \, . 
\end{align}
Suppose that $D(k)\sim k^{-\nu}$ for the small wave number regime, the space-time correlation function should have the following asymptotic form
\begin{align}
C(x,t) &\sim 
t^{-1/(2-\nu)} f\Bigl( {x \over t^{1/(2-\nu)}} \Bigr) \, , ~~~\nu \ge 0 \, .
\end{align}
When $D(x,x')$ exponentially decays, we have $\nu =0$. In this case, the correlation shows the normal diffusion scaling. On the other hand, when $D(x,x')$ shows algebraically decay with a finite $\nu >0$, the L\'{e}vy diffusion appears.

\subsection{Validity of fluctuating hydrodynamics}

To ensure the validity of fluctuating hydrodynamics, we compare the space-time correlations between fluctuating hydrodynamics and direct numerical computation.

In the inset of Fig.2 in the main text, we demonstrate the convergence of the Green-Kubo-like expression \eqref{GK-like} for the nonlocal diffusion coefficient. Using the data, we compute the space-time equilibrium correlation function 
$C(x,t)= \langle\!\langle \delta \epsilon(x,t) \delta \epsilon(0,0) \rangle\!\rangle$ 
in Eq.~(\ref{space-time_fn}), by numerically solving the integro-differential equation \eqref{space-time_eq}
which follows from the fluctuating hydrodynamics framework Eq.~(\ref{fht_sup}). This allows for a direct comparison between the correlation function $C(x,t)$ and $C(n,t)=\langle \delta \hat{h}_{n} (t) \delta \hat{h}_0 (0)\rangle$ computed directly from the microscopic dynamics governed by the Hamiltonian Eq.~(1) in the main text. 

In Fig.\ref{fhd_validity}, we present this comparison for the parameter $\alpha = 1.3$ in the transverse Ising and XYZ models. As clearly seen in the figure, the agreement between the two results is remarkably good. This quantitative concordance provides strong evidence that the fluctuating hydrodynamics with a nonlocal diffusion kernel serves as an accurate and predictive effective description of the system's dynamics.


\begin{figure*}[h]
  \begin{center}
    \begin{tabular}{c}
      \begin{minipage}{0.5\hsize}
        \begin{center}
                \includegraphics[width= \textwidth]{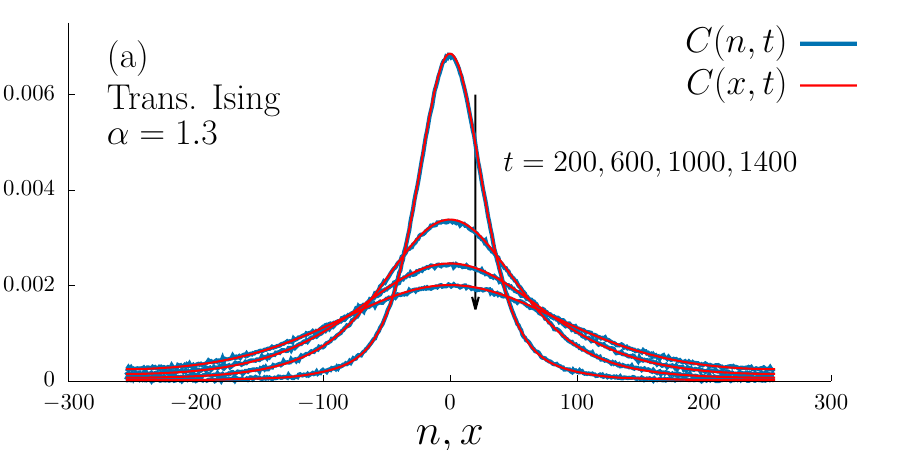}
        \end{center}
      \end{minipage}

      \begin{minipage}{0.5\hsize}
        \begin{center}
                \includegraphics[width= \textwidth]{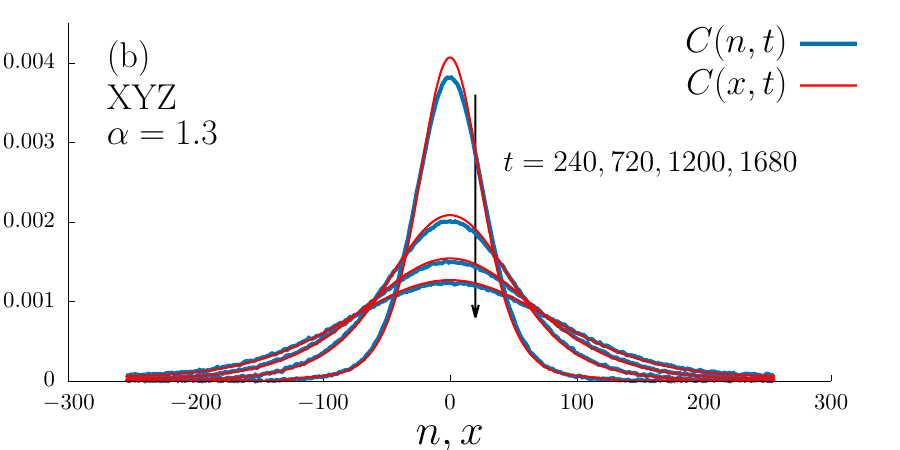}
        \end{center}
      \end{minipage}
    \end{tabular}
\caption{ Comparison of space-time correlations between $C(n,t)$, obtained directly from the microscopic Hamiltonian, and $C(x,t)$, computed via fluctuating hydrodynamics using the input data $D(n-n')$ computed from a Green-Kubo-like formula. (a) Transverse Ising case with $\alpha = 1.3$, $J = 1$, $h = -0.5$, and $N = 512$; (b) XYZ case with $\alpha = 1.3$, $J_x = 1$, $J_y = 0.8$, $J_z = 0.5$, and $N = 512$. } \label{fhd_validity} 
  \end{center}
\end{figure*}

\section{Supplementary numerical data}
In this section, we present supplementary numerical results for the current correlations and the space-time energy correlations. All data here are given for the temperature $T=10$ with $k_{\rm B}\equiv 1$ and the sample size $10^6$.
\\

As supplementary figures for the Fig.2 in the main text, we present Figs.\ref{fig2_sup} and \ref{fig2_sup2}. 
In Fig.\ref{fig2_sup}, the equal-time current correlation $\ave{ \mathcal{J}_n \mathcal{J}_0 }$ and Green-Kubo-like expression $\int^{\infty}_0 dt \ave{ \mathcal{J}_n(t) \mathcal{J}_0 (0)}$ for the transverse Ising model are shown. Fig.\ref{fig2_sup}(a) corresponds to the parameter $\alpha=1.1$ and (b) is for $\alpha=1.6$. In the main panel, we show the current correlation $\ave{ \mathcal{J}_n \mathcal{J}_0 }$, while in the inset we show the Green-Kubo-like integral. We checked that for each parameter, the correlation $\ave{ \mathcal{J}_n(t) \mathcal{J}_0 (0) }$ decays rapidly in time and hence the integration is finite. The data is displayed as a log-log plot, where the power-law exponents $\eta$ agree with the theoretical bounds \eqref{upp-main_xy} and \eqref{upp-main_xyz} in i.e., $\eta=2\alpha-2$. These results demonstrate that the upper bounds in Theorem \ref{main_ising} is optimal in various parameters.

In Fig.\ref{fig2_sup2}, we show the cases of the XY and XYZ model. Fig.\ref{fig2_sup2}(a) corresponds to the XY model with parameter $\alpha=1.3$ and (b) is for the XYZ model with the same parameter. In the main panel, we present the current correlation $\ave{ \mathcal{J}_n \mathcal{J}_0 }$, while in the inset we show the Green-Kubo-like integral. Likewise, the power-law exponents $\eta$ agree with the theoretical bounds \eqref{upp-main_xy} and \eqref{upp-main_xyz} in i.e., $\eta=2\alpha-2$. These figures reinforce the optimality of the theoretical upper bound; hence, fluctuating hydrodynamics can explain the energy diffusion.

\begin{figure}[!ht]
\centering
\includegraphics[width=1.0\textwidth]{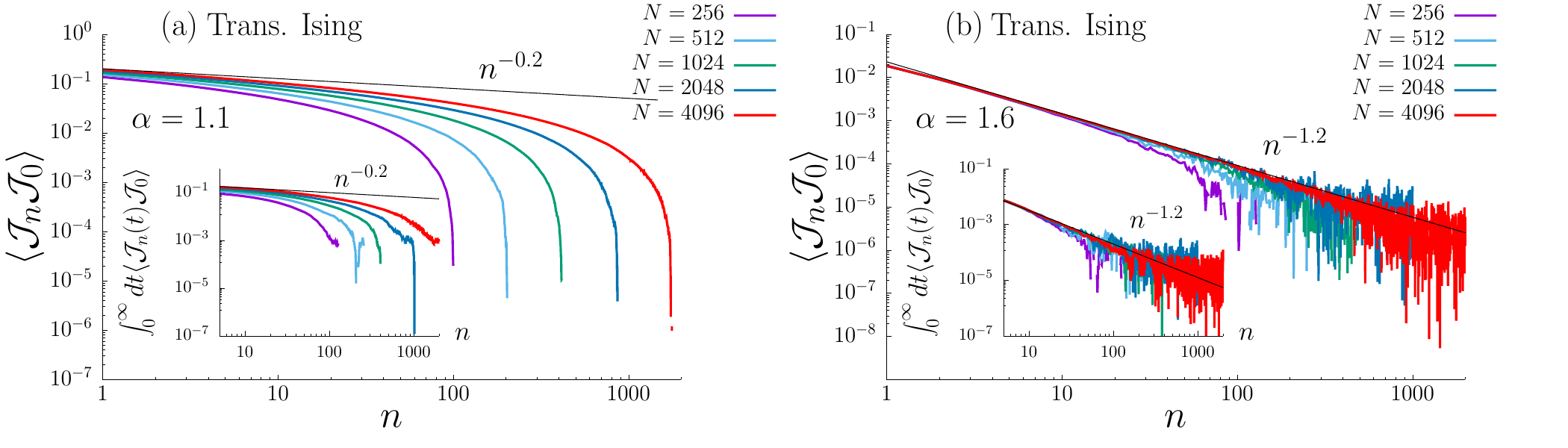}
\caption{The correlation function $\langle {\cal J}_n {\cal J}_0\rangle$ for the transverse Ising model in the main panel, with the parameter $\alpha=1.1$ in (a) and $\alpha=1.6$ in (b). Other parameters are $J=h=1$. The inset shows the Green-Kubo like formula. The solid line is the theoretical line, i.e., $n^{2-2\alpha}$. The power-law behavior can be explained with the equal-time correlation $\langle {\cal J}_n {\cal J}_0\rangle$ and the power-law exponent is equal to $\eta=2\alpha-2$.}
\label{fig2_sup}
\end{figure}

\begin{figure}[!ht]
\centering
\includegraphics[width=1.0\textwidth]{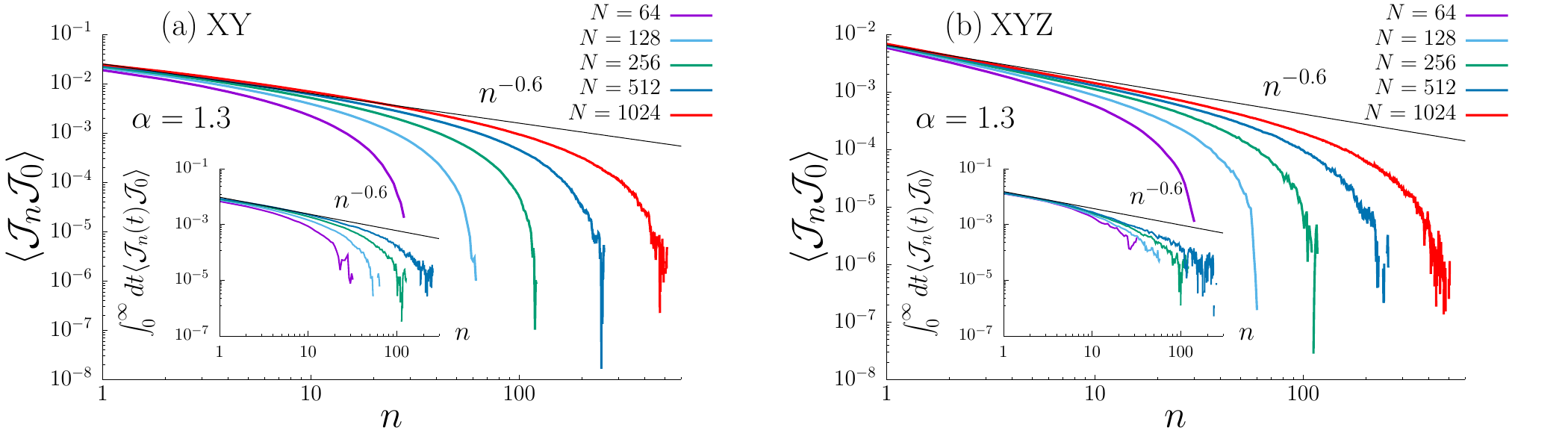}
\caption{The correlation function $\langle {\cal J}_n {\cal J}_0\rangle$ for the XY model in (a) and the XYZ model in (b), with the parameter $\alpha=1.3$ in the main panel. Other parameters in the XY model are $J_x=1, J_y=-0.8$, and in the XYZ model, $J_x=1, J_y=0.8, J_z=0.5$. The inset shows the Green-Kubo like formula. The solid line is the theoretical line, i.e., $n^{2-2\alpha}$. The power-law behavior can be explained with the equal-time correlation $\langle {\cal J}_n {\cal J}_0\rangle$ and the power law exponent is equal to $\eta=2\alpha-2$.}
\label{fig2_sup2}
\end{figure}

As supplementary figures for the Fig.$3$ in the main text, we present Figs.\ref{fig3_sup1} and \ref{fig3_sup2}. In Fig.\ref{fig3_sup1}, we present the space-time energy correlation $C(n,t)= \langle \delta \hat{h}_n(t) \delta \hat{h}_0 (0) \rangle$ for the transverse Ising and XYZ models with the parameter $\alpha=1.1$. Fig.\ref{fig3_sup1}(a) corresponds to the transverse Ising model, while (b) is for the XYZ model. The results obviously show that $\alpha=1.1$ case obeys the L\'{e}vy scaling, reaffirming the scaling exponent $\nu =3-2\alpha$, which is predicted by the fluctuating hydrodynamics.   
Fig.\ref{fig3_sup2} shows the case of the XY model with the parameters $\alpha=1.1, \, 1.3, \, 1.6$. The results clearly show that $\alpha=1.1$ and $\alpha=1.3$ cases obey the L\'{e}vy scaling, while $\alpha=1.6$ case obeys the diffusive scaling. These results in the XY model are consistent with the scaling predicted by the fluctuating hydrodynamics.

\begin{figure*}[!ht]
 	\centering
\includegraphics[width=17.5cm]{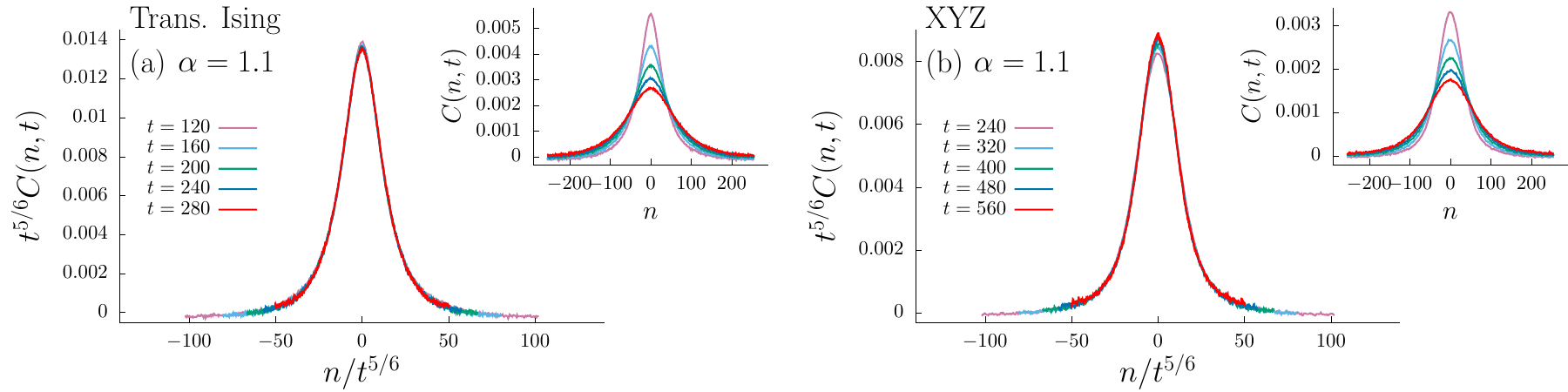}
\caption{Space-time energy correlation function $C(n,t):=\langle \delta \hat{h}_{n} (t) \delta \hat{h}_0(0)\rangle$ for the transverse Ising model with the parameters $\alpha=1.1$ and $J=1$, $h=1$ in (a), and the XYZ model with the parameters $\alpha=1.1$ and $J_x=1,J_y=0.8,J_z=0.5$ in (b). The system size is $N=512$. It is clear that both models obey the L\'{e}vy scaling. }
\label{fig3_sup1}
\end{figure*}

\begin{figure*}[ht]
 	\centering
\includegraphics[width=17.5cm]{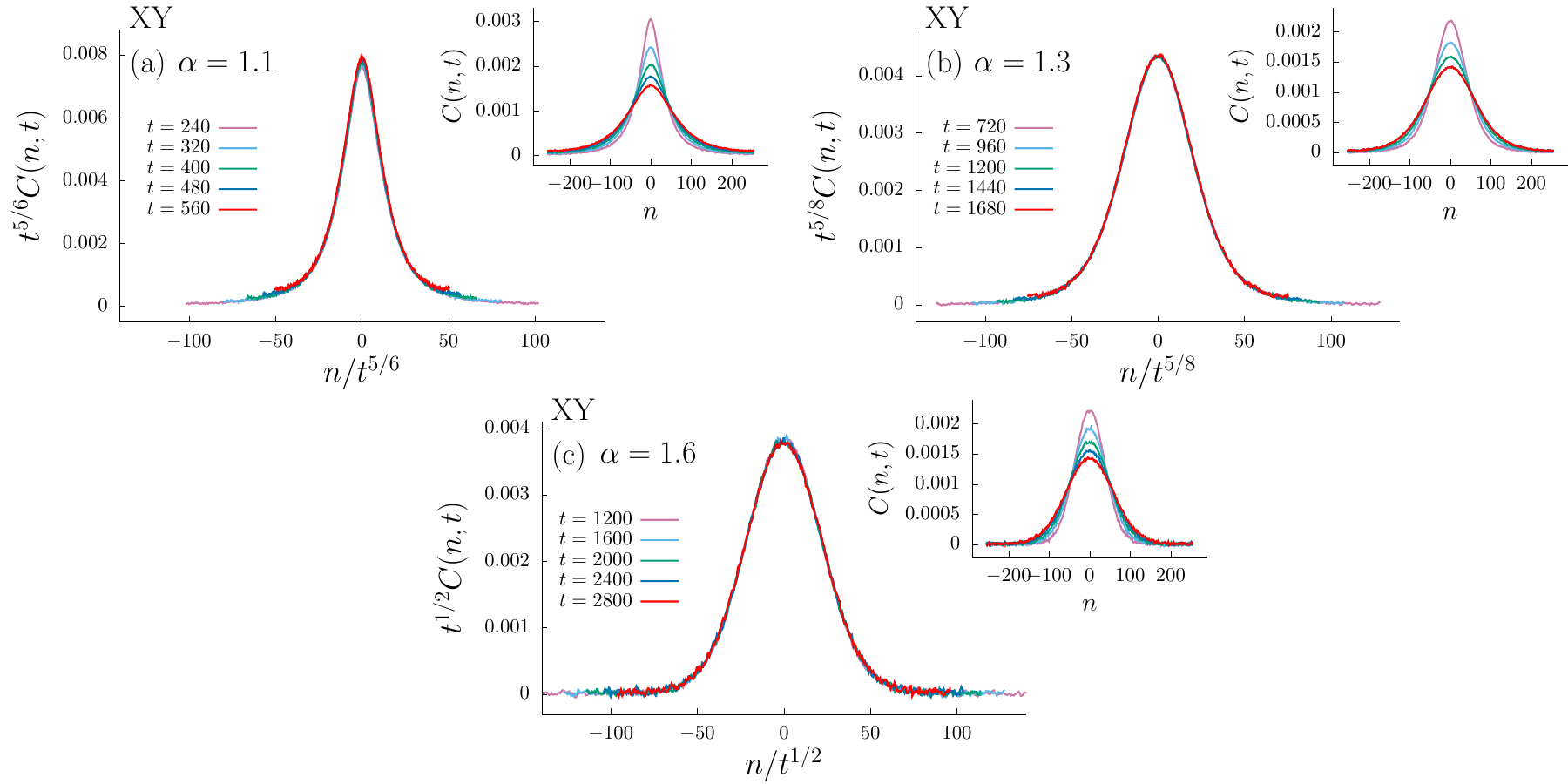}
\caption{Space-time energy correlation function $C(n,t):=\langle \delta \hat{h}_{n} (t) \delta \hat{h}_0(0) \rangle$ for the XY model with the parameters $J_x=1,J_y=-0.8$. We use the exponents $\alpha=1.1$ in (a), $\alpha=1.3$ in (b), and $\alpha=1.6$ in (c). The system size is $N=512$. It is clear that $\alpha=1.1$ and $\alpha=1.3$ cases obey the L\'{e}vy scaling, while $\alpha=1.6$ case shows the diffusive scaling. }
\label{fig3_sup2}
\end{figure*}

\newpage

%
%

\end{widetext}

\vspace*{10cm}


\end{document}